\documentclass[aps,prx,reprint,superscriptaddress,nofootinbib]{revtex4-2}

\usepackage{amsmath,amssymb}
\usepackage{graphicx}
\usepackage{physics}
\usepackage{color}
\usepackage{amsmath,amssymb,mathtools}
\usepackage{tikz}
\usepackage{booktabs}
\usepackage{amsthm}
\usepackage{ragged2e}
\theoremstyle{plain}

\newtheorem{theorem}{Theorem}

\newtheorem{Proposition}[theorem]{Proposition}

\newtheorem*{theorem*}{Theorem}

\newtheoremstyle{axiomwithbreak}
  {3pt}
  {3pt}
  {}
  {}
  {\bfseries}
  {.}
  {\newline}
  {}

\theoremstyle{axiomwithbreak}

\theoremstyle{definition}
\newtheorem{definition}[theorem]{Definition}

\usetikzlibrary{arrows.meta,positioning}
\usepackage[colorlinks=true,citecolor=blue,linkcolor=blue,urlcolor=blue]{hyperref}
\usepackage[most]{tcolorbox}
\newtcolorbox{scopebox}[1][]{
  enhanced, breakable,
  colback=blue!4, colframe=blue!45!black,
  boxrule=0pt, leftrule=2.5pt, arc=0pt,
  left=6pt, right=6pt, top=4pt, bottom=4pt,
  fonttitle=\bfseries, coltitle=blue!45!black,
  attach title to upper={.\ },
  #1
}

\newcommand{\statementend}{\unskip\nobreak\hfill\ensuremath{\diamond}}
\makeatletter
\newcommand{\statementeqend}{%
  \let\statement@tagform\tagform@
  \def\tagform@##1{\statement@tagform{##1}\,\ensuremath{\diamond}}%
}
\makeatother

\providecommand{\F}{\mathbb{F}}
\providecommand{\one}{\mathbf{1}}
\providecommand{\calL}{\mathcal{L}}

\providecommand{\rank}{\operatorname{rank}}
\providecommand{\wt}{\operatorname{wt}}

\graphicspath{{figures/}}
\newcommand{\foldP}{\textcolor{blue}{P}}
\newcommand{\foldR}{\textcolor{red}{R}}

\begin{document}

\title{Design Principles for Ultra-High-Rate Quantum Codes}
\author{Jong Yeon Lee\textsuperscript{*,$\dagger$}}
\affiliation{Department of Physics and Institute for Condensed Matter Theory, University of Illinois Urbana-Champaign, Urbana, Illinois 61801, USA}
\affiliation{Korea Institute for Advanced Study, Seoul 02455, South Korea}
\author{Koki Okada\textsuperscript{*}}
\affiliation{Department of Information and Communications Engineering,
School of Engineering, Institute of Science Tokyo,
Tokyo 152-8550, Japan}
\author{Nishad Maskara}
\affiliation{Center for Theoretical Physics -- a Leinweber Institute, Massachusetts Institute of Technology, Cambridge, MA 02139, USA}
\author{Kenta Kasai\textsuperscript{$\dagger$}}
\affiliation{Department of Information and Communications Engineering,
School of Engineering, Institute of Science Tokyo,
Tokyo 152-8550, Japan}
\author{Hengyun Zhou\textsuperscript{$\dagger$}}
\affiliation{Department of Electrical Engineering and Computer Science, Massachusetts Institute of Technology, Cambridge, MA 02139, USA}

\begin{abstract}
Reducing the qubit overhead of quantum error correction is a central challenge for scalable fault-tolerant quantum computing. Recent ultra-high-rate quantum codes offer a promising route toward this goal, with some constructions requiring as few as two physical data qubits per logical qubit. However, systematic principles for navigating the tradeoffs among encoding rate, distance, check weight, and blocklength remain lacking. Here, we develop and analyze principles for exploring this design space, and use them to design compact code constructions with improved performance. We develop code templates based on a pair-partition construction and a halving transformation that further reduces blocklength. Motivated by ensemble analysis of the degree distributions, we identify column weight as a key design parameter: increasing the column weight enables larger distances at compact blocklengths, at the cost of heavier checks. We find that at physical error rates of 0.1\%, the benefits of increased distance often outweigh the penalty associated with heavier checks. Applying this framework, we identify numerous compact codes with favorable parameters, including $[[90,21,11]]$, $[[140,31,15]]$, and $[[200,43,20]]$ non-CSS codes with check weight 10. Moreover, we develop symmetry-informed strategies for identifying low-weight logical bases with canonical pairing.
These results provide systematic strategies for designing ultra-high-rate quantum codes and navigating their Pareto frontier.
\end{abstract}

\maketitle

\begingroup
\renewcommand{\thefootnote}{*}
\footnotetext[1]{These authors contributed equally.}
\renewcommand{\thefootnote}{$\dagger$}
\footnotetext[2]{Corresponding authors: \href{mailto:jongyeon@illinois.edu}{jongyeon@illinois.edu}; \\ \href{mailto:kenta@ict.eng.isct.ac.jp}{kenta@ict.eng.isct.ac.jp}; \href{mailto:hyzhou@mit.edu}{hyzhou@mit.edu}.}
\endgroup

\makeatletter
\def\l@subsection#1#2{}%
\def\l@subsubsection#1#2{}%
\def\tocstop{\def\l@section##1##2{}}%
\makeatother

\tableofcontents

\section{Introduction}

Quantum error correction (QEC) is expected to be an essential component of large-scale quantum computation, but the associated physical qubit overhead remains a major obstacle. Reducing this overhead requires codes that protect many logical qubits using a comparatively small number of physical qubits, while retaining sufficient distance and admitting reliable syndrome extraction and decoding. High-rate quantum low-density parity-check (qLDPC) codes provide a promising route toward this goal~\cite{gottesman2013fault,breuckmann2021quantum,tillich2014quantum}, with recent asymptotic constructions establishing that qLDPC codes can simultaneously achieve a constant encoding rate and distance linear in the blocklength~\cite{panteleev2022asymptotically,leverrier2022quantum}.

Alongside this asymptotic progress, increasing attention has turned to finite-size instances, with blocklengths in the hundreds to thousands~\cite{bravyi2024high,panteleev2019degenerate,tremblay2022constant,xu2024constant}. Recent constructions have pushed further into the ultra-high-rate regime, with encoding rates ranging from $10\%$ to over $50\%$~\cite{zhao2026towards,kasai2026breaking,cain2026shor,lee2026logicalspectroscopyliftedproductcodes,yang2026designer,lu2026quantum,bhardwaj2026high,yang2023spatially}. These results demonstrate that high encoding rates can be compatible with favorable finite-size performance, but they also reveal a broad and increasingly complex design space: Different constructions make different compromises among encoding rate, distance, decoding threshold, check weight, blocklength, and compatibility with physical implementation. Despite the growing number of available code families, a systematic framework for navigating these tradeoffs remains lacking. In particular, it is often unclear which design choices control the attainable distance and check weight at a given blocklength and encoding rate.

In this work, we develop a unified framework for navigating this design space and selecting ultra-high-rate quantum codes for different operating regimes. Inspired by classical coding theory, we use the row and column weight distributions, together with the lift size, as common coordinates for describing a code family~\cite{richardson2008modern}. This separates the underlying construction, which determines how the commutation constraints are satisfied, from the degree parameters and lift used to obtain a particular code instance. Simple dimension counting then relates the degree distributions to the design rate, providing a direct way to identify feasible combinations of encoding rate and check weight.

These coordinates lead to practical design rules for improving and comparing code families. Within a fixed construction, increasing the lift size increases the blocklength while approximately preserving the rate and degree distributions, and typically allows the distance to grow. Alternatively, increasing the column weight can improve the attainable distance at a given scale, but requires heavier checks at fixed rate, although we find that the increased distance often outweighs the cost of heavier checks at 0.1$\%$ physical error rates.
We further apply ensemble heuristics from classical coding theory to analyze the code performance in the waterfall and error floor regimes~\cite{richardson2001capacity,richardson2001design,Gallager1960}. More broadly, the framework provides a common basis for comparing different constructions: at fixed row and column weights, a better code family should reach a target distance using a smaller blocklength.

Guided by these principles, we propose code construction templates with favorable parameters below 1000 physical qubits. We develop a pair-partition construction template that achieves code parameters at the Pareto frontier of existing constructions.
We then utilize a symplectic halving procedure~\cite{burton2024genons} to produce even more compact non-CSS construction templates, further improving the previous parameter Pareto frontier by a factor of $\sim1.5$. Representative non-CSS instances have parameters including $[[90,21,11]]$, $[[140,31,15]]$, $[[200,43,20]]$ and, at comparable block lengths, achieve larger distances than the Gross and two-Gross codes~\cite{bravyi2024high} while encoding three to four times as many logical qubits. Circuit-level simulations of these instances yield logical error rates in the gigaquop-to-teraquop regime at a circuit-level error rate of $0.1\%$. Together, these results demonstrate how our framework can guide code design, compare construction families, and identify promising regions of the ultra-high-rate code landscape.

The remainder of this paper is organized as follows. Section~\ref{sec:construction} reviews systematic procedures for constructing ultra-high-rate quantum codes. Sections~\ref{sec:cpmpp} and~\ref{sec:noncss} develop the pair-partition and symplectic-halving constructions and present representative finite-size code instances.
Section~\ref{sec:de} analyzes decoding thresholds and ensemble distance trends from the degree distributions.
Sections~\ref{sec:comparison} and~\ref{sec:check_weight} compare these instances across construction families and examine the tradeoff between distance and check weight. Section~\ref{sec:logical_profile} develops symmetry-informed methods for constructing low-weight canonical logical bases.

\section{Design Principles for Ultra-High-Rate Codes}
\label{sec:construction}

In this section, we review recent approaches to constructing ultra-high-rate quantum codes at finite blocklength~\cite{hong2026quantum,bhardwaj2026high,yang2026designer,kasai2026breaking,zhao2026towards,lu2026quantum}. We separate two aspects of the design: the construction template, which specifies the structure of the code family and ensures stabilizer commutation, and the degree parameters, which further influence the tradeoff among check weight, attainable distance, and blocklength. Together, these provide a systematic way to construct quantum codes with high encoding rate, good distance, and small block length. Figure~\ref{fig:design_flow} summarizes the resulting design flow.

\begin{figure*}
\centering
\includegraphics[width=\textwidth]{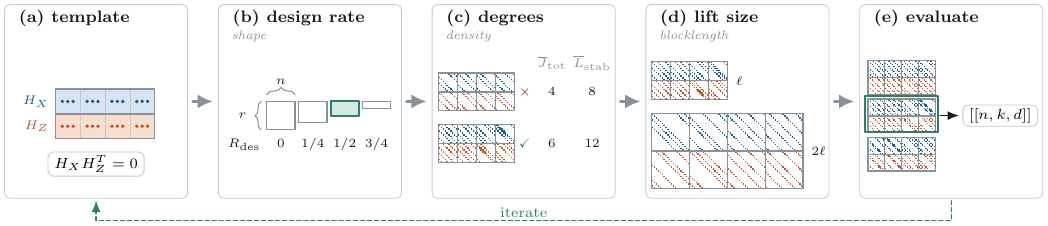}
\caption{\textbf{Design heuristics for ultra-high-rate quantum codes.}
\textbf{(a)}~The construction template---a protograph together with a lift
strategy---fixes which blocks are nonzero and enforces $H_XH_Z^{T}=0$.
\textbf{(b)}~The design rate fixes the \emph{shape}: counting qubit--check
incidences gives $n\overline{J}_{\mathrm{tot}}=r\overline{L}_{\mathrm{stab}}$, so
$R_{\mathrm{des}}=1-r/n$.
\textbf{(c)}~At fixed encoding rate the degree parameters
$\overline{J}_{\mathrm{tot}}$ and $\overline{L}_{\mathrm{stab}}$ are linked together; ensemble analysis suggests $\overline{J}_{\mathrm{tot}}\ge6$ is desirable, and higher degree tends to increase distance.
\textbf{(d)}~The lift size $\ell$, written $P$ for the circulant
constructions of Sec.~\ref{sec:cpmpp}, fixes the \emph{blocklength},
$n=b_{\mathrm q}\ell$, where $b_{\mathrm q}$ is the number of qubit blocks in
the protograph; the shape and the row/column weights are unchanged.
\textbf{(e)}~Candidate lifts are screened and iterated upon.}
\label{fig:design_flow}
\end{figure*}

\medskip
\noindent\textbf{Step 1: Specify the construction template.} The first, and arguably most important step, is to specify a construction template. A construction template consists of a protograph structure together with a lift strategy~\cite{thorpe2003protographs}. Most recent constructions have a parity-check matrix based on a block structure. The protograph specifies which blocks of the parity-check matrices are nonzero, while the lift strategy specifies how these blocks are replaced by larger matrices to obtain a family of finite codes. In a group-based lift, each nonzero protograph entry is replaced by a matrix representation of an element of the chosen group. The lift strategy therefore includes both the choice of group and any constraints relating the group elements assigned to different locations~\cite{richardson2008modern,panteleev2019degenerate,bravyi2024high,lin2024quantum}.

The construction template must, in particular, ensure that the stabilizer checks commute. For CSS codes with parity-check matrices $H_X$ and $H_Z$, this requires $H_X H_Z^T=0$~\cite{calderbank1996good,steane1996error}. The manner in which this condition is enforced can itself introduce distance-limiting structure. In some constructions, for example, it implies the existence of a nontrivial logical operator whose weight is no greater than the check weight~\cite{kasai2026breaking,hagiwara2010quantum}. A useful template should thus avoid such low-weight logical operators while retaining sufficient freedom among its lifts to produce favorable finite instances.

Many leading quantum LDPC constructions can be described within this framework: Hypergraph-product codes use a protograph derived from a product construction together with a trivial lift, whereas lifted-product codes retain the product structure while introducing nontrivial group-based lifts~\cite{tillich2014quantum,panteleev2022quantum}. The product structure can naturally produce nonuniform degree distributions, depending on the constituent codes. A complementary approach is to design a denser protograph and its lift constraints directly. Early quasi-cyclic QLDPC constructions follow this route~\cite{hagiwara2010quantum}, but enforce commutation in a manner that bounds the distance by the check weight. More recent constructions avoid this particular obstruction, by using affine-permutation-matrices and the pair-partition construction introduced here, enabling distances beyond the check weight~\cite{kasai2026breaking,zhao2026towards}. Overall, the construction template should avoid any structural mechanisms that lead to a bound on the distance.

An appropriately-chosen construction template can improve the parameter Pareto frontier accessible to the remaining design steps. For example, non-Abelian two-block group algebra constructions and affine-permutation constructions provide distinct ways to expand the design space of quantum codes~\cite{lin2024quantum,hong2026quantum,bhardwaj2026high,yang2026designer,kasai2026breaking}. Another common method is to replace monomials in the lift by polynomials~\cite{yang2026designer,lee2026logicalspectroscopyliftedproductcodes}. In this work, we identify new construction templates that further improve upon the previously known parameter Pareto frontier.

\medskip
\noindent\textbf{Step 2: Fix the design rate.} Having chosen the construction template, we now aim to reach the desired code parameters, starting with a target encoding rate. Let $\overline{J}_{\mathrm{tot}}$ denote the average number of stabilizer checks incident on a data qubit, counting checks in all bases, and let $\overline{L}_{\mathrm{stab}}$ denote the average check weight. If there are $n$ data qubits and $r$ checks, counting qubit--check incidences gives $n\overline{J}_{\mathrm{tot}}=r\overline{L}_{\mathrm{stab}}$. The realized encoding rate is therefore
\begin{equation}
    R = 1-\frac{\overline{J}_{\mathrm{tot}}}{\overline{L}_{\mathrm{stab}}}+\rho_{\mathrm{dep}},\label{eq:design_rate}
\end{equation}
where $\rho_{\mathrm{dep}}$ is the total rank deficiency of the stabilizer checks normalized by $n$. Thus, $1-\overline{J}_{\mathrm{tot}}/\overline{L}_{\mathrm{stab}}$ is the design rate, while check dependencies contribute additional logical qubits~\cite{richardson2008modern}. For a CSS code, $\overline{J}_{\mathrm{tot}}=\overline{J}_X+\overline{J}_Z$, and in the basis-symmetric case with data-qubit degree $J$ in each basis, $\overline{J}_{\mathrm{tot}}=2J$.

Commonly studied bivariate bicycle codes have zero design rate (also known as ``rate-less codes''), so their logical qubits arise from dependencies among the nominal stabilizer checks~\cite{bravyi2024high,lin2024quantum}. This mechanism can yield useful numbers of logical qubits at small blocklengths, but zero design rate often prevents the number of logical qubits from growing proportionally with blocklength.

In contrast, many recent ultra-high-rate constructions begin with a positive design rate, ensuring a high encoding rate by construction even before accounting for additional check redundancies~\cite{kasai2026breaking,zhao2026towards,hong2026quantum,bhardwaj2026high,lu2026quantum}. An important consequence of Eq.~\eqref{eq:design_rate} is that at a fixed design rate, the average data qubit degree and check weight cannot be varied independently, exposing an important tradeoff.

\medskip
\noindent\textbf{Step 3: Choose the degree parameters.} Fixing the design rate constrains the ratio between the average data qubit degree and the average check weight, but it does not determine their individual values. The data qubit degree therefore remains an important design lever. Intuitively, increasing this degree subjects each qubit to more constraints and can thereby support a larger code distance.

We make this intuition quantitative in two complementary ways. The ensemble analysis in Sec.~\ref{sec:de} predicts how the attainable performance depends on the degree parameters, while the explicit code constructions in Sec.~\ref{sec:cpmpp} and Sec.~\ref{sec:noncss} exhibit the same trend within a fixed construction family. In particular, both the density-evolution threshold and the ensemble weight-growth exponent indicate that average data qubit degree two is generally insufficient to achieve good performance, consistent with familiar principles from classical LDPC code design~\cite{richardson2001capacity,richardson2001design}. Consequently, an average data qubit degree of at least three is often needed; together with the target design rate, this requirement directly lower bounds the average check weight. For example, we expect it to be challenging to design rate-$1/2$ codes with good performance and maximum check weight below 12.

The data qubit degree thus provides a useful knob for navigating the code parameter landscape. Smaller degrees permit lighter checks at a fixed design rate, but may not support the desired distance at compact blocklengths. Increasing the degree can improve the attainable distance at a given scale, although it necessarily requires heavier checks. In Sec.~\ref{sec:check_weight}, we analyze this tradeoff numerically and find that, at commonly assumed physical error rates of 0.1\%, the distance gained by increasing the degree can justify the additional check weight.

\medskip
\noindent\textbf{Step 4: Choose the lift size and generate candidates.} Having fixed the construction template, design rate, and degree parameters, we next choose a lift size and generate concrete code instances. The lift size directly sets the blocklength up to a factor determined by the protograph. Changing it preserves the design rate and degree parameters, although the realized rate can vary because the number of check dependencies may depend on the particular lift. A larger lift also provides more freedom within the same structural family and can therefore support a larger distance.

Candidate lifts can often be generated by random sampling followed by screening, although informed search heuristics may improve the efficiency. Each candidate should first be checked for the defining constraints of the construction, the intended degree distribution, and the absence of any readily detectable low weight logical operators. The lift size and data qubit degree then provide complementary ways to improve the attainable distance.

\medskip
\noindent\textbf{Step 5: Evaluate instances and iterate.} After generating a collection of valid candidates, we evaluate them using criteria appropriate to the intended application. Distance is often a primary consideration and can be screened using various search methods~\cite{webster2026distance,pryadko2022qdistrnd,bhardwaj2026high}. Other relevant criteria include the girth, decoding performance, compatibility with logical operations, and constraints imposed by the target hardware. It can also be useful to retain a Pareto set of candidates to account for these tradeoffs. Finally, if one is not able to achieve the desired performance, one can return to earlier stages of the design to adjust the construction template, degree distribution, or blocklength.

\begin{table}[!t]
    \centering
    \small
    \setlength{\tabcolsep}{3pt}
    \begin{ruledtabular}
    \begin{tabular}{@{}ccccc@{}}
        $(J,L)$ & $[[n,k,d]]$ & Girth & Lift $P$ & Rate $k/n$\\
        \colrule
        $(3,8)$ & $[[472,122,16]]$ & $8$ & $59$ & $0.258$\\
        $(3,8)$ & $[[776,198,20]]$ & $8$ & $97$ & $0.255$\\
        $(3,10)$ & $[[970,392,16]]$ & $6$ & $97$ & $0.404$\\
        $(3,10)$ & $[[1630,656,20]]$ & $8$ & $163$ & $0.402$\\
        $(3,12)$ & $[[2004,1006,16]]$ & $6$ & $167$ & $0.502$\\
        $(3,12)$ & $[[2676,1342,18]]$ & $6$ & $223$ & $0.501$\\
        $(3,14)$ & $[[1414,812,12]]$ & $6$ & $101$ & $0.574$\\
        $(3,14)$ & $[[3122,1788,16]]$ & $6$ & $223$ & $0.573$\\
        \colrule
        $(4,10)$ & $[[220,50,16]]$ & $6$ & $22$ & $0.227$\\
        $(4,10)$ & $[[320,70,20]]$ & $6$ & $32$ & $0.219$\\
        $(4,12)$ & $[[372,130,16]]$ & $6$ & $31$ & $0.349$\\
        $(4,12)$ & $[[492,170,20]]$ & $6$ & $41$ & $0.346$\\
        $(4,14)$ & $[[518,228,16]]$ & $6$ & $37$ & $0.440$\\
        $(4,14)$ & $[[574,252,18]]$ & $6$ & $41$ & $0.439$\\
        $(4,16)$ & $[[848,430,18]]$ & $6$ & $53$ & $0.507$\\
        $(4,16)$ & $[[944,478,20]]$ & $6$ & $59$ & $0.506$\\
        \colrule
        $(5,10)$ & $[[170,8,20]]$ & $6$ & $17$ & $0.047$
    \end{tabular}
    \end{ruledtabular}
    \caption{Representative CPM-based CSS pair-partition codes. A more extended
    list can be found in Table~\ref{tab:cpmpp-code-catalogue-full}.}
    \label{tab:cpmpp-code-catalogue}
\end{table}

\section{CPM-Based Pair-Partition Codes}
\label{sec:cpmpp}

Having described a general recipe for constructing ultra-high-rate quantum codes, we now instantiate it with a new construction template, which we refer to as the pair-partition (PP) code. In this section, we focus on this construction using circulant permutation matrices (CPMs), although the method can be generalized to APMs, which allows exceeding the structural distance bounds on CPM-based constructions imposed by column weight (Appendix~\ref{sec:cpmpp-structural-bounds}).

\subsection{Definition and CSS orthogonality}
\label{sec:cpmpp-construction}

The pair-partition construction enforces CSS orthogonality by pairing equal contributions to the orthogonality equations. For $s\in\mathbb Z_P$, let $C(s)$ denote the $P\times P$ CPM whose nonzero entry in row $a$ lies in column $a-s$, with indices taken modulo $P$. Lifting a protograph by these matrices gives the standard quasi-cyclic LDPC construction~\cite{tanner2004ldpc,fossorier2004quasicyclic}. The pair-partition construction is defined using these CPMs as follows:

\begin{definition}[CPM-based pair-partition construction]
    Fix positive integers $J,P$ and an even positive integer $L$.
Let $M=(M_{ij})_{0\leq i,j<J}$ be a $J\times J$ array whose entry
$M_{ij}$ is a partition of $\mathbb Z_L$ into $L/2$ disjoint unordered
pairs. Choose exponent arrays $E=(e_{i\ell}),D=(d_{j\ell})\in
\mathbb Z_P^{J\times L}$ satisfying
\begin{equation}
    d_{j,u}-e_{i,u}\equiv d_{j,v}-e_{i,v}\pmod P,
    \label{eq:cpmpp-paired-difference}
\end{equation}
for every $0\leq i,j<J$ and every $\{u,v\}\in M_{ij}$. The associated CPM-based pair-partition
(CPM-PP) code is the CSS code with binary check matrices
\begin{equation}
    H_X=\bigl(C(e_{i\ell})\bigr)_{i,\ell},\qquad
    H_Z=\bigl(C(d_{j\ell})\bigr)_{j,\ell}.
    \label{eq:cpmpp-checks}
\end{equation}
The array $M$ specifies the construction template independently of the
lift size; $P$ and compatible exponent arrays specify a code instance.\statementend
\end{definition}

\noindent\textbf{CSS orthogonality.}
The $(i,j)$ block of $H_XH_Z^{\mathsf T}$ is
\begin{equation}
    \sum_{\ell=0}^{L-1}C(e_{i\ell})C(d_{j\ell})^{\mathsf T}
    =\sum_{\ell=0}^{L-1}C(e_{i\ell}-d_{j\ell}).
    \label{eq:cpmpp-cancellation}
\end{equation}
For each $\{u,v\}\in M_{ij}$, Eq.~\eqref{eq:cpmpp-paired-difference} makes the two corresponding summands identical, so they cancel over $\mathbb F_2$. Since $M_{ij}$ partitions all block columns, the entire block is zero. Applying this argument to every $i,j$ gives $H_XH_Z^{\mathsf T}=0$.

Each check matrix is a CPM lift of the all-one $J\times L$ protograph.
Every lifted row has one nonzero entry in each block column, and every
lifted column has one in each block row. The code therefore has
blocklength $n=LP$, row weight $L$, and column weight $J$ in each check
matrix. Its number of logical qubits is
\begin{equation}
    k=LP-\operatorname{rank}H_X
         -\operatorname{rank}H_Z.
    \label{eq:cpmpp-dimension}
\end{equation}
Thus, in the notation of Sec.~\ref{sec:construction}, $\overline{J}_{\mathrm{tot}}=2J$ and
$\overline{L}_{\mathrm{stab}}=L$, with design rate
$R_{\mathrm{des}}=1-2J/L$.

For a pair $B=\{u,v\}\in M_{ij}$, define
$m_{ij}(B)=d_{j,u}-e_{i,u}=d_{j,v}-e_{i,v}$ in
$\mathbb{Z}_P$.
As an additional search constraint, we may require
\begin{equation}
    m_{ij}(B)\ne m_{ij}(B')
    \quad\text{for distinct }B,B'\in M_{ij}.
    \label{eq:cpmpp-mixed-distinct}
\end{equation}
This mixed-distinctness condition prevents different pairs
in the same $M_{ij}$ from sharing a mixed-difference value.
For fixed lifted rows $(i,a)$ of $H_X$ and $(j,b)$ of $H_Z$,
a common nonzero occurs in block column $\ell$ exactly when
$d_{j,\ell}-e_{i,\ell}=b-a$.
Equation~\eqref{eq:cpmpp-mixed-distinct} therefore makes every
such row intersection have size zero or two. This reduces the number of
mixed $4$-cycles in the combined $X/Z$ incidence graph, providing a
heuristic motivation for imposing the condition to improve decoding.

\subsection{Search procedure}
\label{sec:cpmpp-design}

We now provide additional details about our search procedure to identify promising CPM-PP code instances.

We first provide a structural upper bound on Tanner graph girth from the choice of pair-partition array. For a fixed $Z$-block row $j$, combine the pairs in $M_{ij}$ over all $i$ to
form a colored multigraph $\Gamma_j^X$ on the $L$ block-column indices, in which
the $L/2$ pairs of $M_{ij}$ are the edges of color $i$ (Fig.~\ref{fig:pairing-graphs}).  Similarly, combining
$M_{ij}$ over $j$ for fixed $i$ gives $\Gamma_i^Z$, in which the pairs of
$M_{ij}$ are the edges of color $j$. Each color class is a perfect matching of
$\mathbb{Z}_L$, so every pairing graph is $J$-regular on $L$ vertices.  These
$2J$ pairing graphs depend on $J,L$, and $M$, but not on the lift size or the
CPM exponents.

\begin{Proposition}[Pairing-graph girth bound]
    For any exponent arrays satisfying
Eq.~\eqref{eq:cpmpp-paired-difference}, a cycle of length $r\geq2$
in $\Gamma_j^X$ forces a Tanner cycle of length at most $2r$ in
$H_X$. The corresponding statement holds for $\Gamma_i^Z$ and $H_Z$.
Two parallel edges count as a cycle of length two.\statementend
\end{Proposition}
\begin{proof}
    Let $c_0,c_1,\ldots,c_{r-1},c_r=c_0$ be a cycle in
$\Gamma_j^X$, and let $i_t$ be the color of the edge
$\{c_t,c_{t+1}\}\in M_{i_tj}$. Summing the paired-difference equations around the cycle yields
\begin{equation}
    \begin{aligned}
    &\sum_{t=0}^{r-1}\bigl(e_{i_t,c_t}-e_{i_t,c_{t+1}}\bigr)=\sum_{t=0}^{r-1}\bigl(d_{j,c_t}-d_{j,c_{t+1}}\bigr)=0\pmod P.\nonumber
    \end{aligned}
    \label{eq:cpmpp-pairing-cycle-sum}
\end{equation}
Indeed, each $d_{j,c_t}$ occurs once with each sign. Starting at lift
coordinate $a_t$ in block column $c_t$, a step through block row $i_t$
reaches block column $c_{t+1}$ at coordinate
$a_{t+1}=a_t+e_{i_t,c_t}-e_{i_t,c_{t+1}}$.
The zero sum therefore closes the lifted walk after $2r$ Tanner edges, resulting in a
Tanner cycle of length at most $2r$. Interchanging $E$ and $D$ gives
the statement for $H_Z$.
\end{proof}

A repeated edge is the case $r=2$ and forces a Tanner 4-cycle; a triangle
forces a Tanner cycle of length at most six.  Pairing graphs used for a
girth-eight search must therefore be simple and triangle-free. These
are necessary conditions; the chosen exponents must still pass the
lifted Tanner-cycle tests.
It is also known that every CPM lift of the complete $J\times L$ protograph with $J\geq2$ and
$L\geq3$ has Tanner girth at most twelve, independently of
$M,E,D$~\cite{fossorier2004quasicyclic}.
To illustrate these restrictions, write $\langle03,12,46,57\rangle$ for the pair-partition
$\{\{0,3\},\{1,2\},\{4,6\},\{5,7\}\}$.  For $J=3$ and $L=8$, one complete
array is
\begin{equation}
\begin{aligned}
M_{00}&=\langle03,12,46,57\rangle,&
M_{01}&=\langle02,16,37,45\rangle,\\
M_{02}&=\langle06,15,23,47\rangle,&
M_{10}&=\langle05,13,24,67\rangle,\\
M_{11}&=\langle03,15,27,46\rangle,&
M_{12}&=\langle07,14,26,35\rangle,\\
M_{20}&=\langle01,26,35,47\rangle,&
M_{21}&=\langle05,17,24,36\rangle,\\
M_{22}&=\langle02,16,37,45\rangle.
\end{aligned}
\label{eq:cpmpp-example-array}
\end{equation}
The six pairing graphs of the array in Eq.~\eqref{eq:cpmpp-example-array} are drawn in Fig.~\ref{fig:pairing-graphs}. No graph has a repeated edge, so no Tanner $4$-cycle is forced;
four of the six do contain a triangle, leading to girth at most $6$.

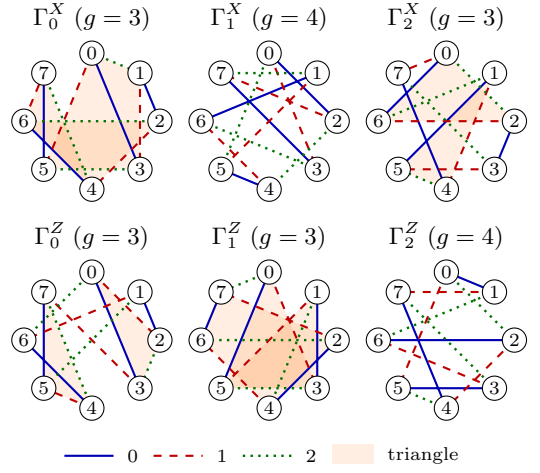
\begin{figure}[t]
    \centering
    \begingroup
\tikzset{
  pairzero/.style={draw=blue!70!black,line width=0.8pt},
  pairone/.style={draw=red!75!black,line width=0.8pt,dashed},
  pairtwo/.style={draw=green!45!black,line width=0.9pt,dotted},
  tricell/.style={fill=orange!88!red,fill opacity=0.12,draw=none},
  pairvertex/.style={circle,draw=black,fill=white,inner sep=0pt,
    minimum size=3.2mm,font=\scriptsize}
}
\newcommand{\paircoords}{%
  \foreach \v/\a in {0/90, 1/45, 2/0, 3/315, 4/270, 5/225, 6/180, 7/135}
    \coordinate (v\v) at (\a:0.90cm);%
}
\newcommand{\pairnodes}{%
  \foreach \v in {0,...,7} \node[pairvertex] at (v\v) {\v};%
}
\newcommand{\pairtitle}[2]{\multicolumn{1}{c}{$#1\;(g=#2)$}}
\setlength{\tabcolsep}{3pt}
\begin{tabular}{ccc}
\pairtitle{\Gamma_0^X}{3} & \pairtitle{\Gamma_1^X}{4} & \pairtitle{\Gamma_2^X}{3} \\[1mm]
\begin{tikzpicture}
  \paircoords
  \fill[tricell] (v0)--(v1)--(v3)--cycle;
  \fill[tricell] (v0)--(v3)--(v5)--cycle;
  \fill[tricell] (v2)--(v4)--(v6)--cycle;
  \fill[tricell] (v4)--(v6)--(v7)--cycle;
  \draw[pairzero] (v0)--(v3) (v1)--(v2) (v4)--(v6) (v5)--(v7);
  \draw[pairone] (v0)--(v5) (v1)--(v3) (v2)--(v4) (v6)--(v7);
  \draw[pairtwo] (v0)--(v1) (v2)--(v6) (v3)--(v5) (v4)--(v7);
  \pairnodes
\end{tikzpicture}
&
\begin{tikzpicture}
  \paircoords
  \draw[pairzero] (v0)--(v2) (v1)--(v6) (v3)--(v7) (v4)--(v5);
  \draw[pairone] (v0)--(v3) (v1)--(v5) (v2)--(v7) (v4)--(v6);
  \draw[pairtwo] (v0)--(v5) (v1)--(v7) (v2)--(v4) (v3)--(v6);
  \pairnodes
\end{tikzpicture}
&
\begin{tikzpicture}
  \paircoords
  \fill[tricell] (v0)--(v2)--(v6)--cycle;
  \fill[tricell] (v1)--(v4)--(v5)--cycle;
  \draw[pairzero] (v0)--(v6) (v1)--(v5) (v2)--(v3) (v4)--(v7);
  \draw[pairone] (v0)--(v7) (v1)--(v4) (v2)--(v6) (v3)--(v5);
  \draw[pairtwo] (v0)--(v2) (v1)--(v6) (v3)--(v7) (v4)--(v5);
  \pairnodes
\end{tikzpicture}
\\[1mm]
\pairtitle{\Gamma_0^Z}{3} & \pairtitle{\Gamma_1^Z}{3} & \pairtitle{\Gamma_2^Z}{4} \\[1mm]
\begin{tikzpicture}
  \paircoords
  \fill[tricell] (v0)--(v2)--(v3)--cycle;
  \fill[tricell] (v4)--(v5)--(v7)--cycle;
  \draw[pairzero] (v0)--(v3) (v1)--(v2) (v4)--(v6) (v5)--(v7);
  \draw[pairone] (v0)--(v2) (v1)--(v6) (v3)--(v7) (v4)--(v5);
  \draw[pairtwo] (v0)--(v6) (v1)--(v5) (v2)--(v3) (v4)--(v7);
  \pairnodes
\end{tikzpicture}
&
\begin{tikzpicture}
  \paircoords
  \fill[tricell] (v0)--(v3)--(v5)--cycle;
  \fill[tricell] (v1)--(v3)--(v5)--cycle;
  \fill[tricell] (v2)--(v4)--(v6)--cycle;
  \fill[tricell] (v2)--(v6)--(v7)--cycle;
  \draw[pairzero] (v0)--(v5) (v1)--(v3) (v2)--(v4) (v6)--(v7);
  \draw[pairone] (v0)--(v3) (v1)--(v5) (v2)--(v7) (v4)--(v6);
  \draw[pairtwo] (v0)--(v7) (v1)--(v4) (v2)--(v6) (v3)--(v5);
  \pairnodes
\end{tikzpicture}
&
\begin{tikzpicture}
  \paircoords
  \draw[pairzero] (v0)--(v1) (v2)--(v6) (v3)--(v5) (v4)--(v7);
  \draw[pairone] (v0)--(v5) (v1)--(v7) (v2)--(v4) (v3)--(v6);
  \draw[pairtwo] (v0)--(v2) (v1)--(v6) (v3)--(v7) (v4)--(v5);
  \pairnodes
\end{tikzpicture}

\end{tabular}

\vspace{1mm}
\begin{tikzpicture}[baseline=-0.6ex]
  \draw[pairzero] (0.00,0)--(0.65,0);
  \node[anchor=west,font=\scriptsize] at (0.72,0) {0};
  \draw[pairone] (1.18,0)--(1.83,0);
  \node[anchor=west,font=\scriptsize] at (1.90,0) {1};
  \draw[pairtwo] (2.36,0)--(3.01,0);
  \node[anchor=west,font=\scriptsize] at (3.08,0) {2};
  \fill[tricell] (3.54,-0.10) rectangle (4.09,0.16);
  \node[anchor=west,font=\scriptsize] at (4.16,0) {triangle};
\end{tikzpicture}
\endgroup
    \caption{\textbf{Pairing graphs of the $J=3$, $L=8$ array of
    Eq.~\eqref{eq:cpmpp-example-array}.}
    In $\Gamma_j^X$ the color-$i$ edges are the pairs of $M_{ij}$ read down
    column $j$; in $\Gamma_i^Z$ they are read across row $i$. Colors $0,1,2$
    are drawn as blue solid, red dashed, and green dotted lines. Every color
    class is a perfect matching, so each graph is $J$-regular on the $L$
    block-column indices, and $g$ is the length of its shortest cycle. A
    cycle of length $r$ here forces a Tanner cycle of length at most $2r$.
    Shaded regions are triangles on the combined graph, each forcing a Tanner $6$-cycle.}
    \label{fig:pairing-graphs}
\end{figure}

We generate $M$ by backtracking over the perfect matchings of
$\mathbb{Z}_L$.  When filling cell $(i,j)$, a matching is accepted only when
it shares no edge with the matchings already placed in row $i$ or column $j$.
For target girth eight, a tentative matching is also rejected if either union
contains a triangle.  Fixing the first matching by a relabeling of the block
columns removes an equivalent set of search branches.

After $M$ has been selected, fix a prime $P$. The paired-difference equations form a
homogeneous system in the $2JL$ entries of $E$ and $D$.  Row reduction over
$\mathbb{Z}_P$ then gives a basis for the solution space.  More explicitly, if
$x=(\operatorname{vec}E,\operatorname{vec}D)^{\mathsf T}$, the equations can
be written as
\begin{equation}
    A_Mx=0,
    \qquad A_M\in\mathbb{Z}_P^{(J^2L/2)\times 2JL},
    \label{eq:cpmpp-linear-system}
\end{equation}
where the row for $\{u,v\}\in M_{ij}$ has coefficients
$-1,+1,+1,-1$ in the positions of
$e_{i,u},e_{i,v},d_{j,u},d_{j,v}$.  The transformation
\begin{equation}
    e'_{i\ell}=e_{i\ell}+\tau_\ell+\alpha_i,
    \qquad
    d'_{j\ell}=d_{j\ell}+\tau_\ell+\beta_j
    \label{eq:cpmpp-gauge}
\end{equation}
preserves orthogonality, binary ranks, Tanner girth, and distance, for any
column offsets $\tau_\ell$ and block-row offsets $\alpha_i,\beta_j$ in
$\mathbb{Z}_P$. We partially fix this relabeling freedom by taking
\begin{equation}
    e_{0\ell}=0\qquad(0\leq\ell<L).
    \label{eq:cpmpp-normalization}
\end{equation}
This leaves $(2J-1)L$ scalar variables before the remaining paired equations
are row-reduced. The normalization does not select a unique representative:
setting $\tau_\ell=0$ and $\alpha_0=0$ in
Eq.~\eqref{eq:cpmpp-gauge} still permits arbitrary shifts of the other
$E$ rows and all $D$ rows.

Our full search procedure can be summarized as follows:
\begin{enumerate}
    \item Choose $J,L$, and a prime $P$, and generate a pair-partition array
    $M$ by backtracking.  Reject edge reuse, and reject triangles when Tanner
    girth at least eight is required.
    \item Form $A_M$ and compute $\ker A_M$ over $\mathbb Z_P$, using
    Eq.~\eqref{eq:cpmpp-normalization} to partially fix the exponent
    relabeling freedom.
    \item Enumerate or sample $E,D$ from this solution space.  Apply
    Eq.~\eqref{eq:cpmpp-mixed-distinct} and the modular cycle tests
    before constructing the binary matrices.
    \item Construct $H_X,H_Z$ and directly verify CSS orthogonality, binary
    ranks, row and column weights, and Tanner girths.
    \item Apply the logical-operator search of
    Appendix~\ref{sec:distance-verification}.  Retain the fixed exponent arrays
    only when they meet the prescribed distance screen.
\end{enumerate}

\subsection{Code instances and distance trends}
\label{sec:cpmpp-instances}

Table~\ref{tab:cpmpp-code-catalogue} summarizes representative finite CPM-PP instances we construct using this method. The complete catalogue is given in Appendix~\ref{sec:cpmpp-full-catalogue}, see the accompanying repository for details~\cite{designprinciplesdata}. The instances achieve code parameters at or near the Pareto frontier of code constructions, for comparable encoding rates and check weights. As we can see, the design rate directly ties together the row and column weight. As the encoding rate increases, the block length required to achieve a given distance also tends to increase. For similar block lengths and encoding rates, the higher column weight $J=4$ instances result in larger distances: for example, compare the $J=3$ code instance $[[530,216,12]]$ with the $J=4$ code instance $[[518,228,16]]$. This can be understood from the fact that larger column weight imposes more constraints, making it easier to avoid low-weight logicals. We provide an additional, ensemble-level perspective on this trend in Sec.~\ref{sec:de}, and show that the increased distance can help performance despite a larger check weight in Sec.~\ref{sec:check_weight}.

The CPM-PP construction has a distance ceiling $d\leq(J+1)!$ for a given column weight, which is $24$ for $J=3$ and $120$ for $J=4$, independent of $P$, see Appendix~\ref{sec:cpmpp-structural-bounds} for more details~\cite{mackay2001evaluation,smarandache2012quasi,fossorier2004quasicyclic}.
The construction can be generalized to use affine permutation matrices (APMs) instead, which circumvents this distance upper bound.
While the greater flexibility afforded by this construction can in principle also increase distance at smaller block lengths, we empirically find that only at larger block lengths above 1000, when the distance approaches the structural bound for the CPM-PP construction, does the APM-PP start outperforming the CPM-PP construction.
This is consistent with the distance at smaller block length not being limited by global structure, and instead by the amount of available space within the lift.

\section{Symplectic Halving and Non-CSS Codes}
\label{sec:noncss}

\begin{figure}[!t]
    \centering
    \includegraphics[width=\columnwidth]{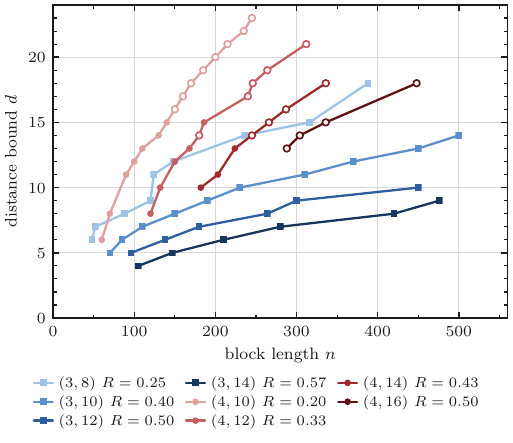}
    \caption{\textbf{Distance frontier of the halved non-CSS search.}
    Best distance $d$ reached at or below blocklength $n$
    for eight $(J,L)$ families, where $J$ is the number of
    check block rows in one basis, $L$ the number of block columns, and $R_{\mathrm{des}}=1-2J/L$ the
    design rate. Filled markers are exact distances from a completed
    exhaustion; open markers are upper bounds from a logical-operator witness.}
    \label{fig:halved-frontier}
\end{figure}

The majority of recent ultra-high-rate code constructions~\cite{cain2026shor,zhao2026towards,hong2026quantum,bhardwaj2026high,kasai2026breaking,lu2026quantum} focus on CSS codes~\cite{calderbank1996good,steane1996error}.
However, general stabilizer codes offer a larger design space by removing the requirement that the stabilizer group admits a generating set of purely $X$-type and purely $Z$-type checks.
This relaxes the structural restrictions on how a fixed number of commuting checks can be arranged to exclude low-weight logical operators, opening additional possibilities for improving finite-size code parameters.
Indeed, improvements in code parameters from non-CSS codes already appear for the smallest distance-three codes: the five-qubit perfect code is a non-CSS $[[5,1,3]]$ code, whereas the shortest CSS counterpart is the $[[7,1,3]]$ Steane code~\cite{laflamme1996perfect,steane1996error}.
A complementary benefit of mixed-basis checks is their ability to tailor a code to biased noise, as demonstrated by the XZZX surface code and its variants~\cite{bonilla2021xzzx,xu2023tailored}.
More recently, non-CSS generalizations of bivariate bicycle codes and two-block group algebra codes were proposed, but they were focused on small blocklengths, check weight six or seven, and vanishing design rate~\cite{khesin2026mirror}.
Here, we instead aim to use the non-CSS design freedom directly in the ultra-high-rate regime and improve the Pareto frontier of code constructions.

\begin{table}[!t]
    \centering
    \small
    \setlength{\tabcolsep}{3pt}
    \begin{ruledtabular}
    \begin{tabular}{@{}cccc@{}}
        $(J,L)$ & $[[n,k,d]]$ & Lift $P$ & Rate $k/n$\\
        \colrule
        $(3,8)$ & $[[52,15,7]]$ & $13$ & $0.288$\\
        $(3,8)$ & $[[316,81,15]]$ & $79$ & $0.256$\\
        $(3,8)$ & $[[388,99,18]]$ & $97$ & $0.255$ \\
        $(3,10)$ & $[[500,202,14]]$ & $100$ & $0.404$\\
        $(3,12)$ & $[[450,227,10]]$ & $75$ & $0.504$\\
        $(3,14)$ & $[[476,274,9]]$ & $68$ & $0.576$\\
        \colrule
        $(4,10)$ & $[[90,21,11]]$ & $18$ & $0.233$\\
        $(4,10)$ & $[[140,31,15]]$ & $28$ & $0.221$\\
        $(4,10)$ & $[[150,33,\leq 16]]$ & $30$ & $0.220$\\
        $(4,10)$ & $[[200,43,20]]$ & $40$ & $0.215$\\
        $(4,12)$ & $[[240,83,\leq 17]]$ & $40$ & $0.346$\\
        $(4,12)$ & $[[312,107,\leq 21]]$ & $52$ & $0.343$\\
        $(4,14)$ & $[[287,126,\leq 16]]$ & $41$ & $0.439$\\
        $(4,14)$ & $[[336,147,\leq 18]]$ & $48$ & $0.438$\\
        $(4,16)$ & $[[336,171,\leq 15]]$ & $42$ & $0.509$\\
    \end{tabular}
    \end{ruledtabular}
    \caption{Representative halved non-CSS pair-partition codes. The full frontier is given in Table~\ref{tab:noncss-code-catalogue-full}.}
    \label{tab:noncss-code-catalogue}
\end{table}

\subsection{Symplectic Halving Method}

A useful starting point is the symplectic doubling map from general stabilizer codes to CSS codes~\cite{Kovalev_2012, kovalev2013quantum,liu2023subsystem,burton2024genons}. Let $H=(A\mid B)$ be the binary symplectic check matrix of an $[[n,k,d]]$ non-CSS stabilizer code.
For binary Pauli vectors $p=(x| z)$ and $q=(x' | z')$, the symplectic form
\begin{equation}
B_{\mathrm{sp}}(p,q)=x^{\mathsf T}z'+z^{\mathsf T}x'
\end{equation}
is zero when the corresponding operators commute and one when they anticommute.
The matrix $H$ represents commuting Pauli checks precisely when $AB^{\mathsf T}+BA^{\mathsf T}=0$. Therefore, the matrices
\begin{equation}
    H_X'=(A\mid B),\qquad H_Z'=(B\mid A)
    \label{eq:stabilizer_doubling}
\end{equation}
satisfy the CSS orthogonality, defining a CSS code $[[n'=2n,k'=2k,d']]$ with $d\leq d'\leq2d$.
We use the reverse operation to construct stabilizer codes from CSS parents.

\begin{definition}[Symplectic halving]
    Suppose a CSS code on $2n$ qubits admits a pairing and a choice of check
generators such that, after placing one member of each pair in the first
half of the qubit ordering and its partner in the second half,
\begin{equation}
    H_X=(A\mid B),\qquad H_Z=(B\mid A)
    \label{eq:halving_parent}
\end{equation}
Its \emph{symplectic halving} gives the non-CSS stabilizer code on $n$ qubits with binary symplectic check matrix
\begin{equation}
    H_{\mathrm{halved}}= ( \underbrace{A}_{\text{Pauli-}X}
      \mid\underbrace{B}_{\text{Pauli-}Z}).
    \label{eq:halving_check}
\end{equation}
For each pair, the entries of a parent $X$-check on the first and second qubits become the $X$ and $Z$ components of the corresponding halved check.
The resulting checks commute because the parent CSS orthogonality implies $H_X^{\vphantom{\dagger}} H_Z^{{\dagger}} \equiv A B^{\mathsf{T}} + B A^{\mathsf{T}} = 0$.
\end{definition}

Symplectic doubling and its inverse provide a correspondence between general stabilizer codes and CSS codes admitting a fixed-point-free ZX duality~\cite{kovalev2013quantum, breuckmann2024fold,quintavalle2023partitioning,burton2024genons}. Earlier applications yield halved non-CSS hypergraph-product and hyperbicycle codes~\cite{Kovalev_2012, kovalev2013quantum}. Here, we derive explicit halving criteria for CPM pair-partition templates, characterize their compatibility with the lift-group symmetry, and incorporate these constraints directly into the code search.

\vspace{5pt}\noindent\emph{Folding map and code parameters.}
Let $S_X=\operatorname{im}H_X^{\mathsf T}$ and $S_Z=\operatorname{im}H_Z^{\mathsf T}$ be the parent check spaces, and let $\pi$ be the fixed-point-free involution satisfying $\pi(S_X)=S_Z$.
Choose one endpoint $q_i$ from each pair $\{q_i,\pi(q_i)\}$. For a binary vector $a$ on $2n$ qubits, define
\begin{align}
F_\pi(a) &= (x|z),\quad x_i=a_{q_i},\quad z_i=a_{\pi(q_i)}. \label{eq:folded-parent-map}
\end{align}
This linear bijection identifies parent binary words with Pauli vectors on $n$ folded qubits.

For any parent binary words $a,b\in\F_2^{2n}$, the folding map produces Pauli vectors $F_\pi(a)=(x|z)$ and $F_\pi(b)=(x'|z')$.
Their binary commutation bit satisfies
\begin{align}
B_{\mathrm{sp}}(F_\pi(a),F_\pi(b))
&=a^{\mathsf T}\pi(b).
\label{eq:folded-twisted-pairing}
\end{align}
Changing the chosen endpoint of a pair exchanges its $X$ and $Z$ components, corresponding to a local Hadamard.
Different orientations therefore give locally Clifford-equivalent folded codes, and this correspondence preserves Pauli weight.

CSS orthogonality and Eq.~\eqref{eq:folded-twisted-pairing} give the folded stabilizer and normalizer spaces:
\begin{align}
S_{\mathrm{fold}} &= F_\pi(S_X),\qquad N_{\mathrm{fold}}=F_\pi(\ker H_Z). \label{eq:folded-stabilizer-normalizer}
\end{align}
Indeed, for $h,h'\in S_X$, the identity $\pi(h')\in S_Z$ and CSS orthogonality give $h\mathbin{\cdot}\pi(h')=0$, so the folded stabilizers commute.
Moreover,
\begin{align}
F_\pi(a)\in N_{\mathrm{fold}} &\Longleftrightarrow a\mathbin{\cdot}\pi(h)=0\quad\text{for every }h\in S_X \nonumber\\
&\Longleftrightarrow a\in S_Z^\perp=\ker H_Z.
\end{align}
These identities establish a bijection between parent $X$-logical classes and folded logical Pauli classes.
If the parent check matrices have rank $r$, the parent encodes $2n-2r$ logical qubits and the fold encodes $n-r$.

The folded Pauli weight counts occupied pairs, with a doubly occupied pair contributing one $Y$ operator.
Writing $c_\pi(a)$ for the number of pairs on which both coordinates of $a$ are one gives
\begin{align}
\wt_{\mathrm P}(F_\pi(a)) &= \wt(a)-c_\pi(a). \label{eq:folded-pauli-weight}
\end{align}
Thus $\wt(a)/2\leq\wt_{\mathrm P}(F_\pi(a))\leq\wt(a)$.
Since $\pi$ exchanges the two CSS sides and preserves Hamming weight, the parent $X$ and $Z$ distances coincide.
The code parameters therefore satisfy
\begin{align}
n_{\mathrm{fold}} = \frac{n'}{2},\quad k_{\mathrm{fold}}=\frac{k'}{2}, \quad
\left\lceil\frac{d'}{2}\right\rceil \leq d_{\mathrm{fold}}\leq d'. \label{eq:folded-code-parameters}
\end{align}
Halving preserves the encoding rate and never increases the supplied check weights.

We apply this procedure to the pair-partition construction in Sec.~\ref{sec:cpmpp}, which lies at the Pareto frontier of ultra-high-rate CSS code constructions, although the same method can be applied to other code families such as certain product codes and bivariate bicycle codes~\cite{bravyi2024high,liang2025self} (see Appendix~\ref{sec:zsz-halving}). Recall that the two CSS check matrices of the PP construction are specified by $J\times L$ arrays $E$ and $D$ of exponents over $\mathbb{Z}_P$, with each exponent lifted to a $P\times P$ circulant permutation matrix. We then have the following criterion:

\begin{Proposition}[CPM halving criterion]
    Suppose there is a permutation $\rho$ of the block rows, a fixed-point-free involution $\sigma$ of the block columns, an element $\eta\in\mathbb{Z}_P$ satisfying $\eta^2\equiv1\pmod P$, row offsets $\alpha_j$, and column offsets $\beta_\ell$ satisfying $\beta_{\sigma(\ell)}=-\eta\,\beta_\ell$. Here, $\alpha_j$ and $\beta_\ell$ are elements of $\mathbb{Z}_P$ added uniformly to all exponents in block row $j$ and block column $\ell$, respectively. If the exponent arrays of a CSS parent satisfy
\begin{equation}
    D_{j,\ell}=\eta\,E_{\rho(j),\sigma(\ell)}+\alpha_j+\beta_\ell \pmod P,
    \label{eq:halving_condition}
\end{equation}
then the parent admits symplectic halving under the qubit involution
\begin{equation}
    \pi(\ell,t)=\bigl(\sigma(\ell),\,\eta\,(t+\beta_\ell)\bigr),
    \label{eq:halving_involution}
\end{equation}
where $\ell$ labels a block column and $t\in\mathbb Z_P$ labels a
qubit within that block. The halved code has blocklength $LP/2$.
\end{Proposition}
\begin{proof}
    Applying $\pi$ twice gives
\[
    \pi^2(\ell,t)
    =\bigl(\ell,t+\beta_\ell+\eta\beta_{\sigma(\ell)}\bigr)
    =(\ell,t).
\]
Here we used $\eta^2=1$ and
$\beta_{\sigma(\ell)}=-\eta\beta_\ell$.
There are no fixed qubits because $\sigma$ has no fixed block columns.
A $Z$-check in block row $j$ with row coordinate $a$ meets block
column $\ell$ at $t=a-D_{j,\ell}$. Under $\pi$, its lift coordinate is
\[
    \eta(a-D_{j,\ell}+\beta_\ell)
    =\eta(a-\alpha_j)-E_{\rho(j),\sigma(\ell)}.
\]
Thus $\pi$ maps this check to the $X$-check in block row $\rho(j)$
with row coordinate $\eta(a-\alpha_j)$. These row maps are bijections,
so $\pi$ exchanges the two check spaces and gives a symplectic halving.
Its $LP/2$ qubit orbits give the stated blocklength.
\end{proof}

For prime $P$, $\eta^2\equiv1\pmod P$ implies $\eta=\pm1$.
For composite $P$, there may be additional square roots of unity and
hence additional dualities of the same form. We focus below on the
two universal choices $\eta=\pm1$.

The two choices give different involutions. For $\eta=+1$, the involution acts within each block by a translation, $\pi(\ell,t)=(\sigma(\ell),t+\beta_\ell)$, with antisymmetric offsets $\beta_{\sigma(\ell)}=-\beta_\ell$; we call this a \emph{plain involution}. For $\eta=-1$, it also reverses the lift index, $\pi(\ell,t)=(\sigma(\ell),-t-\beta_\ell)$, with symmetric offsets $\beta_{\sigma(\ell)}=\beta_\ell$; we call this a \emph{reverse involution}. The column gauge freedom $\tau_\ell$ of Eq.~\eqref{eq:cpmpp-gauge} acts on the offsets as $\beta_\ell\mapsto\beta_\ell+\tau_\ell-\eta\,\tau_{\sigma(\ell)}$, which preserves the required symmetry and can always be used to set $\beta=0$. Modulo gauge, the two cases read $D_{j,\ell}=\eta\,E_{\rho(j),\sigma(\ell)}$ and $\pi(\ell,t)=(\sigma(\ell),\eta\,t)$. Thus, a plain involution exchanges the paired block columns but maintains the cyclic direction, while a reverse involution additionally inverts each cyclic block. A given parent may admit involutions of both types, although we find that for some values of $J$ and $L$ (e.g. $(J,L)=(4,10)$), the plain involution constraints limit the achievable code distance.

The two involution types also differ in how the inherited cyclic symmetry acts after folding. For a plain fold, translation acts by a folded-qubit permutation, possibly combined with local Hadamards, and therefore preserves Pauli weight. For a reversing fold, the two Pauli components are translated in opposite directions, so their overlap and hence the folded Pauli weight can vary along a translation orbit. We analyze the resulting logical-orbit structure further in Sec.~\ref{sec:logical_profile}.

Rather than first searching for pair-partition CSS codes and then testing whether they happen to satisfy Eq.~\eqref{eq:halving_condition}, we impose the halving constraint at the beginning of the search. CSS orthogonality constrains the exponent differences $T_{i,j}(\ell)=E_{i,\ell}-D_{j,\ell}$ of the parent CSS code, requiring that for a given $i$ and $j$ the values across different $\ell$ occur in equal pairs, so that their circulants cancel. Taking $\rho$ to be the identity, which is always possible by relabeling the $Z$-check block rows, and substituting Eq.~\eqref{eq:halving_condition} gives
\begin{equation}
    T_{i,j}(\ell)=E_{i,\ell}-\eta\,E_{j,\sigma(\ell)}-\beta_\ell-\alpha_j \pmod P,
    \label{eq:halving_pair_partition}
\end{equation}
whose last term is constant across $\ell$ and therefore does not affect the pairing.
The cells $(i,j)$ and $(j,i)$ are related by $T_{j,i}(\ell)=-\eta\,T_{i,j}(\sigma(\ell))-\alpha_i-\eta\,\alpha_j$, and since even multiplicity survives both negation and a uniform shift, the second cell is determined by the first. Consequently, only $\binom{J}{2}$ off-diagonal cells require independent pair partitions, rather than all $J^2$ cells of the original construction. For a reverse involution $T_{i,i}(\ell)=T_{i,i}(\sigma(\ell))$ identically, so the diagonal orthogonality conditions require no additional equations. In either case, the $D$ array is derived from $E$, reducing the number of exponent variables from $2JL$ to $JL+J+L/2$ before accounting for gauge choices. The halving constraint thus defines a smaller direct search problem.

\begin{figure*}[!t]
    \centering
    \includegraphics[width=\textwidth]{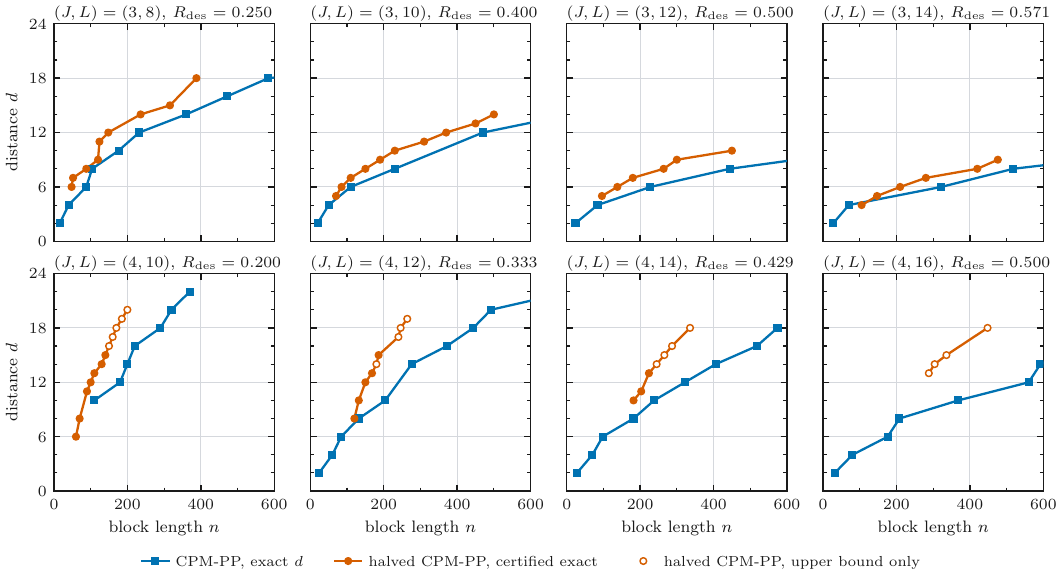}
    \caption{\textbf{Halving at matched degree parameters.}
    The halved CPM-PP frontier (red) against the CPM-PP frontier (blue) at
    matched $(J,L)$, where the halved construction improves the blocklength needed for a given code distance and rate, at the cost of being non-CSS.}
    \label{fig:halved-vs-css}
\end{figure*}

\subsection{Halved Pair-Partition Code Instances}

We instantiate this template by varying the degree parameters and lift size according to the design flow of Sec.~\ref{sec:construction}. We show the Pareto frontier for eight $(J,L)$ families in Fig.~\ref{fig:halved-frontier}, which contains compact instances with parameters such as $[[90,21,11]]$, $[[140,31,15]]$, and $[[200,43,20]]$. As before, increasing column weight improves the distance achievable at a given blocklength. Representative instances of each family are shown in Table~\ref{tab:noncss-code-catalogue}, and a more complete table can be found in Table~\ref{tab:noncss-code-catalogue-full}, see also the accompanying repository for details~\cite{designprinciplesdata}.

At fixed degree parameters, halved instances reach a given distance at a shorter blocklength than the corresponding CSS pair-partition instances. Figure~\ref{fig:halved-vs-css} shows this comparison, where we find a median best improvement at matched $(J,L)$ of 1.6, with the improvement more prominent for the $J=4$ family.

\begin{table*}[t]
\begin{ruledtabular}
\begin{tabular}{lcccccc}
Code & $k$ & Shots & Failures & $\varepsilon$ & $\varepsilon_{\mathrm{round}}$ & $\varepsilon_{\mathrm{round}}^{\mathrm{logical}}$ \\
\colrule
$[[186,65,14]]$ & 65 & $5.58\times 10^{8}$ & 0 & $<5.4\times 10^{-9}$ & $<2.7\times 10^{-10}$ & $<4.1\times 10^{-12}$ \\
$[[95,22,11]]$  & 22 & $3.73\times 10^{8}$ & 10 & $2.7^{+2.2}_{-1.4}\times 10^{-8}$ & $1.3^{+1.1}_{-0.7}\times 10^{-9}$ & $6.1\times 10^{-11}$ \\
\colrule
$[[121,1,11]]$ & 1 & $3.41\times 10^{8}$ & 38 & $1.1^{+0.4}_{-0.3}\times 10^{-7}$ & $5.6^{+2.1}_{-1.6}\times 10^{-9}$ & $5.6\times 10^{-9}$ \\
\end{tabular}
\end{ruledtabular}
\caption{Circuit-level memory simulations of two halved instances at physical
error rate $p=0.1\%$ over $N_{\mathrm r}=20$ syndrome extraction rounds, with a
distance-matched rotated surface code for comparison. $\varepsilon$ is the block
logical error rate per shot, $\varepsilon_{\textrm{round}}$, $\varepsilon_{\textrm{round}}^{\textrm{logical}}$ the corresponding per-round and per-logical-per-round error rates. All numbers use $95\%$ confidence intervals. The $[[186,65,14]]$ residual is empty, so its rates are $95\%$ upper bounds from the rule of three.}
\label{tab:noncss-circuit}
\end{table*}

\begin{figure}[t]
    \centering
    \includegraphics[width=\columnwidth]{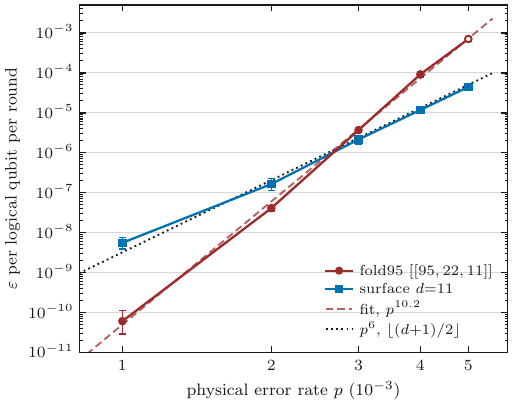}
    \caption{\textbf{Circuit-level $[[95,22,11]]$ memory comparison with the surface code.}
    Logical error rate per logical qubit per round over
    $N_{\mathrm r}=20$ rounds. We compare the halved CPM-PP $[[95,22,11]]$ with a hierarchical decoder (red), against a
    $d=11$ rotated surface code with belief matching~\cite{higgott2023improved1}
    (blue), under the same noise model.}
    \label{fig:fold95-ler}
\end{figure}

\subsection{Circuit-Level Simulations}

To test whether the parameter advantage survives practical noise, we perform circuit-level memory simulations of these non-CSS codes. We construct syndrome extraction schedules by searching for lowest-depth schedules that satisfy the commutation constraints. We use a circuit-level depolarizing noise model at physical error rate $p=0.1\%$, with depolarizing noise on two-qubit gates, resets, and measurements, no idling error, and noiseless state preparation and final measurement due to the lack of transversal initialization and measurement for non-CSS codes. We decode with a hierarchical decoder consisting of multiple tiers of BP and relay-BP~\cite{muller2025improved}, followed by integer programming fallback~\cite{zhao2026towards,bhardwaj2026high,delfosse2020hierarchical}.
Table~\ref{tab:noncss-circuit} reports results for two halved instances over $N_{\mathrm r}=20$ syndrome extraction rounds, together with the surface-code baseline, and Fig.~\ref{fig:fold95-ler} shows the $[[95,22,11]]$ instance across physical error rates.
We note that at $p=0.1\%$, where the decoder parameters were tuned at, all failures we observe for the $[[95,22,11]]$ occur at earlier tiers of the hierarchical decoder, indicating room for further improvement.
We also compare a $d=11$ rotated surface code simulated under the same noise model, the same $N_{\mathrm r}=20$, and the same noiseless preparation and readout, decoded by belief matching~\cite{higgott2023improved1}.
Despite using $28\times$ fewer qubits to encode the same number of logical qubits, the $[[95,22,11]]$ code achieves a better per-round logical error rate than the surface code at $p=0.1\%$.
The halved CPM-PP code also shows a steeper slope in this range, which may be partially due to its decoder being tuned at $p=0.1\%$.
These results suggest that the compact parameters of the halved codes survive noisy syndrome extraction, although a more careful accounting of state preparation and readout would be desirable in future work.

To close this section, we briefly comment on the architectural implications of using general stabilizer codes. CSS codes possess transversal $X$- and $Z$-basis preparation and readout, as well as transversal CNOT gates, which can be convenient for block-parallel operations. However, for selective gates, the most common approach is to utilize code surgery constructions~\cite{surgery2022, williamson2024low,swaroop2024universal,yuan2026parsimoniousquantumlowdensityparitycheck}, which can be performed on both CSS and non-CSS codes, presenting a promising approach to logical computation with high rate and distance at compact blocklengths. We leave detailed analysis of the surgery overheads and architectural trade-off comparisons to future work.

\begin{figure*}[!t]
    \centering
    \includegraphics[width=0.9\textwidth]{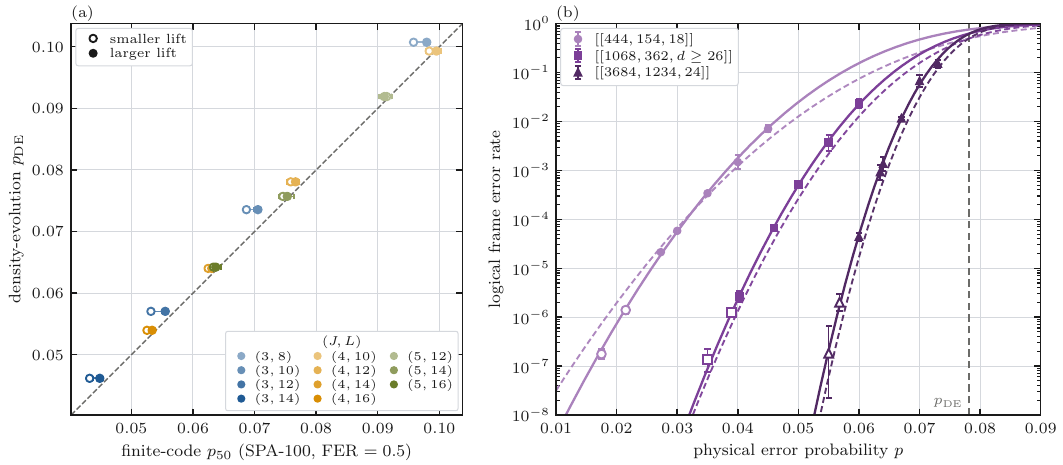}
    \caption{\textbf{Density evolution as a predictor of finite-code
    performance.}
    (a) CPM-PP physical error rate $p_{50}$ that results in $\textrm{LER}=0.5$, versus the regular
    ensemble density-evolution threshold, showing good agreement.  Each point
    pools three independently constructed matrices at fixed $(J,L,P)$.  Open
    and filled circles denote $P=97$ and $P=307$, respectively, and colors
    identify $(J,L)$.  The diagonal is $p_{50}=p_{\mathrm{DE}}$, and
    horizontal error bars are 95\% intervals.  (b) Code-capacity frame error
    rate for three different $(J,L)=(4,12)$ codes. Solid
    curves are Gaussian fits
    $\mathrm{FER}(p)=\Phi\bigl((p-p_{50})/\sigma\bigr)$ with $p_{50}$ and
    $\sigma$ fitted independently for each code.  Dashed curves use the common density evolution threshold
    $p_{\mathrm{DE}}$ and a single width prefactor $\sigma=\alpha/\sqrt{n}$
    shared by the three codes, fitted here with $\alpha=0.265$.
    The vertical line is the $(4,12)$ regular density-evolution heuristic
    $p_{\mathrm{DE}}=0.078$.}
    \label{fig:cpmpp-de}
\end{figure*}

\section{Ensemble Guidance for Degree Selection}
\label{sec:de}

The construction template specifies how stabilizer commutation is enforced, while the degree parameters provide further choices for code design.
At fixed design rate, increasing the average number of checks incident on each qubit also increases the average check weight.
Choosing these degrees therefore requires balancing their effects on decoding and distance.
Classical ensemble methods provide a starting point by studying codes with prescribed degree distributions rather than resolving every individual instance~\cite{richardson2001capacity,richardson2008modern,Gallager1960}.

We therefore study two asymptotic ensemble calculations inspired by their classical counterparts: A heuristic generalization of density evolution to quantum CSS codes gives the belief-propagation threshold under depolarizing code-capacity noise, which influences code properties in the high physical error waterfall regime. In turn, the ensemble weight-growth exponent gives an estimate of the typical minimum-distance scale, which influences code properties in the low physical error rate error-floor regime. These motivate the choices of degree distributions used in finite-size code constructions.
Both methods focus on the classical codes associated with a CSS quantum code and thus serve as heuristic guides that we compare concrete instances against.

\subsection{Quantum Density Evolution and BP Thresholds}

Density evolution is a well-established method for analyzing random LDPC codes under belief propagation (BP) decoding~\cite{richardson2001capacity,richardson2008modern,Gallager1960}. Instead of tracking individual messages on a particular code, it considers a random ensemble with a prescribed degree distribution. We associate a random message variable with each type of qubit-to-check and check-to-qubit connection and update its probability distribution using the BP rules. At each iteration, we approximate the neighborhood of a typical edge as tree-like, allowing us to treat incoming messages as independent. The asymptotic BP threshold is the largest channel error rate for which these distributions converge to a zero-error fixed point as the number of iterations grows.

We apply density evolution to the quantum problem by applying the same idea on the joint $X/Z$ decoding graph; see Appendix~\ref{sec:de_details} for details of its application to CPM-PP and lifted product codes. Before comparing with numerical results, it is worth noting some of the limitations of this method. First, our implementation of density evolution ignores CSS orthogonality and pair-partition constraints, which in particular may miss some of the local structure imposed by the CSS condition. Second, it does not provide accurate predictions in the regime of very low error rates (below the logical error rates in Fig.~\ref{fig:cpmpp-de}(b)), which may be dominated by distance or trapping sets. Finally, it only provides an approximation for BP convergence, which may not be representative of what can be achieved with more accurate decoders.

We compare the density-evolution prediction with simulations on the depolarizing code-capacity channel using joint-Pauli sum-product algorithm belief propagation. Let $p_{\mathrm{DE}}$ denote the population
density-evolution threshold and let $p_{50}(n)$ denote the physical error rate
at which a finite code of blocklength $n$ reaches logical frame error rate
$0.5$. The numerical definitions, simulation settings, and CPM-PP matrices
used in these comparisons are given in Appendix~\ref{sec:de-matrices}.

Figure~\ref{fig:cpmpp-de}(a) compares $p_{50}$ with $p_{\mathrm{DE}}$ across the degree distributions considered here and at two lift sizes, averaging over three randomly-drawn instances at each. The observed BP threshold shows good agreement with the density evolution predictions, with the agreement improving as the block size increases. We also show the logical error rate across a wider range of physical error rates in Fig.~\ref{fig:cpmpp-de}(b) and Appendix~\ref{sec:de-matrices}, where we find that the logical error rate across a wide range of instances is well-described by a single Gaussian function with width inverse-proportional to $\sqrt{n}$. These results indicate that density evolution provides a simple, predictive design tool to determine code capacity BP thresholds and waterfall-region logical error rates for quantum codes.

\subsection{Ensemble Distance Trends}
\label{sec:distance}

Density evolution characterizes the logical error rates in the waterfall regime close to the threshold, in which entropic effects dominate and the combination of many logical error contributions result in a Gaussian shape (Fig.~\ref{fig:cpmpp-de}(b)).
At low physical error rates, the logical error rate is instead dominated by minimum-weight logical operators with weight close to the distance or decoder trapping sets~\cite{pacenti2026trapping}, motivating the ensemble analysis of code distance.

To this end, we adapt Gallager's weight-enumerator argument~\cite{Gallager1960},
which we review in Appendix~\ref{sec:ensemble-distance-derivation}.  The ensemble estimate counts, on average, how many classical codewords occur at each weight and uses a saddle-point approximation to locate where this count changes from exponentially suppressed to exponentially large, providing a heuristic distance scale for the quantum codes.
As with the previous section, we apply this calculation to the classical code
associated with a CSS quantum code.  The resulting scale is therefore a heuristic guide: it does not capture the quotient by stabilizer
degeneracies or the permanent-based upper bounds specific to the CPM-PP
construction (Appendix~\ref{sec:cpmpp-structural-bounds}).

\begin{figure}[!t]
    \centering
    \includegraphics[width=\columnwidth]{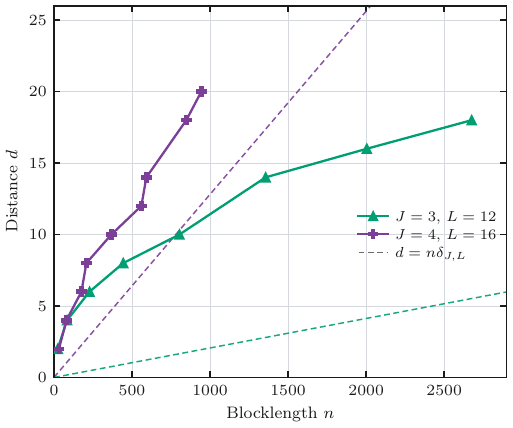}
    \caption{\textbf{Column weight at fixed design rate.}  Exact-distance
    Pareto frontiers of the two CPM-PP families whose design rate $1-2J/L$ is
    exactly $1/2$. Raising the column weight from
    $J=3$ to $J=4$ reaches $d=18$ at $n=848$ rather than $n=2676$, a factor of
    $3.2$ fewer qubits at the same rate, at the cost of check weight $16$
    instead of $12$.  Dashed lines are the regular-ensemble classical estimates
    $d=n\delta_{J,L}$.}
    \label{fig:distance-trend}
\end{figure}

We compare the exact distances of two CPM-PP code families with design rate $1/2$ and column weight $J=3, 4$ in Fig.~\ref{fig:distance-trend}.  Raising $J$ from 3 to 4 reaches distance 18 with a factor of 3.2 fewer qubits, consistent with the qualitative expectations of the ensemble estimate.
For $J=2$, each column of a binary parity-check matrix can be
viewed as an edge connecting two check nodes. The associated
classical code is a cycle code, whose minimum distance equals
half the Tanner girth. In a CSS code, the quantum distance
instead depends on cycles that are nontrivial modulo the
stabilizer space.

The classical ensemble calculation therefore motivates choosing
$J\geq3$ for the regular CPM-PP families studied here.
For basis-symmetric regular CSS codes, the design-rate relation
$R_{\mathrm{des}}=1-2J/L$ then gives $L\geq12$ at
$R_{\mathrm{des}}=1/2$.

While these results are qualitatively predictive, there is still a notable deviation between the linear ensemble distance estimate and the actual code distances, indicating the limitations of this method in predicting quantum code distances.

\section{Finite-Size Comparisons and Design Tradeoffs}
\subsection{Comparison Across Construction Families}
\label{sec:comparison}

\begin{figure*}[t]
    \centering
    \includegraphics[width=\textwidth]{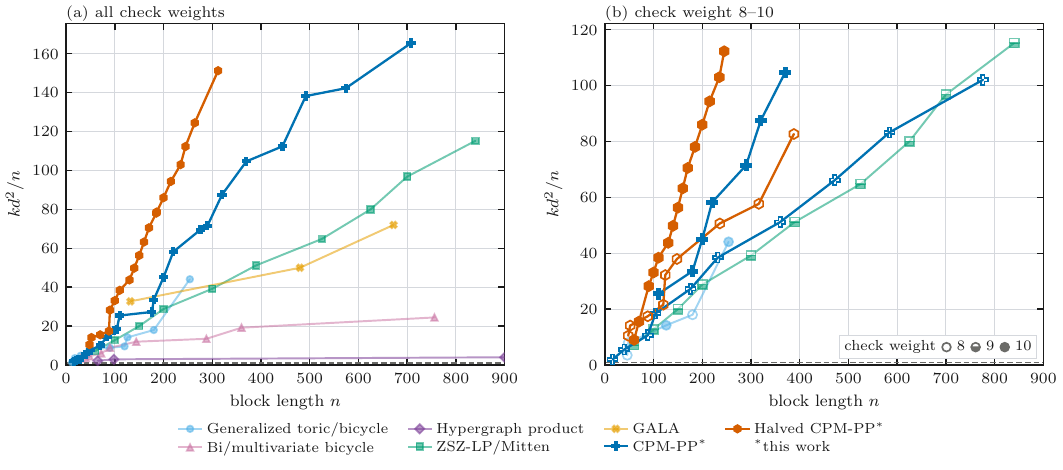}
    \caption{\textbf{Encoding efficiency across construction families.}
    $kd^2/n$ frontier as a function of blocklength $n$ for different code families~\cite{liang2025generalized,panteleev2019degenerate,bravyi2024high,voss2024multivariate,roffe2020decoding,pecorari2025high,hong2026quantum,bhardwaj2026high,yang2026designer}.
    (a) All check weights.
    (b) The same comparison restricted to check weights eight to ten, with the
    pair-partition templates $(3,8)$ and $(4,10)$ drawn as separate series and
    check weight distinguished by marker fill.
    The dashed line marks the rotated surface code, $kd^2/n=1$.
    Each code contributes the distance its source quotes, which is an upper
    bound wherever the distance is not yet exact; this applies to our
    families and to the comparison families alike.}
    \label{fig:family-kd2n}
\end{figure*}

We now compare our constructions against existing families, and show how the templates we design are able to significantly improve the Pareto frontier of code parameters.
Motivated by the BPT bound~\cite{bravyi2010tradeoffs}, we quantify the efficiency of our codes with the parameter $kd^2/n$, which is 1 for the surface code, and aim to minimize the blocklength required to achieve a given value.

The results are shown in Fig.~\ref{fig:family-kd2n}(a). Earlier bivariate bicycle codes achieve $kd^2/n$ of around 25. This is improved upon by the recent ZSZ-LP/Mitten code construction~\cite{bhardwaj2026high,hong2026quantum}, in which the parameter frontier continues growing as the blocklength increases, reaching $kd^2/n\sim 100$ around $n\sim 800$. In contrast, our constructions reach this frontier much earlier, with the CSS CPM-PP construction reaching $kd^2/n\sim 100$ around a blocklength of 370 and the halved version reaching it at blocklength just over 200. At blocklength 200, our best construction improves the $kd^2/n$ value by around a factor of 3. Also shown are the GALA codes~\cite{yang2026designer}, a recent family co-designed for reconfigurable neutral-atom arrays, whose distances are certified exactly; at check weight 12 they reach $kd^2/n=72$ at $n=672$ within this window.

The preceding comparison includes all check weight values, where a higher check weight is in general expected to lead to better parameters. We therefore specialize the comparison to codes with comparable check weight of 8-10 in Fig.~\ref{fig:family-kd2n}(b). We find that the advantages afforded by our construction templates persist: at matched check weight the halved-PP family reaches $kd^2/n=112$ at $n=245$ and the CSS family reaches $kd^2/n=105$ at $n=370$, whereas the non-Abelian lifted products need $n=840$ to reach $kd^2/n=115$.

\subsection{Tradeoff Between Distance and Check Weight}
\label{sec:check_weight}

The ensemble calculations suggest the data qubit degree as an important design parameter. At fixed design rate, increasing this degree requires heavier stabilizer checks and reduces thresholds, but can also support a larger distance at the same blocklength. The density-evolution analysis of Sec.~\ref{sec:de} predicts one cost of this choice: larger degrees reduce the belief-propagation code-capacity threshold. Under circuit-level noise, the penalty can be even stronger, since a heavier check requires more two-qubit gates and therefore introduces more fault locations during syndrome extraction.

To examine this tradeoff, we compare two CSS pair-partition codes with similar blocklengths and encoding rates but different degree profiles and distance: a $[[530,216,12]]$ code with check weight 10, constructed with $(J,L)=(3,10)$, and a $[[518,228,16]]$ code with check weight 14, constructed with $(J,L)=(4,14)$. We use a syndrome extraction schedule with the same depth as the check weight, found via block-wise scheduling of the protograph, and optimized for both circuit-level distance and 4-cycle count. We find a circuit-level distance upper bound of 10 and 14, respectively, two below the nominal code distances. We simulate 20 rounds of syndrome extraction at physical error rate $p=0.1\%$, using the same idling-free circuit-level noise model as Ref.~\cite{zhao2026towards}. We perform decoding using the hierarchical decoding approach described in Ref.~\cite{zhao2026towards,delfosse2020hierarchical,bhardwaj2026high}, which consists of fast belief propagation (BP), relay-BP~\cite{muller2025improved}, and integer programming stages, using the single basis detector error model (DEM). Utilizing the joint DEM can improve performance, at the cost of increased runtime and reduced convergence.

For the $[[530,216,12]]$ code, we observe 13 failures in $3.144\times 10^8$ shots, corresponding to a per-round logical error rate (LER) of $2.1^{+0.7}_{-0.6}\times10^{-9}$ and per-logical-per-round LER of $9.6^{+3.5}_{-2.6}\times10^{-12}$. All failures happen at earlier stages, although some of them persist even under integer programming decoding of the single basis DEM.
For the $[[518,228,16]]$ code, we observe 3 failures in $10^9$ shots, corresponding to a per-round LER of $1.5^{+1.5}_{-0.8}\times10^{-10}$ and per-logical-per-round LER of $7^{+6}_{-4}\times10^{-13}$, all occurring as time-outs of integer programming decoding.
All error bars quoted are $1\sigma$.
Despite its heavier checks, the $[[518,228,16]]$ code has a lower logical error rate, showing that the distance gain can outweigh the additional circuit faults in this operating regime.

This ordering can be understood from a simple threshold-versus-distance estimate. Suppose that the circuit-level threshold scales inversely with check weight, $p_{\mathrm{th}}\propto 1/w$, where $w$ denotes the check weight. Anchoring this relation to the effective threshold $p_{\mathrm{th}}\approx 0.7\%$ observed for the weight-9 mitten codes~\cite{bhardwaj2026high} gives an estimated threshold $p_{\mathrm{th}}^{(530)}\approx 0.63\%$ and $p_{\mathrm{th}}^{(518)}\approx 0.45\%$.

As a simple proxy, we use the leading distance scaling
\begin{equation}
    p_{\mathrm{L}}\propto
    \left(\frac{p}{p_{\mathrm{th}}}\right)^{d_{\mathrm{circ}}/2},
    \label{eq:check-weight-distance-scaling}
\end{equation}
where $p_{\mathrm{L}}$ is the per-round logical error rate and
$d_{\mathrm{circ}}$ the circuit-level distance.

At $p=0.1\%$, the two codes are factors of 6.3 and 4.5 below their
estimated thresholds, with estimated circuit-level distance 10 and 14.
Estimating the ratio between their logical error rates isolates the tradeoff,
\begin{equation}
    \frac{p_{\mathrm{L}}^{(530)}}{p_{\mathrm{L}}^{(518)}}\approx
    \frac{\left(1/6.3\right)^{5}}{\left(1/4.5\right)^{7}}\approx 3.8,
    \label{eq:check-weight-logical-ratio}
\end{equation}
where the two additional powers of suppression gain a factor $6.3^{2}\approx 40$, while the reduced threshold costs a factor $(6.3/4.5)^{7}\approx 11$. This is qualitatively consistent with the performance being better for the higher check weight code.
Under the same assumptions, the two estimates cross at $p_\times = [(p_{\mathrm{th}}^{(518)})^7/(p_{\mathrm{th}}^{(530)})^5]^{1/2}\approx 1.9\times 10^{-3}$, above which the ordering reverses and the lighter-check code is favored.
We observe that the integer programming stage of the decoder becomes very expensive at $p=0.2\%$, making direct verification of this estimate challenging.
The margin at $p=0.1\%$ is modest because the operating point sits only a factor of two below this crossing; the advantage of the higher-distance code grows rapidly as $p$ is reduced further.

Our analysis suggests that at commonly assumed physical error rates, an increased check weight may still be favorable in terms of logical error performance. However, it is worth emphasizing that there are many other factors that may favour a lower check weight:  increased check weights lead to a longer syndrome extraction cycle, which increases the physical runtime and space-time volume of the computation; it also increases the decoding time, which may become particularly challenging at higher physical error rates; finally, gadgets to perform logical operations may lead to a further increase in the check weight in practice.

\section{Designing Low Weight Canonical Bases}
\label{sec:logical_profile}

The preceding sections establish favorable encoding rates, distances, and circuit-level memory performance for the proposed codes.
These results do not, however, specify how the logical operators are organized or how costly they are to manipulate.
Addressing these questions is an essential step to go from a quantum memory to a processor.
In high-rate LDPC codes, logical operations can be implemented through generalized lattice surgery, whose resource requirements depend on the weight of the logical operator being addressed~\cite{surgery2022,williamson2024low,swaroop2024universal, yuan2026parsimoniousquantumlowdensityparitycheck}.
This motivates finding a complete, canonically paired logical basis with low-weight physical representatives, together with physical symmetries that permit surgery gadgets to be reused across different logical qubits~\cite{webster2025explicitconstructionlowoverheadgadgets,zheng2026logical,bhardwaj2026high,blue2026extractorslogicalprocessinghypergraph}. For example, such representatives have been developed for two-block group-algebra and bivariate bicycle codes~\cite{eberhardt2024logicaloperatorsfoldtransversalgates}.

In this section, we develop a toolbox for constructing symmetry-aware logical bases and optimizing their physical weights in any group-valued CSS code.
We first decompose the logical space into representation-theoretic sectors, called \emph{packets}, reducing the logical operator calculation to smaller problems over finite fields~\cite{panteleev2022quantum,lee2026logicalspectroscopyliftedproductcodes}.
This decomposition also identifies which sectors pair under logical commutation (Sec.~\ref{subsec:canonical-orbit-bases}).
We then provide a procedure to recombine packet logicals into canonically paired families generated by translated logical seeds.

This symmetry-aware organization provides a structured starting point for low-weight searches by exposing seed choices and packet recombinations compatible with canonical pairing.
We optimize these choices together with the physical representatives within stabilizer cosets.
In the examples studied below, this approach leads to lower-weight canonical bases than generic binary basis construction followed by stabilizer-coset minimization.
We call the resulting symmetry decomposition, orbit-span ranks, commutation pairings, and representative weights the \emph{logical profile}.
The construction does not require a lifted-product factorization: we demonstrate it on PP codes (Sec.~\ref{subsec:pp-logical-profile}) and extend the analysis to their halved non-CSS descendants (Sec.~\ref{subsec:halved-pp-profile}), where folding changes the commutation form and can alter physical weights along translation orbits.

\subsection{Logical modules and packet decomposition}

We now describe how the algebraic structure of the lift group $G$ organizes the search for a logical basis.
Let $G$ be a finite group whose right-regular action on every qubit and check block preserves the CSS incidence structure, and let $R_G=\F_2[G]$.
Identifying each regular $G$-orbit with $G$ gives a chain complex of free right $R_G$-modules,
\begin{align}
R_G^{r_Z}\xrightarrow{\ H_Z^\dagger\ }R_G^L\xrightarrow{\ H_X\ }R_G^{r_X},\qquad H^{\vphantom{\dagger}}_X  H_Z^\dagger &= 0.\label{eq:group-css-complex}
\end{align}
We emphasize $H_X$ and $H_Z$ are matrices over $R_G$. Here $L$ counts the regular qubit orbits, while $r_X$ and $r_Z$ count the regular orbits of $X$- and $Z$-check generators.
Vectors are columns, check-matrix entries act from the left, and $G$ acts by right translation.
The adjoint is $A^\dagger=(A^*)^{\mathsf T}$, where $*$ extends group inversion $g\mapsto g^{-1}$ linearly.

Since the check maps are right $R_G$-linear, their kernels and images are invariant under translation.
The logical spaces therefore inherit right $R_G$-module structures,
\begin{align}
{\cal L}_Z &= \ker H_X/\operatorname{im}H_Z^\dagger, & {\cal L}_X &= \ker H_Z/\operatorname{im}H_X^\dagger.\label{eq:logical-modules-short}
\end{align}
These quotients can be computed directly over $R_G$, but the presence of zero divisors means that ordinary Gaussian elimination cannot be applied unchanged\footnote{Replacing each group element by its left-regular binary permutation matrix gives the binary check matrices $\mathsf{H}_X$ and $\mathsf{H}_Z$.
This permits binary linear algebra, although a generic binary logical basis need not expose the module structure explicitly.}.

For Abelian $G$ of odd order, the commutative semisimple algebra $R_G$ decomposes into finite fields via
\begin{align} \label{eq:packet-decomposition-short}
\Phi:R_G &\xrightarrow{\sim}\prod_\Omega K_\Omega,\qquad K_\Omega\cong\F_{2^{f_\Omega}}.
\end{align}
The field factors are indexed by Frobenius orbits of characters, called \emph{packets} $\Omega$, and their sizes $f_\Omega=|\Omega|$ determine the corresponding field degrees~\cite{lee2026logicalspectroscopyliftedproductcodes}.
Write $\Phi_\Omega$ for the projection onto $K_\Omega$.
We identify $R_G=\F_2[u]/(u^P-1)$ and factor $u^P-1=\prod_\Omega g_\Omega(u)$ into irreducible polynomials.
The Chinese Remainder Theorem realizes the packet map as
\begin{align}
\Phi(h) &= \bigl(h\bmod g_\Omega\bigr)_\Omega,\quad K_\Omega=\F_2[u]/(g_\Omega(u)), \label{eq:crt-map-short}
\end{align}
Multiplication by $u$ implements cyclic translation, whose action in packet $\Omega$ is multiplication by $\alpha_\Omega:=\Phi_\Omega(u) \in K_\Omega$.

Applying $\Phi_\Omega$ entrywise to the check maps gives a complex over $K_\Omega$.
Using $A_\Omega:=\Phi_\Omega(A)$, its $Z$-logical homology is $\calL_{Z,\Omega}=\ker H_{X,\Omega}/\operatorname{im}(H_Z^\dagger)_\Omega$, with dimension
\begin{align}
m_{Z,\Omega} &= L-\rank_{K_\Omega}H_{X,\Omega}-\rank_{K_\Omega}(H_Z^\dagger)_\Omega. \label{eq:packet-rank-short}
\end{align}
Here, the adjoint $\dagger$ transposes the check matrix and replaces the shift variable $u$ by $u^{-1}$.
Evaluating at $u=\alpha_\Omega$ therefore gives $(H_Z^\dagger)_\Omega=H_Z(\alpha_\Omega^{-1})^{\mathsf T}$.
Let $m_\Omega := m_{Z,\Omega}$. The map $\Phi$ therefore induces the logical decomposition $\calL_Z \simeq \oplus_\Omega \calL_{Z,\Omega}$, where each $\calL_{Z,\Omega} \simeq K_\Omega^{m_\Omega}$.

Note that the multiplicity $m_\Omega$ counts independent logical vectors over $K_\Omega$, not over $\F_2$.
Choose a basis $\{ \ell_{i,\Omega} \}_{i=1}^{m_\Omega}$ of $\calL_{Z,\Omega}$ over $K_\Omega$.
Every packet logical class is a unique combination $\sum_i c_i \ell_{i,\Omega}$ with coefficients $c_i\in K_\Omega$.
In the cyclic case, $1,\alpha_\Omega,\ldots,\alpha_\Omega^{f_\Omega-1}$ form an $\F_2$-basis of $K_\Omega$, so each coefficient has a unique binary expansion $c_i = \sum_{j=0}^{f_\Omega-1}b_{ij}\alpha_\Omega^j$ with $b_{ij}\in\F_2$.
Consequently, the $m_\Omega f_\Omega$ vectors $\alpha_\Omega^j \ell_{i,\Omega}$, with $1\leq i\leq m_\Omega$ and $0\leq j<f_\Omega$, form a binary basis of $\calL_{Z,\Omega}$.
For each fixed $i$, these vectors are proportional over $K_\Omega$ but linearly independent over $\F_2$.
Thus the packet contains $m_\Omega$ independent field directions, each carrying $f_\Omega$ independent binary logical directions, and the total number of logical qubits is $k=\sum_\Omega f_\Omega m_\Omega$.
The $X$-logical calculation follows by interchanging the checks.

Accordingly, for logical seeds $v_1,\ldots,v_s$, let $v_{i,\Omega}=\Phi_\Omega(v_i)$ and let $r_\Omega$ be the dimension over $K_\Omega$ of their span in packet $\Omega$.
The binary dimension generated by all translates of these seeds is $\sum_\Omega f_\Omega r_\Omega$.
In particular, a single seed generates a \emph{regular logical orbit}, consisting of $|G|$ independent translates modulo stabilizers, exactly when it has a nonzero component in every packet.

The inverse map uses idempotents $e_\Omega\in R_G$ satisfying $\Phi_\Lambda(e_\Omega)=\delta_{\Lambda\Omega}$, so multiplication by $e_\Omega$ retains only packet $\Omega$.
For each $a_\Omega\in K_\Omega$ with its polynomial representative $a_\Omega(u)$ of degree less than $f_\Omega$, we have
\begin{align}
\Phi^{-1}\bigl((a_\Omega)_\Omega\bigr)
&= \sum_\Omega a_\Omega(u) e_\Omega(u)
\pmod{u^P-1}.\label{eq:inverse-packet-short}
\end{align}
Applying this map entrywise to packet kernel representatives reconstructs logical representatives in binary qubit coordinates.

\subsection{Canonical orbit bases and physical weights}
\label{subsec:canonical-orbit-bases}

Building upon packet decomposition, we now show how to construct canonically paired logical operators generated from a few seeds using group symmetry, and how to optimize their physical representatives while preserving this pairing.
For group-valued logical representatives $x,z\in R_G^L$, write $B(x,z):=\widehat{x}^{\mathsf T}\widehat{z}\in\F_2$ for their binary commutation pairing, which is well defined modulo stabilizers.
We call logical bases $\{x_a\}_{a=1}^k$ and $\{z_a\}_{a=1}^k$ \emph{canonically paired} when $B(x_a,z_b)=\delta_{ab}$.
Thus each $X$-logical basis operator anticommutes with its matching $Z$ operator and commutes with all other $Z$-basis operators.

We first show that translation invariance restricts logical commutation to reciprocal packets~\cite{lee2026logicalspectroscopyliftedproductcodes}.
Since the group action permutes physical qubits, $B(xg,z)=B(x,zg^{-1})$.
Extending this identity linearly gives $B(xa,z)=B(x,za^*)$ for $a\in R_G$, where $*$ is the group-algebra involution induced by inversion.
Let $e_\Omega$ denote the projector onto packet $\Omega$.
Group inversion sends this projector to that of the reciprocal packet, so $e_\Omega^*=e_{\Omega^{-1}}$.
Consequently,
\begin{align}
B(xe_\Omega,ze_\Lambda)
&= B(x,ze_\Lambda e_{\Omega^{-1}}) = 0,  \quad \forall \Lambda \neq \Omega^{-1}.
\end{align}
Because the full logical pairing is nondegenerate and no other packets can pair, its restriction to each reciprocal packet pair is also nondegenerate.
In particular, $m_{X,\Omega^{-1}}=m_{Z,\Omega}$, and dual generators can be constructed independently within each reciprocal packet pair~\cite{lee2026logicalspectroscopyliftedproductcodes}.

\vspace{5pt} \noindent \emph{Recombination into regular orbits}. Each regular logical orbit supplies one field direction in every packet.
The maximum number of mutually independent regular orbits on each Pauli side is therefore
\begin{align}
q &= \min_\Omega m_{Z,\Omega} = \min_\Omega m_{X,\Omega}.
\label{eq:canonical-regular-existence}
\end{align}
For each $\Omega$, choose representatives $x_{i,\Omega^{-1}}\in\ker H_{Z,\Omega^{-1}}$ of $q$ independent $X$-logical classes over $K_{\Omega^{-1}}$.
Find their partners $z_{j,\Omega}\in K_\Omega^L$ by solving
\begin{align}
H_{X,\Omega}z_{j,\Omega} &= 0, \qquad (x_{i,\Omega^{-1}}^*)^{\mathsf T}z_{j,\Omega} = \delta_{ij}.
\label{eq:packet-dual-system}
\end{align}
Here $*:K_{\Omega^{-1}}\to K_\Omega$ is the field isomorphism induced by group inversion.

Once $q$ canonical pairs across all packets are obtained, for each seed index $i \in \{1,...,q\}$, apply the inverse packet map in Eq.~\eqref{eq:inverse-packet-short} to obtain $q$ pairs of group-valued logical representatives $x_i,z_i\in R_G^L$:
\begin{align}
x_i &= \Phi^{-1}\bigl((x_{i,\Omega})_\Omega\bigr), \quad z_i = \Phi^{-1}\bigl((z_{i,\Omega})_\Omega\bigr).
\label{eq:canonical-seed-recombination}
\end{align}
Eq.~\eqref{eq:packet-dual-system} makes every packet component of $C(x_i,z_j):=x_i^\dagger z_j$ equal to $\delta_{ij}$.
Since the packet map is an algebra isomorphism, $C(x_i,z_j)=\delta_{ij}1$; see Appendix~\ref{app:pairing-correlations}.
Since the coefficients of this correlation record the commutation bits at all relative translations, this implies
\begin{align}
B(x_i g,z_j h)
&= \delta_{ij}\delta_{g,h}.
\label{eq:canonical-orbit-pairing}
\end{align}
Thus the recombined seeds generate $q|G|$ canonical logical pairs, arranged into $q$ regular orbits on each Pauli side.
These pairing relations also certify independence modulo stabilizers.
The remaining $k-q|G|$ logical pairs are chosen in the pairing-orthogonal complements of the regular submodules.

\vspace{5pt} \noindent \emph{Physical weights}. First we note that translating a logical \emph{seed} generates an entire family of equal-weight representatives without changing canonical pairing. Therefore, by optimizing $q$ sets of canonical logical seeds, one can achieve the minimum average weights among $q|G|$ canonical logical pairs. To proceed, for each packet $\Omega$, let $U_\Omega=[x_{1,\Omega} \cdots x_{q,\Omega}]$ and $V=[z_{1,\Omega} \cdots z_{q,\Omega}]$, so $U_\Omega^\dagger V_\Omega^{\vphantom{\dagger}}=I_q$
The simultaneous basis changes $U_\Omega\mapsto U_\Omega A_\Omega$ and $V_\Omega \mapsto V_\Omega(A_\Omega^\dagger)^{-1}$ preserve canonical pairing within the fixed regular submodules for every $A_\Omega\in\operatorname{GL}_q(R_G)$.
More generally, varying $U_\Omega$ and resolving Eq.~\eqref{eq:packet-dual-system} can also change the selected regular submodules.
The physical weights of both seed families are evaluated after CRT recombination, where contributions from different packets can cancel, and their representatives are further optimized by stabilizer dressing.
This separates the choice of canonically paired logical generators from weight reduction within their fixed stabilizer cosets.

\vspace{5pt}\noindent\emph{Even lift sizes.}
For even $P$, logical profiling remains possible, but the finite-field packet decomposition must be replaced by the primary-packet decomposition described in Appendix~D of Ref.~\cite{lee2026logicalspectroscopyliftedproductcodes}.
Its factors are local rings rather than fields and retain the repeated-root structure lost under ordinary packet evaluation.
We discuss the resulting logical modules and canonical pairing in Appendix~\ref{sec:even_P_canonical}.

\subsection{Application: Pair-partition codes}
\label{subsec:pp-logical-profile}

We now apply the preceding methods to the CPM-PP codes introduced in Sec.~\ref{sec:cpmpp}. The pair-partition construction gives a simple packet multiplicity profile.
Consider a fully populated $(J,L)$ PP array over an odd-order Abelian group $G$ and assume that both check matrices have row rank $J$ in every nontrivial packet.
Each array entry is a group element, so both check matrices evaluate to an all-one matrix of rank one at the trivial packet.
Equation~\eqref{eq:packet-rank-short} therefore gives, on either Pauli side,
\begin{align}
m_{\mathrm{triv}} = L-2,\quad m_\Omega=L-2J\quad(\Omega\neq\mathrm{triv}). \label{eq:pp-packet-count-short}
\end{align}
With the assumption, the previous packet recombination implies that, for either logical space $\calL=\calL_X$ or $\calL_Z$,
\begin{align}
\calL &\cong R_G^{\oplus(L-2J)}\oplus\one^{\oplus(2J-2)}. \label{eq:strict-pp-profile-short}
\end{align}
where $\one$ denotes the one-dimensional trivial module.
Each regular summand combines one field direction from every packet, leaving $2J-2$ translation-invariant logical directions.
The preceding construction provides canonically paired regular orbits, while their physical weights remain to be optimized.
For $(J,L)=(3,8)$, the decomposition becomes $\calL \cong R_G^2\oplus\one^4$.

\begin{figure}[!t]
\centering
\includegraphics[width=0.9\columnwidth]{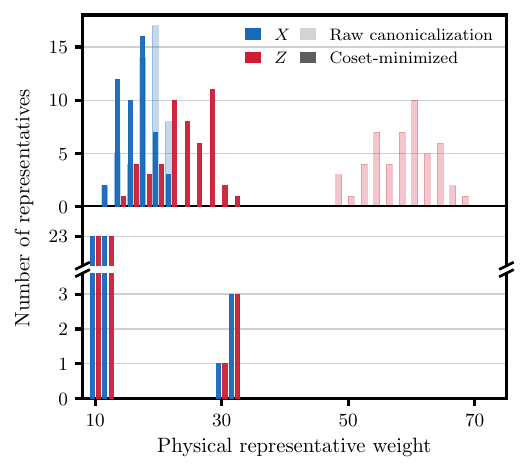}
\caption{\textbf{Logical profile of $[[184,50,10]]$.}
(Top) weight-ordered Gaussian-elimination initialization followed by Gram canonicalization, with the initial $X$ basis held fixed, gives a mean physical weight of $38$.
Coset minimization reduces this to $20$.
(Bottom) Packet-guided canonical basis optimization and coset minimization give $23$ regular logical pairs of weights $(10,12)$ and $23$ of weights $(12,10)$.
The four residual pairs have weights $(30,32)$, $(32,30)$, $(32,32)$, and $(32,32)$, yielding an overall mean weight of $12.6$.
Each displayed basis satisfies $\langle X_i,Z_j\rangle=\delta_{ij}$.
}
\label{fig:p23-logical-profile}
\end{figure}

The $[[184,50,10]]$ PP code with $(J,L,P)=(3,8,23)$ illustrates the construction. For $P=23$, Frobenius squaring returns each nontrivial $23$rd root of unity to itself after eleven steps, since $\operatorname{ord}_{23}(2)=11$.
The $22$ nontrivial roots therefore form two degree-eleven packets.
Root inversion exchanges these packets because $-1$ is not a power of $2$ modulo $23$.
Together with the trivial packet, this gives
\begin{align}
R_{C_{23}} &\cong \F_2\oplus K_+\oplus K_-,\qquad K_\pm\cong\F_{2^{11}}. \label{eq:P23_module}
\end{align}
Both check matrices have packet ranks $(1,3,3)$, giving logical multiplicities $(6,2,2)$ and binary dimensions $(6,22,22)$. Consequently,
\begin{align}
\calL &\cong \F_2^6\oplus K_+^2\oplus K_-^2\cong R_{C_{23}}^2\oplus\one^4. \label{eq:p23-recombination-short}
\end{align}
The logical basis can therefore be organized into two regular length-$23$ orbits and four translation-invariant classes on each Pauli side.

\vspace{5pt}\noindent\emph{Canonical regular pairing.}
For this code, a check-exchanging qubit involution $\pi$ satisfies $\pi T=T^{-1}\pi$, where $T$ is cyclic translation, and every $T^t\pi$ is fixed-point free.
These properties allow both $Z$ seeds to be obtained from two $X$-logical seeds $a,b$. Choose the latter to satisfy
\begin{align}
B(b,T^t\pi a)
&=\delta_{t,0},\qquad 0\leq t<P.
\label{eq:pp-canonical-seed-search}
\end{align}
The paired seeds $(x_1,z_1)=(a,\pi b)$ and $(x_2,z_2)=(b,\pi a)$ then generate canonical translated families.
Indeed, the fixed-point-free involution identities
$B(v,T^t\pi v)=0$ and
$B(v,T^t\pi w)=B(w,T^t\pi v)$
give $B(T^r x_i,T^s z_j)=\delta_{ij}\delta_{r,s}$.
Optimizing the seeds subject to these constraints yields verified representatives of weights $10$ and $12$, hence $23$ canonical pairs of weights $(10,12)$ and $23$ of weights $(12,10)$.
These regular weights are optimal under the full-translation-orbit constraint.
An exhaustive census of minimum-weight logical classes excludes canonical pairing between two weight-$10$ translation orbits.
Moreover, each check matrix has column degree three, so its row sum is the all-one vector and every normalizer word has even weight.
Together with $d=10$, this establishes a minimum maximum seed weight of $12$ and a minimum paired seed-weight sum of $22$, both attained by the construction. See Appendix~\ref{sec:explicit_canonical_illustration} for details.

 \vspace{5pt}\noindent\emph{Residual pairing and complete basis.}
The four remaining logical directions on each side are chosen orthogonal to the regular families.
They are translation-invariant as logical classes, although their physical representatives need not be pointwise invariant.
We search for low-weight representatives of all nonzero residual $X$ classes and obtain the corresponding $Z$-side weight table by reflection.
Selecting dual bases from these tables gives canonical pairs with representative weights
$(30,32)$, $(32,30)$, $(32,32)$, and $(32,32)$.

\vspace{5pt} \noindent \emph{Comparison}. We compared $100$ initializations per method, each given $30$ CPU seconds for initialization, complete-basis construction, and coset minimization.
The mean physical weight over the $2k=100$ representatives in each basis was $23.84 \pm 1.33$ for binary Gram canonicalization~\cite{Wilde2009_gram} with fixed-class coset minimization, versus $14.91 \pm 2.02$ for the packet-guided procedure.
The best complete-basis means were $20.26$ and $12.94$, respectively; Fig.~\ref{fig:p23-logical-profile} shows their weight distributions.
In a separate longer-budget benchmark, the packet-guided search recovered the regular profile $(10,12)^{23},(12,10)^{23}$ in all tested runs, while naive approach stayed stuck.

\subsection{Application: Halved pair-partition codes}
\label{subsec:halved-pp-profile}

Sec.~\ref{sec:noncss} showed that symplectic halving identifies the folded logical space with the parent logical module, while pulling the parent CSS pairing back to an internal symplectic form.
We now determine how the parent lift symmetry is transported through this identification and use it to construct low-weight canonical logical frames. The packet multiplicities are inherited from the parent, whereas the halving involution determines how these packets pair and how physical Pauli weights behave along translation orbits.

\vspace{5pt}\noindent\emph{Non-CSS canonical pairing.}
Let $\pi$ be the fixed-point-free, check-exchanging involution defining the folding map $F_\pi$ in Eq.~\eqref{eq:folded-parent-map}.
Since $\pi$ preserves the parent binary dot product and exchanges $S_X$ with $S_Z$, it also satisfies $\pi(S_Z^\perp)=S_X^\perp$.
It therefore maps both the numerator and denominator of $\calL_X=S_Z^\perp/S_X$ onto those of $\calL_Z=S_X^\perp/S_Z$, inducing an isomorphism $\pi:\calL_X\to\calL_Z$.
Eq.~\eqref{eq:folded-stabilizer-normalizer} shows that $[a]\mapsto[F_\pi(a)]$ induces the quotient isomorphism
\begin{align}
\mathcal Q_{\mathrm{fold}}:=N_{\mathrm{fold}}/S_{\mathrm{fold}} &\cong S_Z^\perp/S_X=\calL_X. \label{eq:folded-parent-quotient-isomorphism}
\end{align}
Thus $2k_{\mathrm{fold}}=k_{\mathrm{parent}}$: the inherited binary dimension counts folded logical Pauli directions, not folded logical qubits.

Under this identification, Eq.~\eqref{eq:folded-twisted-pairing} gives the folded commutation form
\begin{align}
B_\pi(a,b)
&:=B_{\mathrm{sp}}(F_\pi(a),F_\pi(b)) =a^{\mathsf T}\pi(b).
\label{eq:folded-pulled-back-form}
\end{align}
Thus folded commutation is computed by pairing a parent word with the reflected second word.
The induced isomorphism $\pi:\calL_X\to\calL_Z$ makes $B_\pi$ well defined and nondegenerate.
It is alternating because each $\pi$-pair contributes twice to $B(a,\pi a)$, giving zero over $\F_2$.

A canonical folded frame therefore consists of logical classes satisfying
\begin{align}
B_\pi([a_i],[a_j])
&=B_\pi([b_i],[b_j])=0, \nonumber\\
B_\pi([a_i],[b_j])
&=\delta_{ij}.
\label{eq:canonical-frame-conditions}
\end{align}
Each family must be \emph{isotropic}, meaning that its operators mutually commute.
Their physical representatives are $\overline X_i=F_\pi(a_i)$ and $\overline Z_i=F_\pi(b_i)$.
These labels specify canonical roles; both representatives can contain mixed Pauli components.
All weights below use the folded Pauli weight of Eq.~\eqref{eq:folded-pauli-weight}.

\vspace{5pt}\noindent\emph{Inherited translation and packet pairing.}
Let $T$ be a parent qubit automorphism of odd order $P$ preserving both check spaces.
The induced action $K=F_\pi T F_\pi^{-1}$ preserves the folded stabilizer and normalizer spaces, and therefore acts on $\mathcal Q_{\mathrm{fold}}$ with the same packet multiplicities as the parent $\calL_X$.
We call the fold \emph{plain} when $\pi T\pi=T$ and \emph{reversing} when $\pi T\pi=T^{-1}$.
The adjoint relation
\begin{align}
B_\pi([Ta],[b])
&=B_\pi([a],[\pi T^{-1}\pi b])
\label{eq:translation-adjoint}
\end{align}
makes $K$ symplectic in the plain case and self-adjoint in the reversing case.

For a plain fold, parent translation maps each folded qubit pair to another folded qubit pair.
Therefore, the induced action $K$ acts by a folded-qubit permutation, possibly accompanied by local Hadamards when the chosen $X$ and $Z$ endpoints are exchanged, and preserves physical Pauli weight.
For a reversing fold, suppose $\pi$ exchanges two distinct parent translation orbits.
This is automatic for odd $|G|$, since a reversal within a single odd-length orbit would have a fixed point.
Label a parent translation orbit by $q_t=T^tq_0$ and assign $q_t$ and $\pi(q_t)$ as the $X$ and $Z$ endpoints of folded qubit $t$, respectively.
The relation $\pi T\pi=T^{-1}$ means that translation sends $q_t$ to $q_{t+1}$ but $\pi(q_t)$ to $\pi(q_{t-1})$.
Thus an $X$ component at folded qubit $t$ moves to $t+1$, whereas a $Z$ component moves to $t-1$.
The two components of a $Y$ operator can therefore separate onto different qubits, changing the support overlap and hence the physical Pauli weight.
Moreover, a nontrivial reversing-fold logical action is nonsymplectic for odd $P$, so it is an algebraic action rather than a commutation-preserving physical symmetry.

The packet pairing also depends on the involution.
A plain $\pi$ preserves packet labels, retaining the parent's reciprocal-packet pairing.
A reversing $\pi$ exchanges reciprocal labels, so the pulled-back pairing is internal and nondegenerate within each packet.
Table~\ref{tab:fold-comparison} summarizes these differences.

\begin{table}[t]
\centering
\begin{tabular}{lll}
\hline
Property & Plain & Reversing \\
\hline
Logical action & Symplectic & Self-adjoint \\
Pauli weight & Preserved & Can vary \\
Packet pairing & $\Omega\leftrightarrow\Omega^{-1}$ & $\Omega\leftrightarrow\Omega$ \\
Orbit isotropy & Must be checked \,\, & Automatic \\
Cross-Gram matrix\,\, & Circulant & Reverse-circulant \\
\hline
\end{tabular}
\caption{Inherited translation for plain and reversing folds.}
\label{tab:fold-comparison}
\end{table}

\vspace{5pt}\noindent\emph{Canonical orbit families.} For folded logical seeds $p,q$, define the binary cross-Gram matrix by $M_{rs}=B_{\mathrm{sp}}(K^rp,K^sq)$.
The adjoint relation gives
\begin{align}
M_{rs}
&=
\begin{cases}
B_{\mathrm{sp}}(p,K^{s-r}q), & \text{plain},\\
B_{\mathrm{sp}}(p,K^{r+s}q), & \text{reversing}.
\end{cases}
\label{eq:folded-orbit-pairing}
\end{align}
For a plain fold, orbit isotropy requires $B_{\mathrm{sp}}(p,K^tp)=0$ for every $t$. This holds automatically for seeds contained in a non-self-reciprocal packet, but must otherwise be imposed.
For a reversing fold, it is automatic:
\begin{align}
0 &=B_{\mathrm{sp}}(K^rp,K^rp) =B_{\mathrm{sp}}(p,K^{2r}p).
\label{eq:folded-self-gram}
\end{align}
Since $P$ is odd, $t\in\mathbb Z_P$ can be written as $t=2r$ modulo $P$, which gives $B_{\mathrm{sp}}(p,K^tp)=0$ for every $t$.

When both orbit families are isotropic and $M$ is invertible, their $2P$ vectors span a nondegenerate symplectic subspace.
Binary basis changes acting on the rows $A,B\in\operatorname{GL}_P(\F_2)$ give a canonical frame precisely when
\begin{align}
A M B^{\mathsf T} &= I_P. \label{eq:folded-orbit-canonicalization}
\end{align}
These two orbit families constitutes $P$ folded logical pairs.
General recombination can change their physical weights and need not retain the translated form of the representatives.

\vspace{5pt}\noindent\emph{A common P13 parent.}
The $[[104,30,8]]$ PP code with $(J,L,P)=(3,8,13)$ admits both types of involution.
The plain and reversing folds have exact parameters $[[52,15,7]]$ and $[[52,15,6]]$, respectively, with check weights at most eight.
Both inherit the $R_G$-module
\begin{align}
\mathcal Q_{\mathrm{fold}} &\cong R_{C_{13}}^2\oplus\one^4. \label{eq:p13-folded-module}
\end{align}
The unique nontrivial degree-$12$ packet is self-reciprocal because $2^6\equiv-1\pmod{13}$.
Thus both folds pair internally within each packet; their distinction lies in the translation action and physical weight function.

For the plain fold, optimization finds two rank-$13$ isotropic orbits with seed weights eight and nine.
Their cross-Gram matrix is a cyclic permutation matrix, so a cyclic relabeling gives representatives $p,q$ satisfying
\begin{align}
\overline X_i&=K^ip,\quad
\overline Z_i=K^iq, \quad B_{\mathrm{sp}}(\overline X_i,\overline Z_j) =\delta_{ij}.
\label{eq:p13-folded-canonical-orbits}
\end{align}
Weight preservation gives thirteen canonical pairs of uniform weights $(8,9)$.
Their four-dimensional symplectic complement is fixed by logical translation and forms two additional pairs.
The achieved representative-weight profile is
\begin{align}
\bigl((8,9)^{\times13},(17,13)^{\times2}\bigr),
\label{eq:p13-folded-weight-profile}
\end{align}
where superscripts indicate pair multiplicities.

The regular mean operator weight $8.5$ is optimal among two rank-$13$ isotropic translation orbits with a cyclic-permutation cross-Gram matrix.
Exhaustive low-weight enumeration shows that every weight-$7$ logical orbit is nonisotropic and excludes a compatible $(8,8)$ orbit pair.
This restricted optimality does not establish a globally minimum-weight complete frame; the residual weights are achieved representative weights. 
For the reversing fold, the representative weights vary along the regular orbits. Their mean is 8.54, which remains comparable to that of the plain fold despite the smaller code distance. Appendix~\ref{sec:reverse-fold-canonical} presents a detailed study of a separate reversing-fold example.

\section{Discussion}

We have developed a framework for navigating the design space of ultra-high-rate quantum codes. By separating the construction template, degree distribution, and lift size, we provide common coordinates for comparing different code families and navigating the tradeoff between rate, check weight, distance, and blocklength.
Guided by these principles, we introduced a new construction template based on the CPM-PP construction and a symplectic halving transformation that produces still more compact non-CSS codes, resulting in code constructions with significantly improved parameters, such as $[[90,21,11]]$, $[[140,31,15]]$, $[[200,43,20]]$. We complement these constructions with circuit-level simulations, low-weight canonical logical bases, and analysis motivated by classical ensemble heuristics.

Our work also opens up many interesting avenues of future research. Although the construction templates studied here already improve the finite-size frontier, further gains may be possible by changing both the template and its degree distribution. In classical LDPC codes, irregular degree distributions are central to optimizing the tradeoff among rate, belief-propagation threshold, and distance~\cite{richardson2001design,richardson2008modern,etsi}. Developing quantum-compatible irregular templates that maintain commutation and compactness is a natural direction for extending the design principles introduced here.

Extending these code families from quantum memories to processors requires logical operations with low space-time overhead and compatibility with architectural constraints.
The low-weight, symmetry-respecting canonical bases constructed here provide concrete inputs to generalized surgery schemes~\cite{williamson2024low,swaroop2024universal,yuan2026parsimoniousquantumlowdensityparitycheck} and associated ancillary-gadget constructions~\cite{webster2025explicitconstructionlowoverheadgadgets,bhardwaj2026high,zheng2025high,zheng2026logical,cross2024improved,he2025extractors}.
A promising direction is to combine the packet structure of group-valued constructions with projection.
Packet-resolved logical coordinates could help identify commuting sets of logical operators that can be addressed simultaneously using gadgets constructed for smaller projected codes and coupled to the original code through compatible transfer maps~\cite{hirasaki2026liftingliftedproductcodes}.
In parallel, architectural constraints could be incorporated directly into code and gadget design, exploring how the cyclic-lift structure can support implementations on reconfigurable neutral-atom and trapped-ion platforms~\cite{zhao2026towards,tripier2026fault}.

\section*{Data availability}
Supporting code instances and details will be made available in the following repository~\cite{designprinciplesdata}.

\section*{Acknowledgements}
We acknowledge helpful discussions with John Blue, Arpit Dua, Zhiyang He and Yuta Hirasaki. 
H.Z. is supported by NSF CUA, DOE C2QA, DARPA DSO, the DARPA HARQ program, and QuEra Computing. 
J.Y.L is supported by the IBM-Illinois Discovery Accelerator Institute and the Quantum Universe Center scholar program at KIAS. The authors acknowledge the use of AI tools from OpenAI and Anthropic for fleshing out construction details from our intuitions, performing simulations, making figures, and editing the text. The authors have reviewed all content and take full responsibility of the results. This manuscript subsumes the earlier results of Ref.~\cite{okada2026pair}.

\addtocontents{toc}{\protect\tocstop}%
\bibliography{cite}

\appendix

\section{Structural Bounds on the Code Distance of CPM-PP Codes}
\label{sec:cpmpp-structural-bounds}

The polynomial representation of a CPM matrix gives a family-level
upper bound on the code distance, by the argument that bounds the minimum Hamming distance of a
classical quasi-cyclic LDPC code with an all-one base matrix by
$(J+1)!$~\cite{mackay2001evaluation,smarandache2012quasi}.  Let
$R_P=\mathbb{F}_2[x]/(x^P-1)$ and write
\[
    \widehat H_X(x)=(x^{e_{i\ell}})_{i,\ell},\qquad
    \widehat H_Z(x)=(x^{d_{j\ell}})_{j,\ell}.
\]
For a set $S$ of $J+1$ block columns, define
$c_Z^S(x)\in R_P^L$ by
\begin{equation}
    (c_Z^S(x))_\ell=
    \begin{cases}
    \operatorname{perm}\!\left(
      \widehat H_Z(x)_{[:,S\setminus\{\ell\}]}
    \right),&\ell\in S,\\
    0,&\ell\notin S.
    \end{cases}
    \label{eq:cpmpp-permanent-vector}
\end{equation}
The maximal-minor identity in characteristic two gives
$\widehat H_Z(x)c_Z^S(x)^{\mathsf T}=0$.  Each of its $J+1$ possible
nonzero coordinates contains at most $J!$ monomials, so its expanded binary
weight is at most $(J+1)!$.

Suppose that $P$ is odd,
$L\geq2J+1$, and both binary check matrices have the natural maximum rank
\begin{equation}
    \operatorname{rank}H_X=\operatorname{rank}H_Z=J(P-1)+1.
    \label{eq:cpmpp-maximum-rank}
\end{equation}
The rank bound follows because the $P$ rows in every block-row family sum to
the same all-one vector, producing $J-1$ unavoidable row dependencies.  In
the Chinese-remainder decomposition of $R_P$, equality in
Eq.~\eqref{eq:cpmpp-maximum-rank} forces full row rank $J$ in every
nontrivial component.  The binary span $\mathcal C_Z$ of all vectors
$c_Z^S(x)$ and their cyclic shifts consequently has dimension
$(L-J)(P-1)$.  Since
\[
    (L-J)(P-1)>J(P-1)+1=\operatorname{rank}H_X,
\]
$\mathcal C_Z$ cannot be contained in $\operatorname{rowspan}H_X$.
At least one permanent vector is therefore an $X$-logical vector, and
interchanging $X$ and $Z$ gives a $Z$-logical vector.  Thus
\begin{equation}
    d_X\leq(J+1)!,\qquad d_Z\leq(J+1)!,\qquad d\leq(J+1)!.
    \label{eq:cpmpp-factorial-bound}
\end{equation}
The bound is $24$ for $J=3$ and $120$ for $J=4$.  It is independent of $P$;
under the stated rank and size conditions, a fully populated CPM family with
fixed column weight therefore does not have unbounded minimum distance.
Monomial cancellations can make a permanent vector lighter than the
factorial bound.  Instance-specific vectors are checked directly for zero
syndrome and nonmembership in the opposite stabilizer row space and generally
give sharper upper bounds.

\section{Symmetry-Informed Branch and Bound for Minimum-Distance Verification}
\label{sec:distance-verification}

We describe a branch-and-bound procedure for certifying distance lower bounds, first for CSS codes and then for general stabilizer codes. The search uses the irreducible-cluster principle of Dumer, Kovalev, and Pryadko~\cite{dumer2016distanceverificationclassicalquantum}, growing candidate supports through unsatisfied checks and testing zero-syndrome candidates for logical nontriviality. We incorporate cyclic-symmetry reduction for CPM lifts and admissible pruning bounds; for symplectically halved codes, certified parent-distance lower bounds provide an additional pruning criterion. Complementary upper bounds are obtained from explicit logical operators, which can be found using QDistRnd, sQetch, and other existing distance finders~\cite{pryadko2022qdistrnd,bhardwaj2026high,webster2026distance}. The distance is determined exactly when the lower and upper bounds coincide.


\subsection{CSS codes}

For binary matrices $H$ and $G$ satisfying
$\operatorname{rowspan}G\subseteq\ker H$, define
\begin{equation}
    d(H,G)=\min\{\operatorname{wt}x:Hx=0,
    \ x\notin\operatorname{rowspan}G\}.
    \label{eq:fixed-css-distance}
\end{equation}
The $X$-distance is obtained from $(H,G)=(H_Z,H_X)$ and the $Z$-distance
from $(H_X,H_Z)$.  The quantum distance is the smaller of the two.  We use a
complete search that either returns a vector in Eq.~\eqref{eq:fixed-css-distance}
of weight at most a prescribed limit $W$, or proves that no such vector
exists.

A search state consists of a selected support $S$, a locally forbidden set
$F$, and the syndrome $s(S)$ of its indicator vector, so that $|s(S)|$ counts
the unsatisfied checks.  If the syndrome is nonzero, the algorithm
chooses an unsatisfied check with the fewest available neighboring variables.
At least one further variable incident to this check must be selected to
restore even parity.  The search branches over these variables, adding each
previously considered choice to $F$ so that every support is represented by a
single ordered branch.

If $c$ is the selected unsatisfied check and $N(c)$ is the set of variables
incident to it, its available set is $A_c(S,F)=N(c)\setminus(S\cup F)$.  A branch fails immediately when this set
is empty.  When the syndrome becomes zero, the corresponding vector is
reduced against a row-echelon basis of $G$: a nonzero remainder is a logical
operator, whereas a zero remainder is a stabilizer and the branch terminates.

Two lower bounds prune a branch without excluding a possible solution.  If $\Delta$ is
the maximum column weight, at least
\begin{equation}
    b_1=\left\lceil\frac{|s(S)|}{\Delta}\right\rceil
    \label{eq:distance-pruning-bound}
\end{equation}
additional variables are required.  A second bound $b_2$ is obtained from any
set of unsatisfied checks whose available-variable sets are pairwise
disjoint, since distinct future variables are then needed for those checks.
The branch is stopped when $|S|+\max\{b_1,b_2\}>W$.  In the generic search,
each variable is used once as the least element of a possible support.

For binary matrices $H,G$ with
$\operatorname{rowspan}G\subseteq\ker H$, the ordered search described above
returns no vector if and only if $d(H,G)>W$.  Running it to completion for
$(H_Z,H_X)$ and $(H_X,H_Z)$ therefore proves a quantum-distance lower bound.

Completeness follows from two properties of a minimum-weight logical operator.
A minimum-weight member of
$\ker H\setminus\operatorname{rowspan}G$ can be chosen with Tanner-connected
support and contains no proper nonempty zero-syndrome subset.  Along the
ordered branch corresponding to this support, every proper partial support
therefore has a nonzero syndrome.  The selected unsatisfied check always has
an available variable in the remaining support.  Each added variable can
repair at most $\Delta$ unsatisfied checks, and checks used for $b_2$ require
distinct added variables, so neither lower bound can prune this branch.
The search reaches the full support and returns the logical operator.  The
converse, soundness of the search, follows because every returned vector has zero syndrome, lies
outside $\operatorname{rowspan}G$, and has weight at most $W$.

A generic search may require every one of the $LP$ variables as a root.  A
CPM lift is invariant under the common translation
$(\ell,a)\mapsto(\ell,a+c)$ of all lift indices.  Translating the least
variable in a support to lift index zero leaves only the $L$ roots
$(\ell,0)$, one for each block-column index.  For the recorded searches we use
prime $P$ and $W<P$ when applying this reduction.

In matrix form, if $Q$ is the common qubit shift by one lift index, there are
check-row permutation matrices $R_X,R_Z$ such that
\begin{equation}
    H_XQ=R_XH_X,\qquad H_ZQ=R_ZH_Z.
    \label{eq:cpmpp-cyclic-automorphism}
\end{equation}
Thus the shift preserves both kernels and both stabilizer row spaces.  Every
logical support has a translated representative rooted at some $(\ell,0)$,
which is why the reduction retains completeness.

For fixed $J,L,W$, the number of support states visited by the complete search
is bounded by
\begin{equation}
    L\sum_{t=0}^{W-1}(L-1)^t,
    \label{eq:distance-state-bound}
\end{equation}
which is independent of the lift size $P$.  Matrix construction, rank and
row-space preprocessing, syndrome updates, and row-space membership tests
still act on length-$LP$ objects and can grow with $P$.

The rows in any one block-row family of a fully populated CPM matrix sum to
the all-one vector.  Hence every vector in its kernel has even weight.  For an
even target $d_0$, completed searches through $W=d_0-2$ on both CSS sides
prove $d\geq d_0$. More generally, exclusion through $w_X-1$ together with an $X$-logical
vector of weight $w_X$ proves $d_X=w_X$, and the analogous statement holds
for $d_Z$.

\subsection{General stabilizer codes}

For a commuting stabilizer check matrix $H=(A\mid B)$, a Pauli vector
$e=(x\mid z)$ has syndrome $s(e)=Bx+Az$ and weight
$\operatorname{wt}_{\mathrm P}(e)=|\operatorname{supp}x\cup\operatorname{supp}z|$,
so a $Y$ entry counts once.  The distance is
\begin{equation}
    d=\min\{\operatorname{wt}_{\mathrm P}(e):Bx+Az=0,
    \ e\notin\operatorname{rowspan}H\}.
    \label{eq:noncss-search-distance}
\end{equation}
The search branches over Pauli labels on unused qubits.  The syndrome
columns for $X$, $Z$, and $Y$ on qubit $q$ are $B_{:,q}$, $A_{:,q}$,
and $A_{:,q}+B_{:,q}$, respectively.  We select the unsatisfied check
with the fewest available actions and try its anticommuting labels,
grouped by qubit.  Only the tried actions are forbidden in later
branches; labels commuting with this check remain available.  Each
qubit receives at most one label.  At zero syndrome, reduction against
the row space of $H$ tests logical nontriviality, and the branch ends.

The bounds $b_1$ and $b_2$ apply with $\Delta$ the largest single-qubit
syndrome weight and disjointness measured on available qubit sets.
For a further bound, let $c_q$ be the maximum number of currently
unsatisfied checks flipped by an allowed label on unused qubit $q$.
Order these values decreasingly and set
\begin{equation}
    b_3=\min\left\{t:\sum_{i=1}^{t}c_{(i)}\geq|s(e)|\right\}.
    \label{eq:noncss-coverage-bound}
\end{equation}
Any extension on $t$ qubits covers at most this sum of unsatisfied
checks, so $b_3$ is admissible.  Insufficient total coverage makes a
branch infeasible; otherwise we prune when
$\operatorname{wt}_{\mathrm P}(e)+\max\{b_1,b_2,b_3\}>W$.
Trying every least occupied qubit and its three labels shows
completeness.

For symplectically halved codes, Eq.~\eqref{eq:folded-stabilizer-normalizer}
identifies folded logical classes with parent $X$-logical classes, which we use to reduce the cost of the search.
An operator $e=F_\pi(a)$ of Pauli weight $w$ with $y$ entries equal to
$Y$ lifts to parent weight $\operatorname{wt}(a)=w+y$ by
Eq.~\eqref{eq:folded-pauli-weight}.  A certified parent $X$-distance
lower bound $D$ therefore requires $w+y\geq D$ for a folded logical.
Each further qubit adds at most two to the parent weight, giving the
additional pruning rule
\begin{equation}
    w+y+2(W-w)<D.
    \label{eq:noncss-parent-pruning}
\end{equation}
This uses a verified logical-class identification and a certified
parent lower bound. The CSS even-weight shortcut does not extend to
Pauli weight, since two occupied parent coordinates can fold to one
$Y$.

\section{Symplectic Halving of Rate-$1/5$ Balanced Product Codes}
\label{sec:zsz-halving}

Section~\ref{sec:noncss} states the halving condition for the pair-partition
template, but the requirement it implements---a fixed-point-free ZX
duality---is a property of the CSS parent alone and is therefore family
independent.  We illustrate the transfer on a second template, the rate-$1/5$
five-block codes of Refs.~\cite{hong2026quantum,bhardwaj2026high}.

Let $G=\mathbb{Z}_{\ell_1}\rtimes_q\mathbb{Z}_{\ell_2}
=\langle x,y\mid x^{\ell_1}=y^{\ell_2}=1,\;yxy^{-1}=x^{q}\rangle$
be a metacyclic group of order $\ell=\ell_1\ell_2$, and write $L[u]$ and $R[u]$
for the left- and right-regular representations of $u\in\mathbb{F}_2[G]$.  For
four three-term elements $a_1,a_2,b_1,b_2$, set $A=L[a_1]$, $B=L[a_2]$,
$C=R[b_1]$, $D=R[b_2]$ and
\begin{equation}
\begin{aligned}
    H_X&=\begin{pmatrix}A&0&B&0&C^{T}\\0&A&0&B&D^{T}\end{pmatrix},\\[2pt]
    H_Z&=\begin{pmatrix}C&D&0&0&A^{T}\\0&0&C&D&B^{T}\end{pmatrix}.
\end{aligned}
    \label{eq:zsz-blocks}
\end{equation}
The code has
$n=5\ell$, $k=\ell$, hence design rate exactly $1/5$, with check weight $9$ and
data-qubit degree $3$ per basis on the four plain block columns and $6$ on the
transposed fifth.  The template has a distance upper bound when using an Abelian group, motivating the use of a non-Abelian group $G$.

Reference~\cite{hong2026quantum} builds ZX dualities on this template by tying
the seeds to one another through an involutive automorphism $\varphi$ of $G$,
\begin{equation}
    b_1=\varphi(\bar a_1),\qquad b_2=\varphi(\bar a_2),
    \label{eq:zsz-duality-seeds}
\end{equation}
where $\bar u=\sum_{g\in\mathrm{supp}(u)}g^{-1}$ is the antipode.  The map
$\theta(g)=\varphi(g^{-1})$ is then an antiautomorphism, and conjugating by it
exchanges left with right multiplication.  Applying $\theta$ within each
of the five qubit blocks and exchanging the second and third of them gives a
qubit involution $\pi$, in the role of Eq.~\eqref{eq:halving_involution},
carrying $H_X$ to $H_Z$ row by row. However, this duality acts as a fixed point within three of the five qubit blocks, and therefore can not directly be used to perform symplectic halving.

\begin{table}[!t]
\centering
\footnotesize
\setlength{\tabcolsep}{4pt}
\begin{ruledtabular}
\begin{tabular}{lcc}
 & instance I & instance II\\
\colrule
$G$ & $\mathbb{Z}_{4}\rtimes_{3}\mathbb{Z}_{16}$
    & $\mathbb{Z}_{22}\rtimes_{9}\mathbb{Z}_{5}$\\
$\ell=\lvert G\rvert$ & $64$ & $110$\\
$\lvert Z(G)\rvert$ & $16$ & $2$\\
\colrule
$a_1$ & $1+x^{2}y^{2}+xy^{3}$ & $1+x^{21}+x^{6}y^{4}$\\
$a_2$ & $1+xy^{5}+x^{2}y^{5}$ & $1+x^{16}y^{2}+x^{9}y^{3}$\\
$b_1=\varphi(\bar a_1)$ & $1+x^{2}y^{7}+x^{2}y^{10}$ & $1+x^{21}+x^{18}y$\\
$b_2=\varphi(\bar a_2)$ & $1+x^{2}y+xy^{9}$ & $1+x^{17}y^{2}+x^{6}y^{3}$\\
$\varphi(x)$ & $x^{3}y^{8}$ & $x^{21}$\\
$\varphi(y)$ & $x^{3}y^{11}$ & $x^{8}y$\\
admissible $z$ & $x^{2},\,x^{2}y^{8}$ & $x^{11}$\\
$z$ used & $x^{2}$ & $x^{11}$\\
\colrule
parent & $[[320,64,\leq12]]$ & $[[550,110,\leq18]]$\\
parent girth & $6$ & $4$\\
parent check weight & $9$ & $9$\\
\colrule
fold & $[[160,32,\leq11]]$ & $[[275,55,\leq14]]$\\
fixed points of $\pi$ & $0$ & $0$\\
fold check weight & $9$ & $7$--$9$\\
fold qubit degree & $6$--$12$ & $5$--$12$\\
$Y$ positions & $0$ of $1152$ & $12$ of $1968$\\
\end{tabular}
\end{ruledtabular}
\caption{Two halved instances of the five-block template of
Eq.~\eqref{eq:zsz-blocks}, obtained by imposing
Eqs.~\eqref{eq:zsz-duality-seeds} and~\eqref{eq:zsz-twist} in the search.  The seeds $b_1,b_2$ and the
listed admissible twists are derived from $(\varphi,a_1,a_2)$.  Parent and fold distances are upper bounds from explicit
logical operators.  The last row counts $Y$ positions against every
non-identity position of a folded check, a $Y$ arising where a parent check
meets both members of a pair.}
\label{tab:zsz-fold-instances}
\end{table}

The obstruction is removed by twisting the within-block map,
\begin{equation}
    \theta_z(g)=z\,\varphi(g^{-1}),
    \label{eq:zsz-twist}
\end{equation}
for a group element $z$, keeping the involutive automorphism $\varphi$
and the block exchange above fixed.  Applying the map twice gives
\begin{equation}
    \theta_z^2(g)=z\,g\,\varphi(z)^{-1}.
    \label{eq:zsz-twist-square}
\end{equation}
Evaluating at the identity shows that an involution requires
$\varphi(z)=z$; the remaining condition $zgz^{-1}=g$ for all $g$ is
exactly $z\in Z(G)$.  These two conditions are also sufficient.
Centrality preserves the seed duality of
Eq.~\eqref{eq:zsz-duality-seeds} on the right-regular blocks.
A fixed point satisfies $g=z\varphi(g)^{-1}$, or equivalently
$g\varphi(g)=z$.  Thus an involutive twist is fixed-point-free precisely
when $z$ lies outside the image of $\mu(g)=g\varphi(g)$.
Since $b_1$ and $b_2$ do not depend on $z$, the twist leaves the CSS parent
unchanged and changes only the pairing; different admissible pairings
can yield inequivalent non-CSS codes.

Within this family, Eq.~\eqref{eq:zsz-twist} yields a necessary condition
on the lift group alone.  Let $K=Z(G)\cap\mathrm{Fix}\,\varphi$, which
contains every candidate involutive twist $z$.  If $|Z(G)|$ is odd,
then $K$ is an Abelian group of odd order, and squaring is a bijection
on $K$.  Every $z\in K$ therefore has a square root $h\in K$, for which
$\mu(h)=h\varphi(h)=h^2=z$.  Every candidate twist then has a fixed point.
Hence
\begin{equation}
    \lvert Z(G)\rvert\ \text{even}
    \label{eq:zsz-even-centre}
\end{equation}
is necessary for a fixed-point-free twist of the form
Eq.~\eqref{eq:zsz-twist}, for any involutive $\varphi$.
This gives a filter at Step~1 of the design flow in
Sec.~\ref{sec:construction}: a group failing
Eq.~\eqref{eq:zsz-even-centre} can be discarded from this twist search
before any lift is generated.  For example, the lift group of the
$[[150,30,10]]$ code of Ref.~\cite{bhardwaj2026high} has centre
$\mathbb{Z}_5$, excluding this twist construction.  The even-centre condition is
not sufficient.  For a chosen involutive $\varphi$ and seeds satisfying
Eq.~\eqref{eq:zsz-duality-seeds}, this construction requires a
$z\in K\setminus\mu(G)$; the resulting parent distance also depends on
the seed choice.

Table~\ref{tab:zsz-fold-instances} summarizes two instances found by imposing
Eqs.~\eqref{eq:zsz-duality-seeds} and~\eqref{eq:zsz-twist} during the search process, where we find that the blocklength drops by half, the encoding rate remains the same, while the distance drops from 12 to 11 in instance I and from 18 to 14 in instance II.

\section{Further details of constructing canonical bases} \label{app:profiling}

In this section, we provide further details of constructing logical bases for various codes introduced in this manuscript. Supporting code instances and additional details will be made available in the accompanying repository~\cite{designprinciplesdata}.

\subsection{Pairing structure}
\label{app:pairing-correlations}

To construct canonical logical orbit families, we must control how every translated $X$ seed commutes with every translated $Z$ seed.
The binary pairing, group-valued correlation, and packet correlation organize different levels of this calculation.

Let $G=\mathbb{Z}_P$, with $P$ odd, act by qubit permutations on $L$ regular orbits.
Writing $u$ for its generator, identify $R_G=\F_2[u]/(u^P-1)$.
For group-valued representatives $x,z\in R_G^L$, let $\widehat{x}$ and $\widehat{z}$ denote their binary expansions as column vectors. Their commutation bit is
\begin{align}
B(x,z) &:=  \widehat{x}^T \widehat{z} \in\F_2.
\end{align}
Right translation on group-valued vectors is represented on binary coordinates by a permutation matrix $P_g$: $\widehat{x g} = P_g \widehat{x}$.
Therefore, we obtain
\begin{align}
B(xg,zh) &= \widehat{x}^T P_g^T P_h^{\vphantom{T}} \widehat{z} = B(xgh^{-1},z).
\end{align}
Thus the commutation bit depends only on the relative translation of the two representatives.

To encode all relative-translation commutators at once, define the group-valued correlation. In particular, for $x,z \in R^L_G$, the group-algebra correlation is defined as
\begin{align}
C(x,z) &: = x^\dagger z = \sum_{\ell=1}^{L}x^*_\ell z_\ell,
\label{eq:app-correlation-product}
\end{align}
where $*$ is group inversion. Write $x_\ell(u)=\sum_{a=0}^{P-1}x_{\ell,a}u^a$ and similarly for $z_\ell(u)$.
First, note that
\begin{align}
x^*_l z_l = \sum_{a,b} x_{l,a} z_{l,b} u^{b-a} = \sum_{r} u^r \bigg( \sum_a x_{l,a} z_{l,a+r} \bigg)
\end{align}
The inner sum is the contribution of block $\ell$ to $B(xu^r,z)$.
Summing over all blocks therefore gives
\begin{align}
    C(x,z) = \sum_r B(x u^r, z) u^r.
\end{align}
This shows that each monomial $u^r$ labels a relative translation, and its coefficient records the corresponding binary commutator.
In particular, the desired canonical relation
$B(x_i g,z_j h)=B(x_i g h^{-1},z_j)=\delta_{ij}\delta_{g,h}$
is equivalent to $C(x_i,z_j) = \delta_{ij} 1$. Here $1$ is the group-algebra identity, whose only nonzero coefficient is at the identity group element.

Now, consider the projection onto the packet $\Omega$. Write $x_\Omega=\Phi_\Omega(x)$ and $z_\Omega=\Phi_\Omega(z)$, with the packet map applied entrywise.
The group inversion induces a field isomorphism
$*:K_{\Omega^{-1}}\to K_\Omega$, characterized by
$\Phi_\Omega(a^*)=(\Phi_{\Omega^{-1}}(a))^*$.
It need not be the identity even for a self-reciprocal packet.
Applying $\Phi_\Omega$ to $C(x,z)=x^\dagger z$ gives
\begin{align}
C_\Omega(x,z)
&:=\Phi_\Omega(C(x,z)) = (x_{\Omega^{-1}}^*)^{\mathsf T}z_\Omega.
\label{eq:app-correlation-packet}
\end{align}
Thus the $\Omega$ component of the correlation pairs the
$X$ component in $\Omega^{-1}$ with the $Z$ component in $\Omega$.
Since $\Phi$ is an isomorphism and maps $1$ to one in every field factor, normalization in every packet reconstructs the group-algebra identity. More precisely, the packet conditions imply that $C(x_i,z_j)-\delta_{ij}1$ vanishes modulo every irreducible factor $g_\Omega(u)$.
Since these factors are pairwise coprime, it also vanishes modulo their product $u^P-1$.
Hence $C(x_i,z_j)=\delta_{ij}1$ in $R_G$.
Combining this with the coefficient interpretation of $C$, we obtain
\begin{align}
& C_\Omega(x_i,z_j)=\delta_{ij}
\quad\text{for every }\Omega \nonumber\\
&\quad\Longleftrightarrow\quad
C(x_i,z_j)=\delta_{ij}1 \nonumber\\
&\quad\Longleftrightarrow\quad
B(x_i g,z_j h)=\delta_{ij}\delta_{g,h}
\quad\text{for all }g,h\in G.
\end{align}
Thus field-valued duality in every reciprocal packet pair reconstructs canonically paired binary translation orbits.

For completeness, the unshifted binary pairing can also be recovered from the packet correlations using field traces. For $a\in K_\Omega\cong\F_{2^{f_\Omega}}$, define the trace map
\begin{align}
\operatorname{Tr}_{K_\Omega/\F_2}(a)
&:= \sum_{r=0}^{f_\Omega-1}a^{2^r}\in\F_2.
\label{eq:app-field-trace}
\end{align}
The sum is evaluated in $K_\Omega$ but lies in $\F_2$\footnote{The sum is evaluated in $K_\Omega$, but squaring permutes its terms because $a^{2^{f_\Omega}}=a$.
It is therefore fixed by squaring and hence belongs to $\F_2 \subseteq K_\Omega$.}. Then
\begin{align}
B(x,z)
&= \sum_\Omega
\operatorname{Tr}_{K_\Omega/\F_2}
\bigl(C_\Omega(x,z)\bigr).
\end{align}
The packet correlation $C_\Omega(x,z)$ is not itself a binary commutation bit.
For example, if the correlation equals one in a single packet and vanishes in all others, then
$B(x,z)=\operatorname{Tr}_{K_\Omega/\F_2}(1)=f_\Omega\bmod2$,
which is zero for even-degree packets.
By contrast, correlation one in every packet reconstructs $C(x,z)=1$ and therefore gives canonical pairing between the translated families.

\subsection{Canonicalization of $[[184,50,10]]$} \label{sec:explicit_canonical_illustration}

The two exponent matrices of the chosen instance are
\begin{align}
E_X &= \begin{pmatrix}0&0&0&0&0&0&0&0\\0&12&8&21&6&1&19&15\\0&9&18&11&7&17&10&4\end{pmatrix},\\
E_Z &= \begin{pmatrix}0&15&7&22&7&0&22&15\\0&1&2&4&0&1&2&4\\0&5&17&6&6&17&5&0\end{pmatrix}. \label{eq:expanded-p23-exponent-matrices}
\end{align}
Label the physical qubits by $(c,u)$, where $c=0,\ldots,7$ is the qubit-block index and $u=0,\ldots,22$ is the position within that block.
Likewise, label the checks of type $A\in\{X,Z\}$ by $(r,s)$, where $r=0,1,2$ is the check-block index and $s=0,\ldots,22$ is its position.
The group-valued entry $(H_A)_{r,c}=t^{(E_A)_{r,c}}$ specifies the relative cyclic shift between these blocks.
Check $(r,s)$ acts on one qubit in each block $c$, with total support
\begin{align}
\operatorname{supp}(A_{r,s}) &= \bigl\{(c,s-(E_A)_{r,c}):c=0,\ldots,7\bigr\}.
\end{align}
The packet fields follow from
\begin{align}
t^{23}-1 &= (t+1)p_+(t)p_-(t),\\
p_+(t) &= t^{11}+t^9+t^7+t^6+t^5+t+1,\\
p_-(t) &= t^{11}+t^{10}+t^6+t^5+t^4+t^2+1. \label{eq:expanded-p23-factors}
\end{align}
The CRT map is given as
\begin{align}
    \Phi(a) = (a \,\,\,\textrm{mod}\,(t+1), a \,\,\,\textrm{mod}\, p_+, a \,\,\,\textrm{mod}\, p_- ).
\end{align}
The degree-11 factors are reciprocal, while the trivial factor is self-reciprocal.
Thus $R_{23}\cong\F_2\times K_+\times K_-$ with $K_\pm=\F_2[t]/(p_\pm)\cong\F_{2^{11}}$.
For these factors, the CRT projectors can be written explicitly as
\begin{align}
e_0 &= \sum_{s=0}^{22}t^s,\quad e_+=\sum_{s\in\mathcal E_+}t^s,\quad e_-=1+e_0+e_+,\\
\mathcal E_+ &= \{0,1,2,3,4,6,8,9,12,13,16,18\}. \label{eq:expanded-p23-explicit-projectors}
\end{align}
Under the CRT map, these projectors satisfy $\Phi(e_0) = (1,0,0)$, $\Phi(e_+)=(0,1,0)$, $\Phi(e_-)=(0,0,1)$.
Multiplication by each projector therefore retains its designated packet and removes the other two.
For this matrix, the normalizer and stabilizer ranks give the logical packet multiplicities $(m_0,m_+,m_-) = (6,2,2)$ and $k = 6+2\cdot11+2\cdot11=50$.
The largest possible regular multiplicity is therefore $q=2$.
The number two counts the independent field directions in each nontrivial packet; the number eleven is the binary dimension contributed by each such direction.
Two selected trivial directions accompany the two directions in each nontrivial packet, giving $2(1+11+11)=46$ regular logical dimensions.
The remaining four dimensions lie in the trivial packet.

\vspace{5pt}\noindent\emph{Constructing an initial packet-dual frame.}
We first construct an initial frame.
In each nontrivial packet, the $3\times3$ submatrices formed by the first three columns are invertible.
Write $H_{A,\Omega}=(\mathsf{C}_{A,\Omega}\, \mathsf{D}_{A,\Omega})$, where $\mathsf{C}_{A,\Omega}$ is $3\times3$ and $\mathsf{D}_{A,\Omega}$ is $3\times5$.
Since the fields have characteristic two~\footnote{Under characteristic of two, we have $-C^{-1} D = C^{-1} D$.}, bases of the two normalizers are
\begin{align}
N_{X,\Omega} &= \mqty( \mathsf{C}_{Z,\Omega}^{-1}\mathsf{D}_{Z,\Omega} \\ I_5 ), \label{eq:expanded-p23-explicit-kernel-bases}
\end{align}
and similarly for $N_{Z,\Omega}$.
For example, $H_{Z,\Omega}N_{X,\Omega}=\mathsf{D}_{Z,\Omega}+\mathsf{D}_{Z,\Omega}=0$.
The five columns correspond to choosing one of the free physical coordinates $c=3,\ldots,7$ to be one and the other free coordinates to be zero.
For this instance, the first two columns of each normalizer basis are independent modulo the corresponding three-dimensional stabilizer space.
This is verified by adjoining them to the stabilizer columns and obtaining rank five.
Choose these columns as $U_\Omega^{\mathrm{in}}$ and as a provisional reciprocal $Z$ basis $W_\Omega^{\mathrm{in}}$, respectively.
Their mutual Gram matrices need not initially be identities.
Correct by
\begin{align}
F_\Omega^{\mathrm{in}}
&=\bigl((U_{\Omega^{-1}}^{\mathrm{in}})^*\bigr)^{\mathsf T}
W_\Omega^{\mathrm{in}},\quad V_\Omega^{\mathrm{in}}=W_\Omega^{\mathrm{in}}(F_\Omega^{\mathrm{in}})^{-1}. \label{eq:expanded-p23-initial-nontrivial-duals}
\end{align}
Each inverse is a $2\times2$ inverse over $\F_{2^{11}}$, not a binary inverse on the 22-dimensional packet expansion.

The trivial packet is handled separately because its check rank is one rather than three.
Let $b_c$ be the standard unit vector of $\F_2^8$ at physical block $c$.
The trivial normalizer is the even-parity subspace of $\F_2^8$, and the stabilizer space is spanned by the all-ones vector.
The six vectors $b_0+b_c$ for $c=1,\ldots,6$ form a basis of its logical quotient.
Select
\begin{align}
U_0^{\mathrm{in}} &= (b_0+b_1\quad b_0+b_2),\qquad W_0^{\mathrm{in}}=U_0^{\mathrm{in}},\\
F_0^{\mathrm{in}} &= (U_0^{\mathrm{in}})^{\mathsf T}W_0^{\mathrm{in}}=\begin{pmatrix}0&1\\1&0\end{pmatrix},\\
V_0^{\mathrm{in}} &= W_0^{\mathrm{in}}(F_0^{\mathrm{in}})^{-1}=(b_0+b_2\quad b_0+b_1). \label{eq:expanded-p23-initial-trivial-duals}
\end{align}
This particular two-dimensional provisional $Z$ subspace already pairs nondegenerately with the selected $X$ subspace.
For other trivial-packet choices, one can instead use all six provisional $Z$ generators and solve the $2\times6$ right-inverse problem.

The resulting initial frame satisfies the three packet identities
\begin{align}
(U_0^{\mathrm{in}})^{\mathsf T}V_0^{\mathrm{in}} = I_2, \qquad \bigl((U_\mp^{\mathrm{in}})^*\bigr)^{\mathsf T}V_\pm^{\mathrm{in}}
&= I_2,
\label{eq:expanded-p23-three-grams}
\end{align}
as well as all packet normalizer equations.
Each $U_\Omega^{\mathrm{in}},V_\Omega^{\mathrm{in}}$ has eight rows and two columns.
CRT recombination and translation therefore produce 46 canonical pairs before any weight optimization is performed.
For the stated deterministic choices, the two initial seed pairs have weights
\begin{align}
\bigl(\wt(x_1^{\mathrm{in}}),\wt(z_1^{\mathrm{in}})\bigr) &= (46,54),\\
\bigl(\wt(x_2^{\mathrm{in}}),\wt(z_2^{\mathrm{in}})\bigr)&=(48,52). \label{eq:expanded-p23-initial-seed-weights}
\end{align}
These are canonical regular generators obtained directly from the checks.

\vspace{5pt}\noindent\emph{Coset optimization after reconstruction.}
Let $S_X,S_Z\in\F_2^{67\times184}$ contain row-space bases of the binary check matrices $\widehat H_X,\widehat H_Z$.
For each reconstructed seed, fixed-class optimization minimizes
\begin{align}
w_X([x_i^{\mathrm{in}}]) &= \min_{a\in\F_2^{67}}\wt\bigl(x_i^{\mathrm{in}}+S_X^{\mathsf T}a\bigr),\\
w_Z([z_i^{\mathrm{in}}]) &= \min_{b\in\F_2^{67}}\wt\bigl(z_i^{\mathrm{in}}+S_Z^{\mathsf T}b\bigr). \label{eq:expanded-p23-fixed-class-objectives}
\end{align}
Stabilizer dressing preserves the logical classes and all canonical pairings, so the four seeds can be optimized independently.
Translating each improved seed then supplies its entire equal-weight orbit.

For syndrome decoding, choose $C_X,C_Z\in\F_2^{117\times184}$ whose rows form bases of $\ker S_X$ and $\ker S_Z$, respectively.
Then $\ker C_X=\operatorname{im}S_X^{\mathsf T}$ and similarly for $Z$, so fixed-class optimization is equivalent to finding low-weight solutions of
\begin{align}
C_X y &= C_X x_i^{\mathrm{in}},\qquad C_Z v=C_Z z_i^{\mathrm{in}}. \label{eq:expanded-p23-coset-syndrome}
\end{align}
These matrices constrain the optimizer and are not additional physical checks.
Using only the normalizer equations and the 46 regular pairing constraints would give rank 113 rather than 117, leaving four residual logical directions free.
Such constraints preserve canonical regular pairing but do not fix the logical class.

We generate candidates using BP--OSD with varied priors, verify their syndromes exactly, and apply weight-decreasing additions of single check generators or pairs of check generators.
Starting from the initial pair profiles $(46,54),(48,52)$, 48 decoder trials per seed with this local descent yield
\begin{align}
\bigl(\wt(x_1^{\mathrm{coset}}),\wt(z_1^{\mathrm{coset}})\bigr) &= (30,32),\\
\bigl(\wt(x_2^{\mathrm{coset}}),\wt(z_2^{\mathrm{coset}})\bigr) &= (32,34). \label{eq:expanded-p23-coset-reduced-weights}
\end{align}
Explicit stabilizer-difference witnesses verify preservation of the initial logical classes, and the translated frame retains Gram matrix $I_{46}$.
The achieved weights are upper bounds on the corresponding coset minima.

\vspace{5pt}\noindent\emph{Changing the generators to obtain the sparse frame.}
Further optimization can vary the logical generators rather than remaining within the four initial stabilizer cosets.
Paired module basis changes, new trivial-packet attachments, and incremental selection of reciprocal pairs provide this additional freedom while maintaining canonicality.
Our earlier blind packet-constrained search used this broader optimization and found regular pair profiles $(10,12),(12,10)$.

To summarize, at each of the three stages, the achieved seed-pair weights are
\begin{center}
\begin{tabular}{lcc}
Stage & Orbit 1 wt. & Orbit 2 wt. \\
\hline
Initial frame & $(46,54)$ & $(48,52)$ \\
Coset reduction & $(30,32)$ & $(32,34)$ \\
Joint generator search & $(10,12)$ & $(12,10)$
\end{tabular}
\end{center}
For the remainder of the worked example, $U_\Omega,V_\Omega$ denote this optimized frame rather than the deterministic initial frame.
Each matrix $U_\Omega,V_\Omega$ has eight rows and two columns.
For example, the verified sparse frame used below has
\begin{align}
U_0 &= \begin{pmatrix}1&0\\1&1\\1&0\\1&0\\1&1\\0&0\\1&0\\0&0\end{pmatrix},\qquad V_0=\begin{pmatrix}0&0\\0&0\\1&1\\0&1\\0&1\\0&1\\0&1\\1&1\end{pmatrix},\qquad U_0^{\mathsf T}V_0=I_2. \label{eq:expanded-p23-trivial-frames}
\end{align}
The row index is the physical block $c=0,\ldots,7$.
The trivial-packet checks are all-ones rows, and the displayed columns satisfy their even-parity normalizer constraints.

Represent physical block $c$ by $a_c(t)=\sum_{s=0}^{22}a_{c,s}t^s$, where $t$ is a formal translation variable and $s$ labels a qubit position.
The overall binary coordinate is $8s+c$.
For the first $X$ seed in block $c=0$, the three packet coordinates are
\begin{align}
u_{1,0}^{(0)} &= 1,\\
u_{1,+}^{(0)} &= t^9+t^7+t^6+t^4+t^3+t^2+t,\\
u_{1,-}^{(0)} &= t^{10}+t^9+t^7+t^6+t^5+t+1. \label{eq:expanded-p23-one-block-residues}
\end{align}
They are elements of three different fields and are not added as raw coefficient strings of lengths one, eleven, and eleven.
First multiply each polynomial lift by its CRT projector in $R_{23}$.
The resulting three length-23 vectors all refer to the same block of physical qubits.

Let $\mathcal A_\Omega$ denote the exponent support of the embedded component $e_\Omega\widetilde u_{1,\Omega}^{(0)}$ in this block.
Explicit CRT expansion gives
\begin{align}
\mathcal A_0 &= \{0,1,\ldots,22\},\\
\mathcal A_+ &= \{0,3,4,7,9,14,15,16,17,18,20,22\},\\
\mathcal A_- &= \{1,2,5,6,8,10,11,12,13,14,19,21\}. \label{eq:expanded-p23-one-block-supports}
\end{align}
The two nontrivial supports cover all 23 positions and intersect only at position 14.
Every position other than 14 therefore occurs in exactly two of the three supports, whereas position 14 occurs in all three.
Modulo two, only position 14 survives:
\begin{align}
\mathcal A_0\mathbin{\triangle}\mathcal A_+\mathbin{\triangle}\mathcal A_- &= \{14\},\qquad u_1^{(0)}(t)=t^{14}. \label{eq:expanded-p23-one-block-xor}
\end{align}
Here $\triangle$ denotes symmetric difference, which is XOR on indicator vectors.
This is an explicit physical-coordinate example.

\vspace{5pt}\noindent\emph{The complete recombined seeds.}
The same calculation in all eight physical blocks yields the four columns below.
All entries belong to $\F_2[t]/(t^{23}-1)$, and an entry $t^a+t^b$ indicates support at positions $a,b$ within that block.
Each column is one physical seed, with the first and second $X$ seeds paired respectively with the first and second $Z$ seeds.

\begin{center}
\begin{tabular}{c|cccc}
Block $c$ & $u_1$ & $v_1$ & $u_2$ & $v_2$ \\
\hline
0 & $t^{14}$ & $t^{15}+t^{20}$ & $0$ & $t+t^2$ \\
1 & $t^{10}$ & $t^8+t^{16}$ & $t^2$ & $t^2+t^7$ \\
2 & $t^7$ & $t^7$ & $t^5+t^{17}$ & $t^6$ \\
3 & $t^4$ & $t^9+t^{16}$ & $t+t^3$ & $t$ \\
4 & $t^{19}$ & $t^8+t^{20}$ & $t^{18}$ & $t^{18}$ \\
5 & $t^7+t^8$ & $0$ & $t^{12}+t^{17}$ & $t^{18}$ \\
6 & $t^9$ & $t^7+t^9$ & $t+t^{17}$ & $t^6$ \\
7 & $t^{10}+t^{15}$ & $t^{15}$ & $t+t^9$ & $t^7$ \\
\hline
Physical weight & $10$ & $12$ & $12$ & $10$
\end{tabular}
\end{center}

The three embedded components of $x_1=\operatorname{coef}(u_1)$ have weights $138,96,96$, respectively, but their XOR has weight ten.
The corresponding component weights of $z_1=\operatorname{coef}(v_1)$ are $46,84,84$, and their XOR has weight twelve.
These cancellations occur after physical expansion and are not visible from the packet dimensions alone.

Let $T$ multiply every block polynomial by $t$ modulo $t^{23}-1$.
Translating the displayed seeds gives
\begin{align}
B(T^a x_i,T^b z_j ) &= \delta_{ij}\delta_{ab}, \label{eq:expanded-p23-physical-canonicality}
\end{align}
for $i,j\in\{1,2\}$ and $a,b\in\{0,\ldots,22\}$.
Thus $(T^a x_1,T^a z_1)$ gives 23 canonical pairs of weights $(10,12)$, and $(T^a x_2,T^a z_2)$ gives 23 canonical pairs of weights $(12,10)$.
The four invariant residual pairs are obtained from the orthogonal complements of these selected regular sectors.
They are not included in the two low-weight regular families.

\subsection{Even P: $[[400,86,22]]$}
\label{sec:even_P_canonical}

For even-order lifts, the large regular orbit families may persist, but the residual logical directions need not be fixed by translation.
The appropriate replacement is the primary-packet decomposition developed in Appendix~D of Ref.~\cite{lee2026logicalspectroscopyliftedproductcodes}, which retains the repeated-root structure lost under ordinary packet evaluation~\cite{evenP}.

Consider the PP code with $(J,L,P)=(4,10,40)$.
Both binary check matrices have rank $157$, so its parent CSS code encodes $400-2\cdot157=86$ logical qubits.
Writing $R_{40}:=\F_2[u]/(u^{40}-1)$, we have
\begin{align}
u^{40}-1 &= (u+1)^8\Phi_5(u)^8,
\label{eq:p40-primary-factorization}
\end{align}
where $\Phi_5(u)=u^4+u^3+u^2+u+1$ is irreducible.
The Chinese remainder theorem separates the two coprime primary factors, but each retains its eighth power.
Their logical components therefore consist of cyclic blocks $\F_2[u]/(f^a)$, where the exponent $a$ records information that evaluation at a root alone cannot recover.

Let $T$ denote translation on the logical quotient.
Computing the kernels of successive powers $f(T)^r$ determines these block lengths.
On each Pauli side, the $u+1$ component has lengths $(8,8,2,2,1,1)$, while the $\Phi_5$ component has lengths $(8,8)$.
One length-eight block from each component recombines into a regular copy of $R_{40}$.
Thus, with $U_2:=\F_2[u]/((u+1)^2)$,
\begin{align}
\calL_X\cong\calL_Z &\cong R_{40}^2\oplus U_2^2\oplus\one^2.
\label{eq:p40-parent-logical-module}
\end{align}
Its binary dimension is $2\cdot40+2\cdot2+2=86$.

Each $U_2$ describes a pair of logical classes exchanged by translation.
Indeed, for a generator $v$, the nonzero class $w=(T+I)v$ satisfies
\begin{align}
Tv &= v+w, \qquad Tw=w, \qquad T^2v=v.
\label{eq:p40-length-two-action}
\end{align}
The independent classes $v$ and $Tv$ form a period-two orbit, but only their sum is fixed.
Consequently, $U_2$ cannot be replaced by two independent fixed directions.
These relations hold modulo stabilizers.

Choosing one seed in each summand gives the orbit-rank profile
\begin{align}
(40,40,2,2,1,1).
\label{eq:p40-parent-orbit-profile}
\end{align}
The six residual dimensions therefore comprise two translated pairs and two fixed directions.
Including the invariant direction inside each regular summand, the full fixed subspace has dimension six.
An unjustified substitution into the odd-order formula would instead predict $R_{40}^2\oplus\one^6$, with the same total dimension but eight fixed directions.

The coefficient interpretation of the group-valued correlation remains valid for even $P$.
For the regular orbit seeds at $P=40$, canonicality is equivalent to
$C(x_i,z_j)=\delta_{ij}1$ modulo both $(u+1)^8$ and $\Phi_5(u)^8$.
Normalization only in the residue fields does not establish these identities, because it discards the nilpotent information needed to determine all translated binary pairings.
The residual sectors must be chosen orthogonally to the opposite regular sectors and here cannot be restricted to fixed classes. 

A preliminary optimization produced $40$ canonical pairs of weights $(24,44)$ and the other $40$ canonical pairs of weights $(44,24)$ in two regular sectors, along with average weight 60 residual pairs. All together, they constitutes 86 canonical pairs of average weight $35.8$.
Independently, the weight-$22$ witnesses and the separate lower-bound certificate establish the code distance.

\subsection{Reverse fold: $[[92,25,8]]$ from $[[184,50,10]]$}
\label{sec:reverse-fold-canonical}

We now apply the reversing fold to the parent code studied in Appendix~\ref{sec:explicit_canonical_illustration}.
The folded logical Pauli space inherits one parent logical module,
\begin{align}
\mathcal Q_{\mathrm{fold}}
&\cong\calL_X
\cong R_{C_{23}}^2\oplus\one^4.
\end{align}
Its binary dimension is $50$, corresponding to $25$ folded logical qubits on $92$ physical qubits.
The trivial packet and the two reciprocal degree-$11$ packets therefore have dimensions
\begin{align}
\bigl(
\dim\mathcal V_{\mathrm{triv}},
\dim\mathcal V_{\Omega_+},
\dim\mathcal V_{\Omega_-}
\bigr)
&=(6,22,22).
\label{eq:p23-folded-packet-profile}
\end{align}
The reversing $\pi$ sends the parent $X$ packet $\Lambda$ to the parent $Z$ packet $\Lambda^{-1}$.
Since the parent CSS pairing couples reciprocal packets, the folded pairing instead couples matching packets.
Thus distinct folded packets are mutually orthogonal, and the symplectic form restricted to each packet is nondegenerate.

\vspace{5pt}\noindent\emph{Canonical orbits inherited from the parent.}
The canonical parent frame is reflection-compatible.
After choosing translation origins, its two $X$ seeds $a,b$, of weights $10$ and $12$, have canonical $Z$ partners $\pi(b)$ and $\pi(a)$, respectively.\footnote{For the seed table in Appendix~\ref{sec:explicit_canonical_illustration}, take $a=T^{17}x_1$ and $b=T^{20}x_2$.
Their corresponding $Z$ partners satisfy $T^{17}z_1=\pi(b)$ and $T^{20}z_2=\pi(a)$.}
For $i\in\mathbb Z_{23}$, define the directly folded representatives
\begin{align}
\overline X_{i+1}^{(0)}
&:=F_\pi(T^i a),&
\overline Z_{i+1}^{(0)}
&:=F_\pi(T^{-i}b).
\label{eq:p23-direct-fold-orbits}
\end{align}
Using $\pi T^{-j}=T^j\pi$ and parent canonicality gives
\begin{align}
B_{\mathrm{sp}}\bigl(
\overline X_{i+1}^{(0)},
\overline Z_{j+1}^{(0)}
\bigr)
&=B(T^i a,\pi T^{-j}b)
\nonumber\\
&=B(T^i a,T^j\pi(b))
=\delta_{ij}.
\label{eq:p23-direct-fold-canonicality}
\end{align}
The within-family pairings vanish by the off-diagonal parent CSS pairing conditions.
The forward orbit of $F_\pi(a)$ and the reverse orbit of $F_\pi(b)$ therefore form $23$ canonical pairs without further Gram-matrix inversion.

These two regular families span a $46$-dimensional nondegenerate symplectic subspace with packet dimensions $(2,22,22)$.
Its symplectic complement is four dimensional and lies entirely in the trivial packet.
It is therefore fixed by $K$ and supplies the remaining two canonical pairs.
Thus the parent regular sector yields $23$ folded logical qubits, while the orbit-rank profile of the folded Pauli space is $(23,23,1,1,1,1)$.

\vspace{5pt}\noindent\emph{Folded weights and coset optimization.}
Although the canonical pairing is inherited directly, physical weights must be evaluated after folding.
For a parent binary word $w$, let $c_\pi(w)$ count the $\pi$-pairs on which both coordinates are one.
Writing $\wt$ for parent Hamming weight and $\wt_{\mathrm P}$ for folded Pauli weight gives
\begin{align}
\wt_{\mathrm P}\bigl(F_\pi(w)\bigr)
&=\wt(w)-c_\pi(w).
\label{eq:folded-weight-map}
\end{align}
A doubly occupied pair becomes one folded $Y$ operator.
Consequently,
\begin{align}
\wt_{\mathrm P}\bigl(K^tF_\pi(w)\bigr)
&=\wt(w)-c_\pi(T^t w).
\end{align}
Parent translation preserves Hamming weight but can change the number of doubly occupied pairs, so folded Pauli weight need not be constant along an orbit.

For the chosen translation origins, $F_\pi(a)$ and $F_\pi(b)$ have folded weights $8$ and $12$.
The first orbit has weights $8,9,10$ with multiplicities $1,4,18$, respectively.
The second has weights $10,11,12$ with multiplicities $2,4,17$.
Together, the $46$ directly folded representatives have total weight $492$ and mean operator weight $492/46\approx10.696$.

Let $\mathcal S_X:=\operatorname{im}\widehat H_X^{\mathsf T}$ be the parent binary $X$-stabilizer space.
The minimum folded weight of a logical class $[w]\in\calL_X$ is
\begin{align}
w_{\mathrm{fold}}([w])
&:=\min_{s\in\mathcal S_X}
\wt_{\mathrm P}\bigl(F_\pi(w+s)\bigr).
\label{eq:p23-folded-coset-weight}
\end{align}
Because $K$ does not preserve folded Pauli weight, optimizing one seed representative need not optimize the other classes in its orbit.
We therefore search these stabilizer cosets independently.
The improved representatives retain
\begin{align}
[\overline X_{i+1}]=K^i[F_\pi(a)], \qquad [\overline Z_{i+1}]=K^{-i}[F_\pi(b)].
\label{eq:p23-dressed-canonical-orbits}
\end{align}
These are relations between logical classes; the physical representatives need not remain literal translates of a single representative.
Stabilizer dressing preserves canonical pairing, the regular logical subspace, and its symplectic complement. The search reduces the $Z$ representatives of pairs $i=12$ and $22$ from weight $12$ to weight $11$. The regular total weight therefore decreases from $492$ to $490$.

\vspace{5pt}\noindent\emph{Completing the canonical frame.}
We choose and optimize two canonical pairs in the four-dimensional fixed complement, obtaining pair weights $(26,21)$ and $(25,23)$.
Together with the regular sector, they give a complete frame satisfying
\begin{align}
B_{\mathrm{sp}}(\overline X_i,\overline X_j)
&=B_{\mathrm{sp}}(\overline Z_i,\overline Z_j)=0,
\nonumber\\
B_{\mathrm{sp}}(\overline X_i,\overline Z_j)
&=\delta_{ij},
\qquad 1\leq i,j\leq25.
\label{eq:p23-complete-canonical-frame}
\end{align}
Table~\ref{tab:p23-complete-canonical-weights} lists the achieved pair weights.
Including the four residual operators gives total weight $585$ and mean operator weight $585/50=11.70$.
All weights count a site carrying $Y$ once.

\begin{table}[t]
\centering
\begin{tabular}{rc@{\quad}rc@{\quad}rc@{\quad}rc}
\toprule
$i$ & $p_i$ & $i$ & $p_i$ & $i$ & $p_i$ & $i$ & $p_i$ \\
\midrule
1 & $(8,12)$  & 7  & $(10,11)$ & 13 & $(10,11)$ & 19 & $(9,10)$  \\
2 & $(9,12)$  & 8  & $(9,12)$  & 14 & $(10,12)$ & 20 & $(10,12)$ \\
3 & $(10,12)$ & 9  & $(10,12)$ & 15 & $(10,12)$ & 21 & $(10,12)$ \\
4 & $(10,12)$ & 10 & $(9,11)$  & 16 & $(10,12)$ & 22 & $(10,11)$ \\
5 & $(10,12)$ & 11 & $(10,10)$ & 17 & $(10,12)$ & 23 & $(10,12)$ \\
6 & $(10,12)$ & 12 & $(10,11)$ & 18 & $(10,11)$ &    &           \\
\midrule
24 & $(26,21)$ & 25 & $(25,23)$ & & & & \\
\bottomrule
\end{tabular}
\caption{Canonical pair weights $p_i=(\wt_{\mathrm P}(\overline X_i),\wt_{\mathrm P}(\overline Z_i))$ for the folded $[[92,25,8]]$ code.
Pairs $1$--$23$ are obtained by directly folding the reflection-compatible parent frame and then reducing weights within folded stabilizer cosets.
Pairs $24$ and $25$ lie in the fixed complement.
A site carrying $Y$ contributes one to the weight.}
\label{tab:p23-complete-canonical-weights}
\end{table}

\vspace{5pt}\noindent\emph{Distance certificate.}
The directly folded representative $F_\pi(a)$ has weight eight, giving $d_{\mathrm{fold}}\leq8$.
For the matching lower bound, suppose a folded Pauli has $y$ sites carrying $Y$ and $s$ sites carrying only $X$ or only $Z$.
Its folded weight is $y+s$, while its inverse image under $F_\pi$ has parent weight $2y+s$.
Every nontrivial folded logical lifts to a nontrivial parent $X$-logical class, so the parent distance $10$ excludes candidates with $2y+s<10$.
It therefore suffices to exclude nontrivial folded logicals in the region
\begin{align}
y+s\leq7, \qquad 2y+s\geq10.
\end{align}
Ordering the $92$ folded qubits partitions the candidates into $3\times92=276$ branches according to the first occupied site and its Pauli symbol.
An exhaustive branch-and-bound calculation completes all branches and excludes every element of $N(S_{\mathrm{fold}})\setminus S_{\mathrm{fold}}$ in this region.
Together with the weight-eight witness, this establishes
\begin{align}
\boxed{[[92,25,8]]}.
\label{eq:p23-folded-parameters}
\end{align}
The distance certificate is exact, whereas the displayed basis weights are achieved representative weights, not certified coset minima.

\section{Details of the Ensemble Studies}
\label{sec:de-matrices}

This appendix summarizes the numerical protocols and code instances used to compare
density evolution with finite-size instance BP simulations in Sec.~\ref{sec:de}. We
also provide further details behind the ensemble-distance
comparison of Sec.~\ref{sec:distance}.

\subsection{Finite-code midpoint comparisons}

We use the depolarizing code-capacity channel and joint-Pauli SPA belief
propagation.  The density-evolution value $p_{\mathrm{DE}}$ is the bisection
boundary at which 180-iteration population density evolution, using 30,000
messages, reaches residual component error rate $10^{-5}$.  For a finite code
of blocklength $n$, $p_{50}(n)$ is the physical error rate at which the full
logical frame error rate equals $0.5$ under SPA with at most 100 iterations and
syndrome-based early stopping.  We obtain $p_{50}$ by linear interpolation
between adjacent measured points that bracket frame error rate $0.5$.

For the degree-distribution comparison, we evaluate all 11 regular ensembles
with $3\leq J\leq5$, $L\leq16$, and positive design rate at most $0.60$.  For each
$(J,L)$, we construct three independent matrices at $P=97$ and three at
$P=307$.  Every matrix is checked for CSS orthogonality, binary rank, the
specified row and column weights, distinct permutation assignments, and the
absence of Tanner 4-cycles.  The finite-code campaign uses 177,750 frames.
Each point in Fig.~\ref{fig:cpmpp-de}(a) pools the three matrices at fixed
$(J,L,P)$ and carries a 95\% interval.  Increasing the lift reduces the RMSE
between $p_{50}$ and $p_{\mathrm{DE}}$ from $0.00276$ to $0.00152$ and the mean
signed gap $p_{\mathrm{DE}}-p_{50}$ from $0.00229$ to $0.00119$.  The absolute
gap is smaller at the larger lift for 10 of the 11 degree distributions.

We separately test the representative degree distribution $(J,L)=(4,12)$ at
$n=444$, $732$, $1164$, $1812$, $2676$, and $3684$.  At each length we
construct three independent codes without Tanner 4-cycles.  A direct search
finds no nontrivial logical operator of weight at most 12 in any of the 18
selected codes.  The measured midpoint increases from $p_{50}=0.07321$ at
$n=444$ to $0.07792$ at $n=3684$, approaching
$p_{\mathrm{DE}}=0.07814$.  The remaining difference at the largest length is
$2.2\times10^{-4}$.

\subsection{Low-FER comparison}

The low-FER study in Fig.~\ref{fig:cpmpp-de}(b) uses three CPM-PP codes with degree profile $(J,L)=(4,12)$:
$[[444,154,18]]$, $[[1068,362,d\geq26]]$, and $[[3684,1234,24]]$.

One frozen decoder schedule is used throughout.  It begins with sum-product belief propagation (SPA) for at
most 100 iterations, followed by the bounded CPM-aware residual repair of
Appendix.~\ref{sec:cpmpp-code-capacity-decoder}.  A remaining nonzero residual is
processed by continued SPA up to 1000 iterations and another repair, followed
by normalized min-sum (NMS) with 100 iterations, normalization factor $0.8$,
and another repair.  The prescribed low-$b$ residual class receives a final
NMS--repair stage with normalization factor $0.75$.

All error bars are exact 95\%
Clopper--Pearson intervals, and the deepest measurement of each code contains
15, 3 and 1 final failures at $n=444$, $1068$ and $3684$.  The vertical reference is the finite-iteration
heuristic $p_{\mathrm{DE}}=0.07814$; the depolarizing hashing bound~\cite{Hashing1996} at design
rate $R=1/3$, where $R=1-h_2(p)-p\log_2 3$, is $p_{\mathrm{hash}}=0.10835$ and
lies outside the plotted range.

Each waterfall is described by a Gaussian in the physical error rate.  Fitting
$\mathrm{FER}(p)=\Phi\bigl((p-p_{50})/\sigma\bigr)$ to each code separately by
binomial maximum likelihood over the 17 fully postprocessed points gives
$(p_{50},\sigma)=(0.0721,0.0109)$, $(0.0752,0.0077)$ and $(0.0768,0.0043)$ at
$n=444$, $1068$ and $3684$.  The description is accurate across the whole
measured range, with residual deviance $10.7$ on $11$ degrees of freedom
($p=0.46$).

Two trends follow.  The midpoints increase with blocklength toward the
density-evolution value, $p_{\mathrm{DE}}-p_{50}$ falling from $0.0060$ to
$0.0029$ to $0.0014$, and the widths narrow approximately as $n^{-1/2}$, with
$\sigma\sqrt{n}$ equal to $0.230$, $0.251$ and $0.260$.  The $n^{-1/2}$ law is
the finite-length scaling form expected of a random-like ensemble, but it does
not hold exactly over this range.  Imposing a single $\sigma=\alpha/\sqrt{n}$
while leaving the midpoints free gives $\alpha=0.249$.  The discrepancy is more visible and carried by the shortest code,
whose fitted width is $7\%$ narrower than the common law predicts.

\subsection{Derivation of the regular-ensemble distance scale}
\label{sec:ensemble-distance-derivation}

Consider the standard $(J,L)$-regular ensemble and its average
weight enumerator~\cite{Gallager1960,richardson2008modern}.  It has $n$
variable nodes and
\begin{equation}
    m=\frac{nJ}{L}
\end{equation}
check nodes.  Hence the blocklengths in the ensemble are restricted by the
integrality of $nJ/L$. If we imagine the variable and check sides to have ``sockets" that correspond to an incidence relation, for a fixed support of $w$ variables, its $wJ$
incident variable sockets are matched uniformly to a subset of the $nJ$
check sockets.  The support is a codeword precisely when every check receives
an even number of these sockets.  It follows that the ensemble-average weight
enumerator is
\begin{equation}
    \mathbb{E}[A_w]
    =\binom{n}{w}\frac{N_{\mathrm{even}}(wJ)}{\binom{nJ}{wJ}},
    \label{eq:regular-average-enumerator}
\end{equation}
where
\begin{align}
    q(x)&=\sum_{\substack{0\leq i\leq L\\ i\ \mathrm{even}}}
        \binom{L}{i}x^i
        =\frac{(1+x)^L+(1-x)^L}{2},\\
    N_{\mathrm{even}}(a)&=[x^a]q(x)^m.
\end{align}
In particular, $N_{\mathrm{even}}(wJ)=0$ when $wJ$ is odd.

Define the binary entropy
\begin{equation}
    h_2(\omega)=-\omega\log_2\omega
    -(1-\omega)\log_2(1-\omega).
\end{equation}
For $0<\omega\leq 1/2$, take integer weights $w_n$ along blocklengths for
which $nJ/L$ is integral and $Jw_n$ is even, with $w_n/n\to\omega$, and
define
\begin{equation}
    G_{J,L}(\omega)
    =\lim_{n\to\infty}\frac{1}{n}
        \log_2\mathbb{E}[A_{w_n}],
    \label{eq:regular-distance-exponent-definition}
\end{equation}
along such admissible sequences.  The choice of even-lattice approximants
affects only subexponential prefactors.  Saddle-point evaluation of the
coefficient in Eq.~\eqref{eq:regular-average-enumerator}, with
$w_n/n\to\omega$, then gives
\begin{equation}
    \frac{s q'(s)}{q(s)}=L\omega
    \label{eq:regular-distance-saddle}
\end{equation}
and
\begin{equation}
    G_{J,L}(\omega)
    =(1-J)h_2(\omega)+\frac{J}{L}\log_2 q(s)
    -J\omega\log_2 s.
    \label{eq:regular-distance-growth-rate}
\end{equation}
Here $s\in(0,1]$.  The left-hand side of
Eq.~\eqref{eq:regular-distance-saddle} is the mean even subset size under the
tilted weights proportional to $\binom{L}{i}s^i$.  Its derivative with
respect to $\log s$ is the corresponding variance and is strictly positive,
so the saddle is unique.  It increases from zero to $L/2$ as $s$ increases
from zero to one.  At $\omega=1/2$, $s=1$ and
\begin{equation}
    G_{J,L}(1/2)=1-\frac{J}{L}.
\end{equation}

Near zero, the exponent has the expansion
\begin{equation}
    G_{J,L}(\omega)
    =\left(\frac{J}{2}-1\right)\omega\log_2\omega+O(\omega).
\end{equation}
Thus it is negative immediately above zero for $J\geq3$.  For $J=2$ the
leading term cancels, and the sharper expansion is
\begin{equation}
    G_{2,L}(\omega)=\omega\log_2(L-1)+O(\omega^2)>0
\end{equation}
for sufficiently small positive $\omega$ when $L>2$.  We define
$\delta_{J,L}$ as the first positive zero of $G_{J,L}$, with
$\delta_{2,L}=0$.  Numerically, Eq.~\eqref{eq:regular-distance-saddle} is
solved by bisection.

For any fixed admissible relative weight below the first positive zero where
$G_{J,L}(\omega)<0$, Eq.~\eqref{eq:regular-average-enumerator} shows that the
expected number of such codewords vanishes exponentially.  Markov's
inequality therefore excludes them with high probability.  This is only a
lower-side exclusion: the first-moment argument neither proves the existence
of codewords above $\delta_{J,L}$ nor establishes two-sided concentration of
the minimum distance.

The calculation above concerns the unconditioned classical ensemble
$\ker H$ for one random parity-check matrix.  In a CSS code, by contrast,
$X$- and $Z$-type logical operators are quotient spaces, for example
$\ker H_Z/\operatorname{row}(H_X)$ and
$\ker H_X/\operatorname{row}(H_Z)$.  Stabilizer rows are removed by these
quotients, and the constraint $H_XH_Z^{\mathsf T}=0$ conditions the two
matrices on orthogonality.  The CPM-PP codes impose the further structure of
circulant-permutation lifts and therefore are not draws from the configuration
ensemble.  Consequently, $n\delta_{J,L}$ is used here only as a heuristic
screen for degree distributions; it is not a theorem for the logical distance
of the CSS quotient or for a structured CPM-PP lift.

\subsection{Ensemble-distance comparison}

\begin{figure*}[!t]
    \centering
    \includegraphics[width=\textwidth]{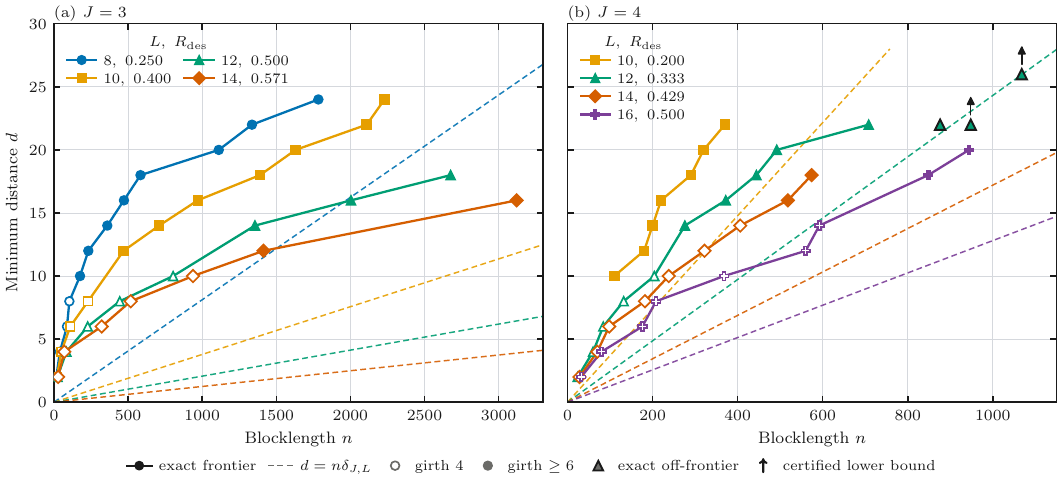}
    \caption{\textbf{CPM-PP distance frontiers by degree distribution.}
    Exact-distance Pareto frontiers, grouped by
    $(J,L)$, with the corresponding regular-ensemble scales
    $d=n\delta_{J,L}$. Filled markers have Tanner girth at least six, and open
    markers have girth four. The upward arrows in
    panel (b) give the proved lower bounds $d\geq22,26$ at $n=948,1068$,
    respectively. Dashed lines show
    $d=n\delta_{J,L}$.}
    \label{fig:cpmpp-distance-frontier}
\end{figure*}

We show additional comparisons of the per-family $(n,d)$ Pareto frontiers against the ensemble distance estimate in
Fig.~\ref{fig:cpmpp-distance-frontier}. Within each $(J,L)$ family, the best certified distance increases
with blocklength, while the slopes and absolute levels of the finite-code
frontiers vary with the permutation assignments; $\delta_{J,L}$ depends only on
the regular degree distribution. We find that while the qualitative trends agree, the ensemble estimate does not accurately predict the quantitative distance value.

\section{CPM-aware residual-syndrome decoding}
\label{sec:cpmpp-code-capacity-decoder}

\noindent\textbf{Decoder.}
Consider one binary component of CSS decoding.  Let $H$ be the corresponding
parity-check matrix, $s$ the measured syndrome, and $\widehat e$ the estimate
from an iterative decoder.  Its residual syndrome is
\begin{equation}
    \rho=s+H\widehat e.
    \label{eq:cpmpp-residual-syndrome}
\end{equation}
For a prescribed weight budget $w$, the postprocessor seeks a correction $c$
satisfying
\begin{equation}
    Hc=\rho,\qquad \operatorname{wt}(c)\leq w,
    \label{eq:cpmpp-residual-search}
\end{equation}
and replaces $\widehat e$ by $\widehat e+c$ only after this equation has been
verified directly.  It occupies the same place in the pipeline as
ordered-statistics post-processing~\cite{panteleev2019degenerate}, but returns
only corrections of weight at most $w$.  The procedure is applied separately
to the two CSS components.

Let $Q_\tau$ and $R_\tau$ denote simultaneous cyclic shifts of the qubit and
check-row lift indices by $\tau\in\mathbb Z_P$.  The CPM structure gives
\begin{equation}
    HQ_\tau=R_\tau H.
    \label{eq:cpmpp-decoder-equivariance}
\end{equation}
Consequently, $Q_\tau c$ solves $R_\tau\rho$ whenever $c$ solves $\rho$.
We represent every orbit $\{R_\tau\rho:\tau\in\mathbb Z_P\}$ by its
lexicographically smallest sparse syndrome and denote it by
$\operatorname{can}(\rho)$.

The bounded repair first maps $\rho$ to $\operatorname{can}(\rho)$ and
records the inverse shift.  It then queries a table of pre-enumerated
connected elementary trapping-set
supports~\cite{richardson2003error,chilappagari2010on,raveendran2021trapping} indexed by canonical
syndrome.  A unique table hit is shifted back and accepted only after direct syndrome
verification.  If the lookup fails or is ambiguous, deterministic
iterative-deepening search is performed through weights $1,\ldots,w$.  At a
nonzero search syndrome, the algorithm selects an unsatisfied check and
branches over its incident variables.  A minimum-weight solution of weight
at most $w$ can be chosen with every connected component touching the
residual syndrome; each such component lies within Tanner radius $2w-1$.
The search may therefore be restricted to that neighborhood.

Every unsuccessful state is memoized using
\begin{equation}
    \bigl(\operatorname{can}(\rho'),r\bigr),
    \label{eq:cpmpp-decoder-memo-key}
\end{equation}
where $\rho'$ is the current search syndrome and $r$ is the remaining weight
budget.  A branch is pruned when
\begin{equation}
    |\rho'|>r\Delta,
    \label{eq:cpmpp-decoder-pruning}
\end{equation}
where $\Delta$ is the maximum column weight of $H$.  Variable reliabilities
from the preceding iterative decoder determine branch order but not
acceptance.  A prescribed state budget $B$ makes the search terminate with
either a verified correction or a declared failure.

The wrapper may follow either SPA, the sum-product BP
decoder~\cite{Gallager1960,richardson2008modern}, or
NMS~\cite{chen2002density}.
After
each iterative-decoder stage, a zero residual is accepted and a nonzero
residual is passed to the bounded repair.  The fixed sequence and budgets
used for the code-capacity measurements are stated in
Appendix~\ref{sec:de-matrices}.

\noindent\textbf{Complexity.}
Let $E_{\mathrm{FG}}$ be the number of edges in the joint factor graph.  An
SPA or NMS stage with iteration cap $T$ requires
$O(TE_{\mathrm{FG}})$ message updates; for a fully populated $(J,L)$ CPM-PP
code, $E_{\mathrm{FG}}=2Jn$.  A direct unrestricted support search through
weight $w$ would examine $\sum_{i=0}^{w}\binom{n}{i}$ candidates.  The
online bounded search instead visits at most $B$ states per CSS component.
If $C_{\mathrm{can}}$ denotes the cost of canonicalizing one sparse residual
syndrome, incremental syndrome updates give online time
$O\bigl(B(C_{\mathrm{can}}+\Delta)\bigr)$ in addition to the iterative
decoder.  The failed-state table stores at most $B$ entries, while the
depth-first stack has depth at most $w$.  The ETS table is constructed once
for a fixed parity-check matrix and reused for every received syndrome.

\section{Detailed Derivations for Density Evolution}
\label{sec:de_details}

\begin{figure}[!t]
    \centering
    \begingroup
\tikzset{
  decheck/.style={draw=black,fill=black!10,rectangle,inner sep=0pt,
    minimum size=3.0mm},
  devar/.style={draw=black,fill=white,circle,inner sep=0pt,
    minimum size=5.4mm,font=\scriptsize},
  depsi/.style={draw=black,fill=orange!70!red,fill opacity=0.30,
    draw opacity=1,rectangle,inner sep=0pt,minimum size=4.6mm},
  deedge/.style={draw=black!60,line width=0.5pt},
  destrong/.style={draw=black,line width=0.9pt},
  debox/.style={draw=black!70,rounded corners=1pt,fill=blue!5,
    inner sep=2.2pt,font=\tiny,align=center},
  deflow/.style={-{Latex[length=1.5mm,width=1.3mm]},draw=black!70,
    line width=0.5pt},
  delab/.style={font=\scriptsize},
  detiny/.style={font=\tiny},
  deaxis/.style={draw=black!45,line width=0.4pt},
  decurve/.style={draw=blue!55!black,line width=0.7pt},
}
\newcommand{\dedensity}[5]{%
  \begin{scope}[shift={(#1,#2)}]
    \draw[deaxis] (-0.60,0) -- (0.60,0);
    \draw[decurve] plot[domain=-0.58:0.58,samples=50]
      (\x, {#3*exp(-30*(\x-0.28)*(\x-0.28))+#4*exp(-40*(\x+0.30)*(\x+0.30))});
    \node[detiny] at (0,-0.25) {#5};
  \end{scope}}
\begin{tabular}{@{}c@{}}
\begin{tikzpicture}
  \foreach \y in {0.72,0,-0.72} {
    \node[decheck] at (-2.30,\y) {};
    \draw[deedge] (-2.30,\y) -- (-1.15,0);
  }
  \foreach \y in {0.72,0,-0.72} {
    \node[decheck] at (2.30,\y) {};
    \draw[deedge] (2.30,\y) -- (1.15,0);
  }
  \draw[dashed,draw=black!45,rounded corners=2pt]
    (-1.62,-0.36) rectangle (1.62,0.36);
  \node[devar] at (-1.15,0) {$x_j$};
  \node[depsi] at (0,0) {};
  \node[font=\scriptsize] at (0,0) {$\psi$};
  \node[devar] at (1.15,0) {$z_j$};
  \draw[destrong] (-0.88,0) -- (-0.23,0);
  \draw[destrong] (0.23,0) -- (0.88,0);
  \node[delab,align=center] at (-2.30,1.30) {$H_Z$ rows\\[-2pt]{\tiny(detect $x$)}};
  \node[delab,align=center] at (2.30,1.30) {$H_X$ rows\\[-2pt]{\tiny(detect $z$)}};
  \node[detiny,anchor=north] at (-1.15,-0.52) {$M_{x,j}(x)$};
  \node[detiny,anchor=north] at (1.15,-0.52) {$M_{z,j}(z)$};
  \node[detiny,anchor=south] at (0,0.40) {qubit $j$};
  \node[delab] at (-3.90,1.30) {(a)};
\end{tikzpicture}
\\[2.5mm]
\begin{tikzpicture}
  \dedensity{-2.70}{0}{0.44}{0.30}{$a^{(\ell)}$}
  \node[debox] (chk) at (-0.85,0.10) {check node\\[-1pt]$L-1$ inputs};
  \dedensity{1.20}{0}{0.50}{0.21}{$b^{(\ell)}$}
  \node[debox] (var) at (-0.85,-1.15) {variable node\\[-1pt]$J-1$ inputs,\\[-1pt]then $\psi$};
  \dedensity{-2.70}{-1.24}{0.56}{0.10}{$a^{(\ell+1)}$}
  \draw[deflow] (-2.05,0.10) -- (chk.west);
  \draw[deflow] (chk.east) -- (0.55,0.10);
  \draw[deflow,rounded corners=3pt]
    (1.20,-0.38) -- (1.20,-1.15) -- (var.east);
  \draw[deflow] (var.west) -- (-2.05,-1.15);
  \draw[deflow,rounded corners=3pt]
    (-3.34,-1.24) -- (-3.80,-1.24) -- (-3.80,0) -- (-3.36,0);
  \node[detiny,rotate=90] at (-3.96,-0.62) {$\ell\to\ell+1$};
  \node[delab] at (-3.90,0.62) {(b)};
\end{tikzpicture}
\end{tabular}
\endgroup
    \caption{\textbf{Butterfly factor graph and the density-evolution
    recursion.}  (a) The gadget at one qubit $j$.  The $X$ component $x_j$ is a
    variable node of the Tanner graph of $H_Z$, which detects $X$ errors, and
    the $Z$ component $z_j$ a variable node of the Tanner graph of $H_X$.  The
    two component graphs share no edge and are joined only by the coupling
    factor $\psi_{xz}=\Pr(x,z)$, which carries the joint Pauli prior; the
    extrinsic beliefs $M_{x,j}$ and $M_{z,j}$ arriving from the two sides are
    combined by Eq.~\eqref{eq:de-joint-belief}.  (b) One iteration of the
    recursion.  Density evolution propagates message distributions rather than
    messages: the variable-to-check density $a^{(\ell)}$ passes through the
    check-node update with $L-1$ inputs to give $b^{(\ell)}$, then through the
    variable-node update with $J-1$ inputs followed by the coupling factor to
    give $a^{(\ell+1)}$.  The curves are schematic.  The recursion is run on
    the same pair of message densities on both component graphs.}
    \label{fig:butterfly-de}
\end{figure}
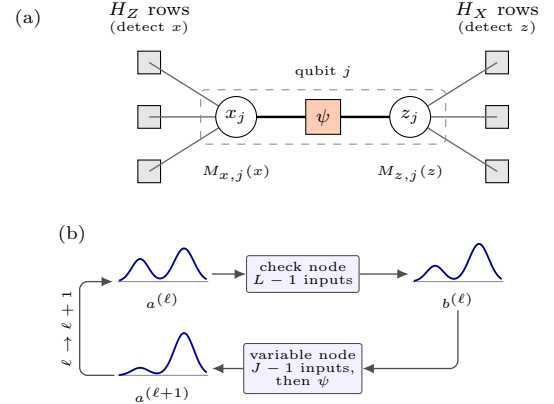

Quantum density evolution runs on the butterfly factor graph used by the joint $X/Z$ decoder. Here, we focus on the depolarizing code-capacity channel. A Pauli error $X^xZ^z$ on qubit $j$ is represented by two binary variables, $x_j$ and $z_j$, each of which belongs to its own Tanner graph. The two variables are joined by a single coupling factor $\psi_{xz}=\Pr(x,z)$ that contains the joint Pauli prior. There is one copy of this factor for every qubit, and there are no other edges between the two component graphs. Under depolarizing noise of strength $p$, its four entries are $\psi_{00}=1-p$ and $\psi_{01}=\psi_{10}=\psi_{11}=p/3$. Figure~\ref{fig:butterfly-de}(a) shows the gadget at one qubit.

Let $M_{x,j}(x)$ and $M_{z,j}(z)$ denote the extrinsic beliefs supplied to
qubit $j$ by the two Tanner graphs.  The butterfly factor combines them into
the joint belief
\begin{equation}
    b_j(x,z)\propto \psi_{xz}M_{x,j}(x)M_{z,j}(z).
    \label{eq:de-joint-belief}
\end{equation}
This equation summarizes the role of the joint decoder. Within each Tanner
graph, the messages are updated using ordinary binary BP. The coupling factor
then allows the evidence about one Pauli component to influence the other
while preserving the correlation in the channel prior.

Density evolution tracks the probability distributions of these messages
rather than the messages on a particular finite graph.  At each iteration, the
ordinary BP rules update these distributions according to the prescribed
degrees, and the butterfly factor combines the two sides through
Eq.~\eqref{eq:de-joint-belief}, as sketched in
Fig.~\ref{fig:butterfly-de}(b).  We represent the evolving distributions by
large populations of messages.  For the regular pair-partition ensembles
considered here, all check nodes are statistically equivalent, as are all
qubit positions, and the two bases are symmetric.  The ensemble is therefore
characterized by one evolving variable-to-check distribution
$a^{(\ell)}$; the check-to-variable distribution $b^{(\ell)}$ is obtained
from it by the check update and does not need to be independently tracked.  The
resulting four-outcome Pauli decisions estimate the residual error
probability.  Repeating this procedure at fixed $p$ determines whether the
error probability approaches zero, and a bisection in $p$ gives the
density-evolution threshold.
The same population method extends to lifted-product graphs with non-uniform degree distribution by employing a separate message distribution for each edge type~\cite{richardson2008modern}. The single- and multi-edge recursions are
derived in Secs.~\ref{sec:de-single-edge} and~\ref{sec:de-multi-edge} below.

\subsection{Single-Edge-Type Derivation for Pair-Partition Codes}
\label{sec:de-single-edge}

We first give the log-likelihood-ratio recursion used for the regular
pair-partition ensembles in Sec.~\ref{sec:de}.  For a binary message $m$, we
use the convention
\begin{equation}
    \mathcal L(m)=\log\frac{m(0)}{m(1)},
    \label{eq:de-llr-convention}
\end{equation}
so that a positive value favors the absence of the corresponding Pauli
component.  It is convenient to condition on the sampled error and work
with adjusted messages.  If the true component on a variable is $u\in
\{0,1\}$, the adjusted LLR is $\widetilde{\mathcal L}=(-1)^u\mathcal L$.
This change of variables removes the measured syndrome from the distributional
check update.

Let $V_s^{(\ell)}$ and $C_s^{(\ell)}$ denote the adjusted variable-to-check
and check-to-variable LLRs at iteration $\ell$ on side $s\in\{x,z\}$.
Independent copies of a random message are indicated by a second subscript.
For a $(J,L)$-regular ensemble, the check update is
\begin{equation}
    C_s^{(\ell)} \mathrel{\overset{\mathrm d}{=}}
    2\operatorname{atanh}\!\left[
        \prod_{r=1}^{L-1}
        \tanh\!\left(\frac{V_{s,r}^{(\ell)}}{2}\right)
    \right].
    \label{eq:de-single-check}
\end{equation}
Indeed, the unadjusted check message contains the syndrome sign.  The
syndrome is the parity of the true neighboring bits, so this sign cancels the
signs introduced by adjusting the $L-1$ incoming messages and the outgoing
message.

The correlation between the two Pauli components enters through the
butterfly factor.  Given an unadjusted extrinsic LLR $t$ for the opposite
component, the factor-to-variable messages are
\begin{align}
    \gamma_x(t)
    &=\log\frac{\psi_{00}+\psi_{01}e^{-t}}
                     {\psi_{10}+\psi_{11}e^{-t}},
    \label{eq:de-gamma-x}\\
    \gamma_z(t)
    &=\log\frac{\psi_{00}+\psi_{10}e^{-t}}
                     {\psi_{01}+\psi_{11}e^{-t}}.
    \label{eq:de-gamma-z}
\end{align}
For the depolarizing prior, $\psi_{00}=1-p$ and
$\psi_{01}=\psi_{10}=\psi_{11}=p/3$, and hence
$\gamma_x=\gamma_z$.  To update an $x$-side edge, we first draw the true pair
$(x,z)$ from $\psi$ and form the full adjusted extrinsic message from the
opposite side,
\begin{equation}
    B_z^{(\ell)}=\sum_{r=1}^{J}C_{z,r}^{(\ell)}.
    \label{eq:de-single-opposite-belief}
\end{equation}
The quantity supplied to $\gamma_x$ is $(-1)^zB_z^{(\ell)}$, which converts
this adjusted sum back to an ordinary LLR.  The complete variable update is
therefore
\begin{align}
    V_x^{(\ell+1)}
    &\mathrel{\overset{\mathrm d}{=}}
      (-1)^x\gamma_x\!\left((-1)^zB_z^{(\ell)}\right)
      +\sum_{r=1}^{J-1}C_{x,r}^{(\ell)},
    \label{eq:de-single-variable-x}\\
    V_z^{(\ell+1)}
    &\mathrel{\overset{\mathrm d}{=}}
      (-1)^z\gamma_z\!\left((-1)^xB_x^{(\ell)}\right)
      +\sum_{r=1}^{J-1}C_{z,r}^{(\ell)}.
    \label{eq:de-single-variable-z}
\end{align}
The opposite-side sum contains all $J$ check messages because the target edge
belongs to the other Tanner graph.  The same-side sum contains only $J-1$
messages, as required for an extrinsic variable-to-check update.  The initial
populations are obtained by setting the opposite-side belief to zero,
\begin{equation}
    V_x^{(0)}\mathrel{\overset{\mathrm d}{=}}(-1)^x\gamma_x(0),
    \qquad
    V_z^{(0)}\mathrel{\overset{\mathrm d}{=}}(-1)^z\gamma_z(0).
    \label{eq:de-single-initialization}
\end{equation}

The residual error is evaluated at a full variable node.  With
$B_s=\sum_{r=1}^{J}C_{s,r}$ and unadjusted extrinsic LLRs
$L_x=(-1)^xB_x$ and $L_z=(-1)^zB_z$, the four posterior log scores, up to a
common normalization, are
\begin{equation}
    Q_{ab}=\log\psi_{ab}-aL_x-bL_z,
    \qquad (a,b)\in\{0,1\}^2.
    \label{eq:de-posterior-scores}
\end{equation}
The decoder chooses the pair with the largest score, and density evolution
tracks the probability that this pair differs from the sampled $(x,z)$.  In
the population implementation, every random variable on the right-hand sides
of Eqs.~\eqref{eq:de-single-check}--\eqref{eq:de-single-variable-z} is sampled
independently from its current population.  The LLRs are clipped at magnitude
30, and the argument of $\operatorname{atanh}$ is clipped below unity to
avoid numerical saturation.

\subsection{Multi-Edge-Type Derivation for Lifted Product Codes}
\label{sec:de-multi-edge}

Lifted-product graphs generally contain inequivalent sets of sockets, so a
single message distribution no longer specifies the local neighborhood.  We
use the multi-edge-type ensemble description~\cite{richardson2008modern}.
Let $t\in\{1,\ldots,T\}$ label the edge types.  A variable-node type $v$ has
node fraction $\nu_v$ and degree vector
$\boldsymbol d_v=(d_{v,1},\ldots,d_{v,T})$, while a check-node type $c$ has
$\mu_c$ checks per variable node and degree vector
$\boldsymbol e_c=(e_{c,1},\ldots,e_{c,T})$.  The number of sockets per
variable node of each type is
\begin{equation}
    \epsilon_t=\sum_v\nu_v d_{v,t}
              =\sum_c\mu_c e_{c,t}.
    \label{eq:de-met-socket-balance}
\end{equation}
Consequently, a uniformly selected edge of type $t$ sees the node types with
edge-perspective probabilities
\begin{equation}
    \lambda_{v\mid t}=\frac{\nu_vd_{v,t}}{\epsilon_t},
    \qquad
    \rho_{c\mid t}=\frac{\mu_ce_{c,t}}{\epsilon_t}.
    \label{eq:de-met-edge-probabilities}
\end{equation}

Let $V_{s,t}^{(\ell)}$ and $C_{s,t}^{(\ell)}$ be the adjusted message LLRs on
side $s$ and edge type $t$.  For the message leaving a check along an edge of
type $t$, we draw the check type $c$ according to $\rho_{c\mid t}$.  Its
distribution is then
\begin{equation}
    C_{s,t}^{(\ell)}\mathrel{\overset{\mathrm d}{=}}
    2\operatorname{atanh}\!\left[
      \prod_{t'=1}^{T}
      \prod_{r=1}^{e_{c,t'}-\delta_{tt'}}
      \tanh\!\left(\frac{V_{s,t',r}^{(\ell)}}{2}\right)
    \right].
    \label{eq:de-met-check}
\end{equation}
Thus the check combines all socket classes present at the sampled node, with
only the outgoing socket omitted.

For a variable-to-check message of type $t$, we instead draw the variable
type $v$ according to $\lambda_{v\mid t}$.  Define its same-side extrinsic
sum and its full opposite-side sum by
\begin{align}
    A_{s,v,t}^{(\ell)}
    &=\sum_{t'=1}^{T}
      \sum_{r=1}^{d_{v,t'}-\delta_{tt'}}C_{s,t',r}^{(\ell)},
    \label{eq:de-met-same-sum}\\
    B_{s,v}^{(\ell)}
    &=\sum_{t'=1}^{T}
      \sum_{r=1}^{d_{v,t'}}C_{s,t',r}^{(\ell)}.
    \label{eq:de-met-full-sum}
\end{align}
The coupled variable updates are
\begin{align}
    V_{x,t}^{(\ell+1)}
    &\mathrel{\overset{\mathrm d}{=}}
      (-1)^x\gamma_x\!\left((-1)^zB_{z,v}^{(\ell)}\right)
      +A_{x,v,t}^{(\ell)},
    \label{eq:de-met-variable-x}\\
    V_{z,t}^{(\ell+1)}
    &\mathrel{\overset{\mathrm d}{=}}
      (-1)^z\gamma_z\!\left((-1)^xB_{x,v}^{(\ell)}\right)
      +A_{z,v,t}^{(\ell)}.
    \label{eq:de-met-variable-z}
\end{align}
The opposite-side belief is assembled from the same variable type $v$ as the
outgoing message.  This conditioning preserves correlations between a qubit
block and the socket classes to which it belongs.  Mixing the opposite-side
beliefs over all variable types would erase the lifted-product block
structure.  Initialization again sets the opposite-side sums to zero.  For
the error estimate, the variable type is drawn with its node probability
$\nu_v$, the full sums in Eq.~\eqref{eq:de-met-full-sum} are formed on both
sides, and the four scores in Eq.~\eqref{eq:de-posterior-scores} are compared.

For a lifted product constructed from an all-monomial
$m\times n$ matrix $A$ over the lift ring.  With $A^\dagger$ denoting the
antipode transpose, its two check matrices have the block form
\begin{align}
    H_X&=\left(A\otimes I_n\ \middle|\ I_m\otimes A^\dagger\right),
    \nonumber\\
    H_Z&=\left(I_n\otimes A\ \middle|\ A^\dagger\otimes I_m\right).
    \label{eq:de-lp-check-matrices}
\end{align}
There are $n^2$ qubits in the first block, $m^2$ in the second block, and
$mn$ checks on each CSS side.  Writing $N_0=n^2+m^2$, the corresponding
two-edge ensemble is
\begin{align}
    (\nu_1,\boldsymbol d_1)
      &=\left(\frac{n^2}{N_0},(m,0)\right),
    &
    (\nu_2,\boldsymbol d_2)
      &=\left(\frac{m^2}{N_0},(0,n)\right),
    \label{eq:de-lp-variable-types}\\
    (\mu,\boldsymbol e)
      &=\left(\frac{mn}{N_0},(n,m)\right).
    \label{eq:de-lp-check-type}
\end{align}

\clearpage
\newpage

\onecolumngrid

\section{Complete Catalogue of Halved Non-CSS Code Instances}
\label{sec:noncss-full-catalogue}

\begin{table*}[!h]
    \caption{Halved non-CSS pair-partition codes parameter frontiers for each $(J,L)$. We focus on the instances that achieve a given distance at the smallest block size we can find for each $(J,L)$. Distances written $\leq d$ are upper bounds from an
    explicit logical operator; unqualified distances are exact. Below distance
    eight the instances are small demonstrations rather than frontier points.
    The $J=3$ families are listed on the left and the $J=4$ families on the
    right. The fold column distinguishes the plain involution (\foldP, blue) from the reversing one (\foldR, red).}
    \label{tab:noncss-code-catalogue-full}
    \centering
    \scriptsize
    \setlength{\tabcolsep}{2.2pt}
    \renewcommand{\arraystretch}{0.88}
    \begin{ruledtabular}
    \begin{tabular}{@{}ccccc|ccccc@{}}
        $[[n,k,d]]$ & $(J,L)$ & $P$ & $k/n$ & Fold &
        $[[n,k,d]]$ & $(J,L)$ & $P$ & $k/n$ & Fold\\
        \colrule
        $[[88,24,8]]$ & $(3,8)$ & $22$ & $0.273$ & \foldR & $[[70,17,8]]$ & $(4,10)$ & $14$ & $0.243$ & \foldR\\
        $[[120,32,9]]$ & $(3,8)$ & $30$ & $0.267$ & \foldR & $[[90,21,11]]$ & $(4,10)$ & $18$ & $0.233$ & \foldR\\
        $[[124,33,11]]$ & $(3,8)$ & $31$ & $0.266$ & \foldP & $[[100,23,12]]$ & $(4,10)$ & $20$ & $0.230$ & \foldR\\
        $[[148,39,12]]$ & $(3,8)$ & $37$ & $0.264$ & \foldP & $[[110,25,13]]$ & $(4,10)$ & $22$ & $0.227$ & \foldR\\
        $[[236,61,14]]$ & $(3,8)$ & $59$ & $0.258$ & \foldP & $[[130,29,14]]$ & $(4,10)$ & $26$ & $0.223$ & \foldR\\
        $[[316,81,15]]$ & $(3,8)$ & $79$ & $0.256$ & \foldP & $[[140,31,15]]$ & $(4,10)$ & $28$ & $0.221$ & \foldR\\
        $[[388,99,18]]$ & $(3,8)$ & $97$ & $0.255$ & \foldP & $[[150,33,\leq 16]]$ & $(4,10)$ & $30$ & $0.220$ & \foldR\\
        $[[150,62,8]]$ & $(3,10)$ & $30$ & $0.413$ & \foldP & $[[160,35,\leq 17]]$ & $(4,10)$ & $32$ & $0.219$ & \foldR\\
        $[[190,78,9]]$ & $(3,10)$ & $38$ & $0.411$ & \foldP & $[[170,37,\leq 18]]$ & $(4,10)$ & $34$ & $0.218$ & \foldR\\
        $[[230,94,10]]$ & $(3,10)$ & $46$ & $0.409$ & \foldP & $[[185,40,\leq 19]]$ & $(4,10)$ & $37$ & $0.216$ & \foldR\\
        $[[310,126,11]]$ & $(3,10)$ & $62$ & $0.406$ & \foldP & $[[200,43,20]]$ & $(4,10)$ & $40$ & $0.215$ & \foldR\\
        $[[370,150,12]]$ & $(3,10)$ & $74$ & $0.405$ & \foldP & $[[215,46,\leq 21]]$ & $(4,10)$ & $43$ & $0.214$ & \foldR\\
        $[[450,182,13]]$ & $(3,10)$ & $90$ & $0.404$ & \foldR & $[[235,50,\leq 22]]$ & $(4,10)$ & $47$ & $0.213$ & \foldR\\
        $[[500,202,14]]$ & $(3,10)$ & $100$ & $0.404$ & \foldP & $[[245,52,\leq 23]]$ & $(4,10)$ & $49$ & $0.212$ & \foldR\\
        $[[264,134,8]]$ & $(3,12)$ & $44$ & $0.508$ & \foldR & $[[120,43,8]]$ & $(4,12)$ & $20$ & $0.358$ & \foldR\\
        $[[300,152,9]]$ & $(3,12)$ & $50$ & $0.507$ & \foldR & $[[132,47,10]]$ & $(4,12)$ & $22$ & $0.356$ & \foldR\\
        $[[450,227,10]]$ & $(3,12)$ & $75$ & $0.504$ & \foldR & $[[150,53,12]]$ & $(4,12)$ & $25$ & $0.353$ & \foldR\\
        $[[420,242,8]]$ & $(3,14)$ & $60$ & $0.576$ & \foldR & $[[168,59,13]]$ & $(4,12)$ & $28$ & $0.351$ & \foldR\\
        $[[476,274,9]]$ & $(3,14)$ & $68$ & $0.576$ & \foldR & $[[180,63,\leq 14]]$ & $(4,12)$ & $30$ & $0.350$ & \foldR\\
         &  &  &  &  & $[[186,65,15]]$ & $(4,12)$ & $31$ & $0.349$ & \foldP\\
         &  &  &  &  & $[[186,65,14]]$ & $(4,12)$ & $31$ & $0.349$ & \foldR\\
         &  &  &  &  & $[[240,83,\leq 17]]$ & $(4,12)$ & $40$ & $0.346$ & \foldR\\
         &  &  &  &  & $[[246,85,\leq 18]]$ & $(4,12)$ & $41$ & $0.346$ & \foldR\\
         &  &  &  &  & $[[264,91,\leq 19]]$ & $(4,12)$ & $44$ & $0.345$ & \foldR\\
         &  &  &  &  & $[[312,107,\leq 21]]$ & $(4,12)$ & $52$ & $0.343$ & \foldR\\
         &  &  &  &  & $[[182,81,10]]$ & $(4,14)$ & $26$ & $0.445$ & \foldR\\
         &  &  &  &  & $[[203,90,11]]$ & $(4,14)$ & $29$ & $0.443$ & \foldR\\
         &  &  &  &  & $[[224,99,13]]$ & $(4,14)$ & $32$ & $0.442$ & \foldR\\
         &  &  &  &  & $[[245,108,\leq 14]]$ & $(4,14)$ & $35$ & $0.441$ & \foldR\\
         &  &  &  &  & $[[266,117,\leq 15]]$ & $(4,14)$ & $38$ & $0.440$ & \foldR\\
         &  &  &  &  & $[[287,126,\leq 16]]$ & $(4,14)$ & $41$ & $0.439$ & \foldR\\
         &  &  &  &  & $[[336,147,\leq 18]]$ & $(4,14)$ & $48$ & $0.438$ & \foldR\\
         &  &  &  &  & $[[288,147,\leq 13]]$ & $(4,16)$ & $36$ & $0.510$ & \foldR\\
         &  &  &  &  & $[[304,155,\leq 14]]$ & $(4,16)$ & $38$ & $0.510$ & \foldR\\
         &  &  &  &  & $[[336,171,\leq 15]]$ & $(4,16)$ & $42$ & $0.509$ & \foldR\\
    \end{tabular}
    \end{ruledtabular}
\end{table*}

\newpage

\section{Complete Catalogue of CPM-PP Code Instances}
\label{sec:cpmpp-full-catalogue}

\begin{table}[!h]
    \caption{CPM-based CSS pair-partition codes parameter frontier for each $(J,L)$. We focus on the instances that achieve a given distance at the smallest block size we can find for each $(J,L)$.}
    \label{tab:cpmpp-code-catalogue-full}
    \centering
    \scriptsize
    \setlength{\tabcolsep}{2.2pt}
    \renewcommand{\arraystretch}{0.88}
    \begin{ruledtabular}
    \begin{tabular}{@{}ccccc@{}}
        $[[n,k,d]]$ & girth & $(J,L)$ & $P$ & Rate $k/n$\\
        \colrule
        $[[176,48,10]]$ & $6$ & $(3,8)$ & $22$ & $0.273$\\
        $[[232,62,12]]$ & $6$ & $(3,8)$ & $29$ & $0.267$\\
        $[[360,94,14]]$ & $6$ & $(3,8)$ & $45$ & $0.261$\\
        $[[472,122,16]]$ & $8$ & $(3,8)$ & $59$ & $0.258$\\
        $[[584,150,18]]$ & $8$ & $(3,8)$ & $73$ & $0.257$\\
        $[[776,198,20]]$ & $8$ & $(3,8)$ & $97$ & $0.255$\\
        $[[1112,282,20]]$ & $8$ & $(3,8)$ & $139$ & $0.254$\\
        $[[1336,338,22]]$ & $8$ & $(3,8)$ & $167$ & $0.253$\\
        $[[1784,450,24]]$ & $8$ & $(3,8)$ & $223$ & $0.252$\\
        $[[470,192,12]]$ & $6$ & $(3,10)$ & $47$ & $0.409$\\
        $[[530,216,12]]$ & $6$ & $(3,10)$ & $53$ & $0.408$\\
        $[[710,288,14]]$ & $6$ & $(3,10)$ & $71$ & $0.406$\\
        $[[970,392,16]]$ & $6$ & $(3,10)$ & $97$ & $0.404$\\
        $[[1390,560,18]]$ & $6$ & $(3,10)$ & $139$ & $0.403$\\
        $[[1630,656,20]]$ & $8$ & $(3,10)$ & $163$ & $0.402$\\
        $[[2110,848,22]]$ & $6$ & $(3,10)$ & $211$ & $0.402$\\
        $[[2230,896,24]]$ & $8$ & $(3,10)$ & $223$ & $0.402$\\
        $[[1356,682,14]]$ & $6$ & $(3,12)$ & $113$ & $0.503$\\
        $[[2004,1006,16]]$ & $6$ & $(3,12)$ & $167$ & $0.502$\\
        $[[2676,1342,18]]$ & $6$ & $(3,12)$ & $223$ & $0.501$\\
        $[[1414,812,12]]$ & $6$ & $(3,14)$ & $101$ & $0.574$\\
        $[[3122,1788,16]]$ & $6$ & $(3,14)$ & $223$ & $0.573$\\
        \colrule
        $[[180,42,12]]$ & $6$ & $(4,10)$ & $18$ & $0.233$\\
        $[[200,46,14]]$ & $6$ & $(4,10)$ & $20$ & $0.230$\\
        $[[220,50,16]]$ & $6$ & $(4,10)$ & $22$ & $0.227$\\
        $[[290,64,18]]$ & $6$ & $(4,10)$ & $29$ & $0.221$\\
        $[[320,70,20]]$ & $6$ & $(4,10)$ & $32$ & $0.219$\\
        $[[370,80,22]]$ & $6$ & $(4,10)$ & $37$ & $0.216$\\
        $[[276,98,14]]$ & $6$ & $(4,12)$ & $23$ & $0.355$\\
        $[[372,130,16]]$ & $6$ & $(4,12)$ & $31$ & $0.349$\\
        $[[444,154,18]]$ & $6$ & $(4,12)$ & $37$ & $0.347$\\
        $[[492,170,20]]$ & $6$ & $(4,12)$ & $41$ & $0.346$\\
        $[[516,178,20]]$ & $6$ & $(4,12)$ & $43$ & $0.345$\\
        $[[708,242,22]]$ & $6$ & $(4,12)$ & $59$ & $0.342$\\
        $[[876,298,22]]$ & $6$ & $(4,12)$ & $73$ & $0.340$\\
        $[[322,144,12]]$ & $4$ & $(4,14)$ & $23$ & $0.447$\\
        $[[406,180,14]]$ & $4$ & $(4,14)$ & $29$ & $0.443$\\
        $[[518,228,16]]$ & $6$ & $(4,14)$ & $37$ & $0.440$\\
        $[[574,252,18]]$ & $6$ & $(4,14)$ & $41$ & $0.439$\\
        $[[560,286,12]]$ & $4$ & $(4,16)$ & $35$ & $0.511$\\
        $[[592,302,14]]$ & $4$ & $(4,16)$ & $37$ & $0.510$\\
        $[[848,430,18]]$ & $6$ & $(4,16)$ & $53$ & $0.507$\\
        $[[944,478,20]]$ & $6$ & $(4,16)$ & $59$ & $0.506$\\
        \colrule
        $[[110,8,12]]$ & $6$ & $(5,10)$ & $11$ & $0.073$\\
        $[[170,8,20]]$ & $6$ & $(5,10)$ & $17$ & $0.047$\\
        $[[190,8,18]]$ & $6$ & $(5,10)$ & $19$ & $0.042$\\
    \end{tabular}
    \end{ruledtabular}
\end{table}

\end{document}